\documentclass[acmsmall,screen]{acmart}

\usepackage{amssymb}
\usepackage{amsmath}
\usepackage{amsfonts}
\usepackage{caption}
\usepackage{subcaption}
\usepackage{xspace}
\usepackage{mathtools}
\usepackage{mathpartir}
\usepackage{ifpdf}
\usepackage{graphicx}
\usepackage{stmaryrd}
\usepackage{amsthm}
\usepackage{listings}          
\usepackage{wrapfig}
\usepackage{textcomp}
\usepackage{tabularx}
\usepackage{color}
\usepackage{url}
\usepackage{tikz}
\usepackage{multirow,array}
\usepackage[utf8]{inputenc}
\usepackage[T1]{fontenc}
\usepackage{microtype}

\newenvironment{nop}{}{}
\newenvironment{smathpar}{
\begin{nop}\small\begin{mathpar}}{
\end{mathpar}\end{nop}\ignorespacesafterend}

\definecolor{Bittersweet}{rgb}{1.0, 0.44, 0.37}
\definecolor{MidnightBlue}{rgb}{0.0, 0.2, 0.4}
\definecolor{BrightBlue}{rgb}{0.0, 0.2, 0.7}

\newcommand{\lsttxnimp}{\lstset{
      language=c,
      basicstyle=\ttfamily\small,
      flexiblecolumns=false,
			tabsize=2,
      escapechar=',                        
      keywordstyle=\color{Bittersweet}\bfseries,
      commentstyle=\color{blue}\itshape,
      stringstyle=\color{MidnightBlue},
      morekeywords={transaction,txn,cobegin,from,to,atomic},
			classoffset=1,
			upquote=true,
			keywordstyle=\color{Fuchsia}\bfseries,
			classoffset=0,
			mathescape=true
    }}
\lstnewenvironment{txnimpcode}
    { 
      \lsttxnimp
      \lstset{}%
      \csname lst@setfirstlabel\endcsname}
    { 
      \csname lst@savefirstlabel\endcsname}

\newcommand{\lstml}{
\lstset{ %
language=ML, 
basicstyle=\footnotesize\ttfamily,       
keywordstyle=\color{Bittersweet},
numberstyle=\tiny,      
stepnumber=1,                   
numbersep=5pt,                  
showspaces=false,               
showstringspaces=false,         
showtabs=false,                 
tabsize=2,                      
captionpos=b,                   
breaklines=true,                
breakatwhitespace=false,        
commentstyle=\itshape\color{BrightBlue},
mathescape=true,
morekeywords={module, begin, match, when, @@deriving, not, : , txn_do, do, SQL/\\}
}}
\lstnewenvironment{ocaml}
    { 
			\lstml
      \lstset{}%
      \csname lst@setfirstlabel\endcsname}
    { 
      \csname lst@savefirstlabel\endcsname}

\newcommand{\C}[1]{\code{#1}}

\newcommand*{\rom}[1]{\expandafter\romannumeral #1}

\newcommand{\code}[1]{{\tt #1}}
\newcommand{\spc}[0]{\quad}
\newcommand{\ALT}{~\mid~}
\newcommand{\llangle}{\langle}
\newcommand{\rrangle}{\rangle}

\newcommand{\conj}{~\wedge~}
\newcommand{\disj}{~\vee~}
\newcommand{\rulelabel}[1]{\textrm{\sc {#1}}}

\newcommand{\RULE}[2]{\frac{\begin{array}{c}#1\end{array}}
                           {\begin{array}{c}#2\end{array}}}
\newcommand{\txnimp}{\mbox{${\mathcal T}$}}

\newcommand{\cskip}{\C{SKIP}}
\newcommand{\ctxnr}[3]{{\sf txn}_{#1}\llangle #2 \rrangle\{#3\}}
\newcommand{\ctxn}[3]{\C{TXN}_{#1}\llangle #2 \rrangle\{#3\}}

\newcommand{\stepsto}{\longrightarrow}
\newcommand{\rstepsto}{\longrightarrow_{R}}
\newcommand{\stepssto}[1]{\longrightarrow^{*}_{R}}
\newcommand{\rstepssto}[1]{\longrightarrow^{*}_{R}}

\newcommand{\tbox}[1]{\lbrack #1 \rbrack}
\newcommand{\hoare}[3]{\{#1\}\;\tbox{#2}_i\;\{#3\}}
\newcommand{\defeq}[0]{\overset { \mathit{def} }{ = } }
\newcommand{\rg}[3]{\{#1\}\;#2\;\{#3\}}

\newcommand{\I}{\mathbb{I}}
\newcommand{\R}{\mathbb{R}}
\newcommand{\F}{{\sf F}}
\newcommand{\T}{{\sf T}}

\newcommand{\eval}{\textsf{eval}}
\newcommand{\dom}{\textsf{dom}}

\newcommand{\stable}{\mathtt{stable}}
\newcommand{\iso}[1]{\emph{#1}}

\newcommand{\eg}{\emph{e.g.,}\xspace}

\newcommand{\ectx}{\mathcal{E}}

\newcommand{\Prop}{\mathbb{P}}
\newcommand{\Pow}[1]{\mathcal{P}\left(#1\right)}
\newcommand{\bind}{\gg=}
\newcommand{\ite}[3]{\C{IF}\;#1\;\C{THEN}\;#2\;\C{ELSE}\;#3}
\newcommand{\lete}[3]{\C{LET}\;#1=#2\;\C{IN}\;#3}
\newcommand{\foreache}[2]{\texttt{FOREACH}\;#1\;\texttt{DO}\;#2}
\newcommand{\foreachr}[3]{{\sf foreach}\llangle #1 \rrangle\;#2\;{\sf do}\;#3}
\newcommand{\inserte}[1]{\texttt{INSERT}\;#1}
\newcommand{\selecte}[1]{\texttt{SELECT}\;#1}
\newcommand{\deletee}[1]{\texttt{DELETE}\;#1}
\newcommand{\updatee}[2]{\texttt{UPDATE}\;#1\;#2}
\newcommand{\stl}{\delta}
\newcommand{\stg}{\Delta}
\newcommand{\rec}{r}
\newcommand{\idf}{\texttt{id}}
\newcommand{\delf}{\texttt{del}}
\newcommand{\txnf}{\texttt{txn}}

\newcommand{\elabsto}{\Longrightarrow_{\langle i,\R,I \rangle}}
\newcommand{\with}{~\C{with}~}
\newcommand{\itel}[3]{{\sf if}\;#1\;{\sf then}\;#2\;{\sf else}\;#3}
\newcommand{\itec}[3]{\C{if}\;#1\;\C{then}\;#2\;\C{else}\;#3}
\newcommand{\stabilize}[1]{\llfloor #1 \rrfloor_{\langle \R,I \rangle}}

\newcommand{\mssemof}[2]{\llbracket #2 \rrbracket_{\langle #1 \rangle}}
\newcommand{\existsl}{{\sf exists}}
\newcommand{\fresh}{{\sf fresh}}
\newcommand{\SL}{\mathcal{S}}
\newcommand{\thetool}{{\sc ACIDifier}\xspace}
\newcommand{\nubar}{\overline{\nu}}
\newcommand{\vbar}{\overline{\upsilon}}
\renewcommand{\v}{\upsilon}
\newcommand{\G}{{\sf G}}

\newcommand{\Fx}{{\sf F}_{\sf ctxt}}
\newcommand{\Fempty}{{\sf F}_{\emptyset}}
\newcommand{\inctxt}[2]{#1[#2]}

\setcopyright{rightsretained} 
\acmJournal{PACMPL}
\acmYear{2018} 
\acmVolume{2} 
\acmNumber{POPL} 
\acmArticle{27}
\acmMonth{1} 
\acmPrice{}
\acmDOI{10.1145/3158115}

\begin{document}

\setlength{\pdfpageheight}{\paperheight}
\setlength{\pdfpagewidth}{\paperwidth}

\makeatletter\if@ACM@journal\makeatother
\startPage{1}
\else\makeatother
\fi

\bibliographystyle{ACM-Reference-Format}
\citestyle{acmauthoryear}   

\title{Alone Together:
  Compositional Reasoning and Inference for Weak Isolation}


\author{Gowtham Kaki}
\affiliation{
  \institution{Purdue University}            
  \country{USA}
}
\email{gkaki@purdue.edu}          
\author{Kartik Nagar}
\affiliation{
  \institution{Purdue University}            
  \country{USA}
}
\email{nagark@purdue.edu}          
\author{Mahsa Najafzadeh}
\affiliation{
  \institution{Purdue University}            
  \country{USA}
}
\email{mnajafza@purdue.edu}          
\author{Suresh Jagannathan}
\affiliation{
  \institution{Purdue University}            
  \country{USA}
}
\email{suresh@cs.purdue.edu}          


\begin{abstract}

  Serializability is a well-understood correctness criterion that
  simplifies reasoning about the behavior of concurrent transactions
  by ensuring they are \emph{isolated} from each other while they
  execute.  However, enforcing serializable isolation comes at a steep
  cost in performance because it necessarily restricts opportunities
  to exploit concurrency even when such opportunities would not
  violate application-specific invariants. As a result, database
  systems in practice support, and often encourage, developers to
  implement transactions using weaker alternatives. These alternatives
  break the strong isolation guarantees offered by serializablity to
  permit greater concurrency. Unfortunately, the semantics of weak
  isolation is poorly understood, and usually explained only
  informally in terms of low-level implementation
  artifacts. Consequently, verifying high-level correctness properties
  in such environments remains a challenging problem.

  To address this issue, we present a novel program logic that enables
  compositional reasoning about the behavior of concurrently executing
  weakly-isolated transactions.  Recognizing that the proof burden
  necessary to use this logic may dissuade application developers, we
  also describe an inference procedure based on this foundation that
  ascertains the weakest isolation level that still guarantees the
  safety of high-level consistency invariants associated with such
  transactions.  The key to effective inference is the observation
  that weakly-isolated transactions can be viewed as functional
  (monadic) computations over an abstract database state, allowing us
  to treat their operations as state transformers over the database.
  This interpretation enables automated verification using
  off-the-shelf SMT solvers.
  
  Our development is parametric over a transaction's specific
  isolation semantics, allowing it to be applicable over a range of
  weak isolation mechanisms.  Case studies and experiments on
  real-world applications (written in an embedded DSL in OCaml)
  demonstrate the utility of our approach, and provide strong evidence
  that automated verification of weakly-isolated transactions can be
  placed on the same formal footing as their strongly-isolated
  serializable counterparts.

\end{abstract}



\maketitle

\section{Introduction}

Database transactions allow users to group operations on multiple
objects into a single logical unit, equipped with a set of four key
properties - atomicity, consistency, isolation, and durability (ACID).
Concurrency control mechanisms provide specific instantiations of
these properties to yield different ACID variants that characterize
how and when the effects of concurrently executing transactions become
visible to one another.  \emph{Serializability} is a particularly
well-studied instantiation that imposes strong atomicity and isolation
constraints on transaction execution, ensuring that any permissible
concurrent schedule yields results equivalent to a serial one in which
there is no interleaving of actions from different transactions.

The guarantees provided by serializability do not come for free,
however - pessimistic concurrency control methods require databases to
use expensive mechanisms such as two-phase locking that incur overhead
to deal with deadlocks, rollbacks, and
re-execution~\cite{twopl,ullmanbook}.  Similar criticisms apply to
optimistic multi-version concurrency control methods that must deal
with timestamp and version management~\cite{BG81}.  These issues are
exacerbated when the database is replicated, requiring
additional coordination
mechanisms~\cite{cap,sernotavlbl,bailishat,bernsigmod13}.

Because serializable transactions favor correctness over performance,
there has been long-standing interest~\cite{gray1976} in the database
community to consider weaker variants that try to recover performance,
even at the expense of simplicity and ease of reasoning.  These
instantiations permit a transaction to witness various effects of
newly committed, or even concurrently running, transactions while it
executes, thus weakening serializability's strong isolation
guarantees.  The ANSI SQL 92 standard defines three such weak
isolation levels which are now implemented in many relational and
NoSQL databases. Not surprisingly, weakly-isolated transactions have
been found to significantly outperform serializable transactions on
benchmark suites, both on single-node databases and multi-node
replicated stores~\cite{dbtuningbook,bailishat,bailisvldb}, leading to
their overwhelming adoption. A 2013 study~\cite{bailishotos} of 18
popular ACID and ``NewSQL'' databases found that only three of them
offer serializability by default, and half, including Oracle 11g, do
not offer it at all.  A 2015 study~\cite{bailisferal} of a large
corpus of database applications finds no evidence that applications
manifestly change the default isolation level offered by the
database. Taken together, these studies make clear that
weakly-isolated transactions are quite prevalent in practice, and
serializable transactions are often eschewed.

Unfortunately, weak isolation admits behaviors that are difficult to
comprehend~\cite{berenson}. To quantify weak isolation anomalies,
Fekete \emph{et al.}~\cite{feketevldb09} devised and experimented with
a microbenchmark suite that executes transactions under \emph{Read
Committed} weak isolation level - default level for 8 of the 18
databases studied in~\cite{bailishotos}, and found that 25 out of
every 1000 rows in the database violate at least one integrity
constraint. Bailis \emph{et al.}~\cite{bailisferal} rely on Rails'
\emph{uniqueness validation} to maintain uniqueness of records while
serving Linkbench's~\cite{linkbench} insertion workload (6400 records
distributed over 1000 keys; 64 concurrent clients), and report
discovering more than 10 duplicate records.  Rails relies on database
transactions to validate uniqueness during insertions, which is
sensible if transactions are serializable, but incorrect under the
weak isolation level used in the experiments. The same study has found
that 13\% of all invariants among 67 open source Ruby-on-Rails
applications are at risk of being violated due to weak isolation.
Indeed, incidents of safety violations due to executing applications
in a weakly-isolated environment have been reported on web services in
production~\cite{starbucksbug, scimedbug}, including in
safety-critical applications such as bitcoin
exchanges~\cite{poloniexbug, bitcoinbug}. While enforcing
serializability for all transactions would be sufficient to avoid
these errors and anomalies, it would likely be an overly conservative
strategy; indeed, 75\% of the invariants studied in~\cite{bailisferal}
were shown to be preserved under some form of weak isolation.  When to
use weak isolation, and in what form, is therefore a prominent
question facing all database programmers.\footnote{This position has
been echoed by database researchers who lament the lack of a better
understanding of this problem; see e.g., \url{
  http://www.bailis.org/blog/understanding-weak-isolation-is-a-serious-problem}.}

A major problem with weak isolation as currently specified is that its
semantics in the context of user programs is not easily
understood. The original proposal~\cite{gray1976} defines multiple
``degrees'' of weak isolation in terms of implementation details such
as the nature and duration of locks held in each case. The ANSI SQL 92
standard defines four levels of isolation (including serializability)
in terms of various undesirable \emph{phenomena} (\eg \emph{dirty
  reads} - reading data written by an uncommitted transaction) each is
required to prevent. While this is an improvement, this style of
definition still requires programmers to be prescient about the
possible ways various undesirable phenomena might manifest in their
applications, and in each case determine if the phenomenon can be
allowed without violating application invariants. This is
understandably hard, especially in the absence of any formal
underpinning to define weak isolation semantics.  Adya~\cite{adyaphd}
presents the first formal definitions of some well-known isolation
levels in the context of a sequentially consistent (SC) database.
However, there has been little progress relating Adya's system model
to a formal operational semantics or a proof system that can
facilitate rigorous correctness arguments.  Consequently, reasoning
about weak isolation remains an error prone endeavor, with major
database vendors~\cite{postgresiso, mysqliso, oracleiso} continuing to
document their isolation levels primarily in terms of the undesirable
phenomena a particular isolation level may induce, placing the burden
on the programmer to determine application correctness.

Recent results on reasoning about application invariants in the
presence of weak consistency~\cite{burckhardt14, redblueosdi,
  redblueatc, ecinec, gotsmanpopl16} address broadly related concerns.
Weak consistency is a phenomenon that manifests on replicated data
stores, where atomic operations are concurrently executed against
different replicas, resulting in an execution order inconsistent with
any sequential order. In contrast, weak isolation is a property of
concurrent transactions interfering with one another, resulting in an
execution order that is not serializable. Unlike weak consistency,
weak isolation can manifest even in an unreplicated setting, as
evident from the support for weakly-isolated transactions on
conventional (unreplicated) databases as mentioned above.

In this paper, we propose a program logic for weakly-isolated
transactions along with automated verification support to allow
developers to verify the soundness of their applications, without
having to resort to low-level operational reasoning as they are forced
to do currently.  We develop a set of syntax-directed compositional
proof rules that enable the construction of correctness proofs for
transactional programs in the presence of a weakly-isolated
concurrency control mechanism.  Realizing that the proof burden
imposed by these rules may discourage applications programmers from
using them, we also present an inference procedure that automatically
verifies the weakest isolation level for a transaction while ensuring
its invariants are maintained.  The key to inference is a novel
formulation of database state (represented as sets of tuples) as a
monad, and in which database computations are interpreted as state
transformers over these sets.  This interpretation leads to an
encoding of database computations amenable for verification by
off-the-shelf SMT solvers.  The paper makes the following
contributions:
\begin{enumerate}
  \item We analyze properties of weak isolation in the context of a
    DSL embedded in OCaml that treats SQL-based relational database
    operations (e.g., inserts, selects, deletes, updates, etc.) as
    computations over an abstract database state.
  \item We develop an operational semantics and a compositional
    rely-guarantee style proof system for this language capable of
    relating high-level application invariants to database state,
    parameterized by a weak isolation semantics that selectively
    exposes the visibility of these operations to other transactions.
  \item We devise an inference algorithm capable of discovering the
    weakest isolation level that is sound with respect to a
    transaction's high-level consistency requirements. The algorithm
    interprets database operations as state transformers expressed in
    a language amenable for translation into a decidable fragment of
    first-order logic, and is thus suitable for automated verification
    using off-the-shelf SMT solvers.
  \item We present details of an implementation along with an
    evaluation study on real database benchmarks that justify our
    approach, and demonstrate the utility of our inference mechanism.
\end{enumerate}
\noindent Our results provide the first formalization of
weakly-isolated transactions, along with an expressive and
compositional proof automation framework capable of verifying the
safety of high-level consistency conditions attached to these
transactions.  Collectively, these contributions allow weakly-isolated
transactions to enjoy the same rigorous reasoning capabilities as
their strongly-isolated (serializable) counterparts.

The remainder of the paper is organized as follows. The next section
provides motivation and background on serializable and weakly-isolated
transactions. \S\ref{sec:opsem} presents an operational semantics for
a core language that supports weakly-isolated transactions,
parameterized over different isolation notions. \S\ref{sec:reasoning}
formalizes the proof system that we use to reason about program
invariants, and establishes the soundness of these rules with respect
to the semantics. \S\ref{sec:inference} describes the inference
algorithm, and the state transformer encoding.  We describe our
implementation in \S\ref{sec:implementation}, and provide case studies
and benchmark results in \S\ref{sec:case-studies}.  Related work is
given in \S\ref{sec:relatedwork}, and \S\ref{sec:conclusions} concludes.

\section{Motivation}
\label{sec:motivation}

We present our ideas in the context of a DSL embedded in OCaml that
manages an abstract database state that can be manipulated via a
well-defined \C{SQL} interface. Arbitrary database computations can be
built around this interface, which can then be run as transactions
using the \C{atomically\_do} combinator provided by the DSL.

\begin{figure}[!t]
\centering
\begin{ocaml}
let new_order (d_id, c_id, item_reqs) = atomically_do @@ fun () ->
  let dist = SQL.select1 District (fun d -> d.d_id = d_id) in
  let o_id = dist.d_next_o_id in
  begin
    SQL.update(* UPDATE *) District 
              (* SET *)(fun d -> {d with d_next_o_id = d.d_next_o_id + 1})
              (* WHERE *)(fun d -> d.d_id = d_id );
    SQL.insert(* INSERT INTO *) Order (* VALUES *){o_id=o_id;  
              o_d_id=d_id; o_c_id=c_id; o_ol_cnt=S.size item_reqs; };
    foreach item_reqs @@ fun item_req ->
      let stk = SQL.select1(* SELECT * FROM *) Stock 
                (* WHERE *)(fun s -> s.s_i_id = item_req.ol_i_id &&
                                     s.s_d_id = d_id)(* LIMIT 1 *) in
      let s_qty' = if stk.s_qty >= item_req.ol_qty + 10 
                  then stk.s_qty - item_req.ol_qty 
                  else stk.s_qty - item_req.ol_qty + 91 in
      SQL.update Stock (fun s -> {s with s_qty = s_qty'}) 
                       (fun s -> s.s_i_id = item_req.ol_i_id);
      SQL.insert Order_line {ol_o_id=o_id; ol_d_id=d_id; 
                             ol_i_id=item_req.ol_i_id; ol_qty=item_req.ol_qty}
  end
 
\end{ocaml}
\caption{TPC-C \C{new\_order} transaction}
\label{fig:new_order_code}
\vspace*{-10pt}
\end{figure}

Fig.~\ref{fig:new_order_code} shows a simplified version of the TPC-C
\C{new\_order} transaction written in this language. TPC-C is a
widely-used and well-studied Online Transaction Processing (OLTP)
benchmark that models an order-processing system for a wholesale parts
supply business. The business logic is captured in 5 database
transactions that operate on 9 tables; \C{new\_order} is one such
transaction that uses \C{District}, \C{Order}, \C{New\_order},
\C{Stock}, and \C{Order\_line} tables. The transaction acts on the
behalf of a customer, whose id is \C{c\_id}, to place a new order for
a given set of items (\C{item\_reqs}), to be served by a warehouse
under the district identified by \C{d\_id}.  Fig.~\ref{fig:schema}
illustrates the relationship among these different tables.

The transaction manages order placement by invoking appropriate SQL
functionality, captured by various calls to functions defined by the
\C{SQL} module. All \C{SQL} operators supported by the module take a
table name (a nullary constructor) as their first argument. The
higher-order \C{SQL.select1} function accepts a boolean function that
describes the selection criteria, and returns any record that meets
the criteria (it models the SQL query \C{SELECT \ldots\xspace LIMIT
  1}). \C{SQL.update} also accepts a boolean function (its 3$^{rd}$
argument) to select the records to be updated. Its 2$^{nd}$ argument
is a function that maps each selected record to a new (updated)
record. \C{SQL.insert} inserts a given record into the specified table
in the database.

The \C{new\_order} transaction inserts a new \C{Order} record, whose
id is the sequence number of the next order under the given district
(\C{d\_id}). The sequence number is stored in the corresponding
\C{District} record, and updated each time a new order is added to the
system. Since each order may request multiple items (\C{item\_reqs}),
an \C{Order\_line} record is created for each requested item to relate
the order with the item. Each item has a corresponding record in the
\C{Stock} table, which keeps track of the quantity of the item left in
stock (\C{s\_qty}). The quantity is updated by the transaction to
reflect the processing of new orders (if the stock quantity falls below
10, it is automatically replenished by 91).

\begin{figure}[!t]
  \centering
	\begin{subfigure}{0.4\textwidth}
	  \includegraphics[width=0.9\textwidth]{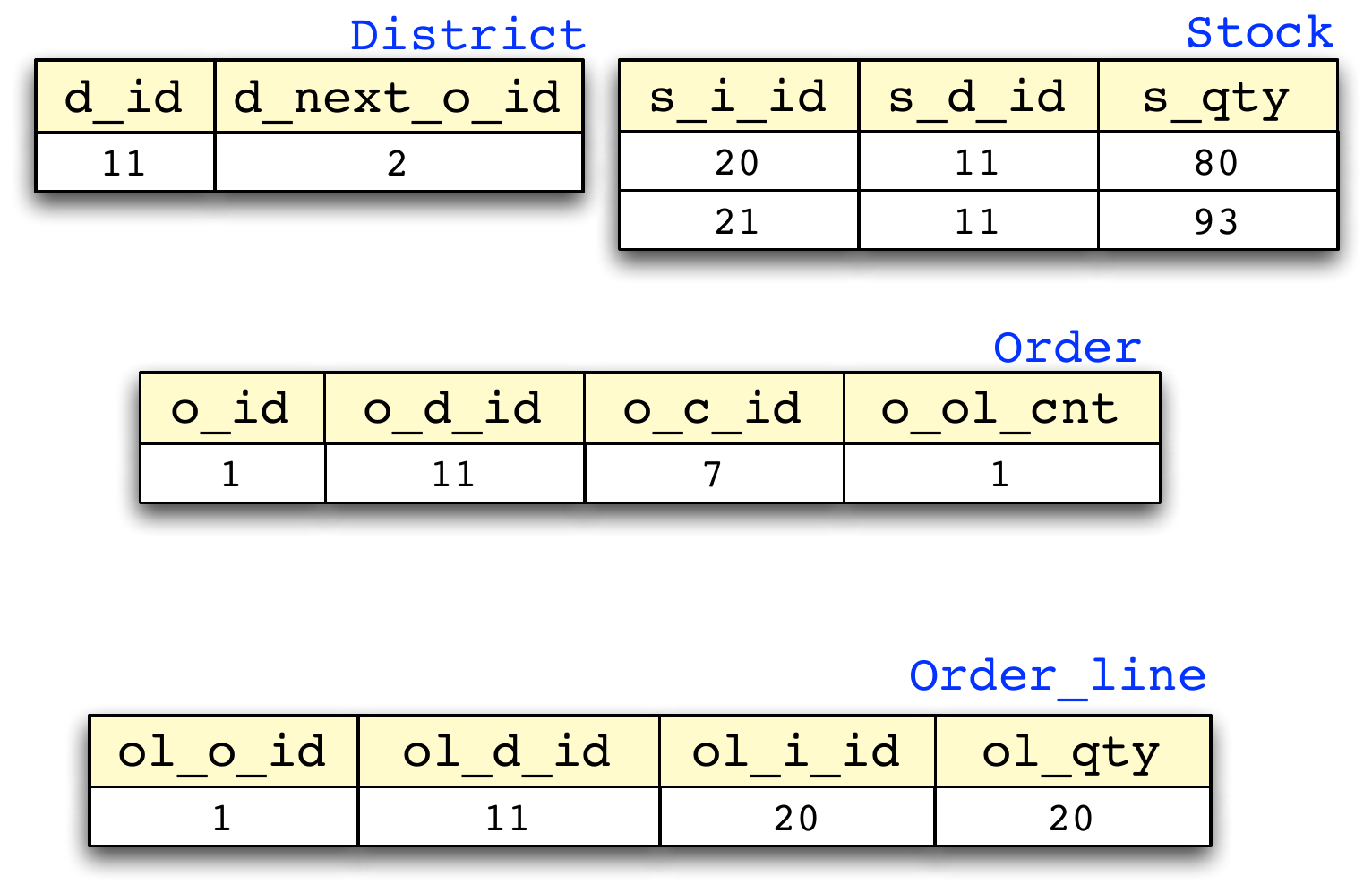}
    \caption{A valid TPC-C database. The only existing order belongs
      to the district with \C{d\_id}=11. Its id (\C{o\_id}) is one
      less than the district's \C{d\_next\_o\_id}, and its order
      count (\C{o\_ol\_cnt}) is equal to the number of order line
      records whose \C{ol\_o\_id} is equal to the order's id.  }
		\label{fig:tpcc_db1}
	\end{subfigure}
  \hspace*{.2in}
	\begin{subfigure}{0.52\textwidth}
		\includegraphics[width=0.9\textwidth]{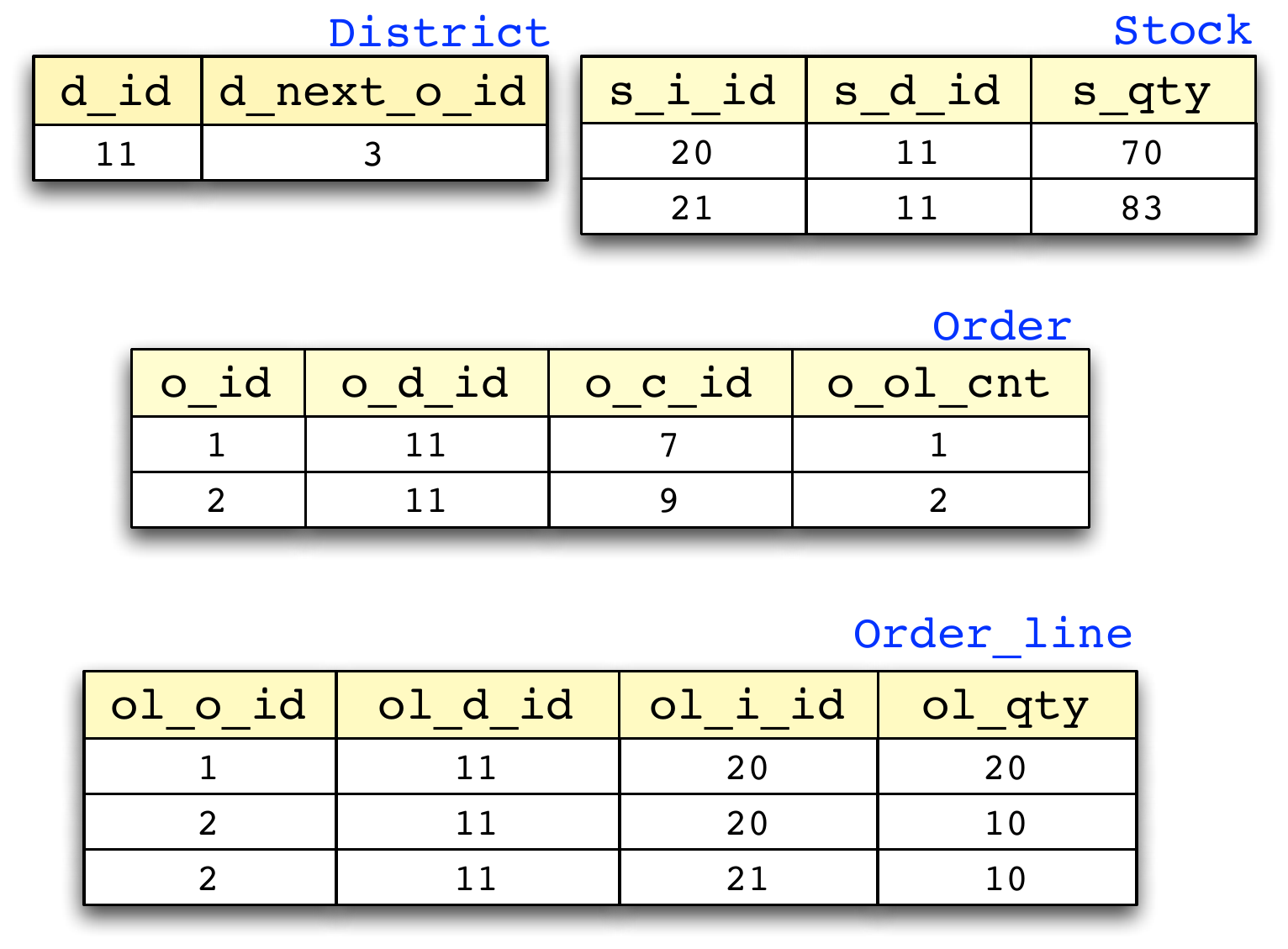}
    \caption{The database in Fig.~\ref{fig:tpcc_db1} after correctly
      executing a \C{new\_order} transaction. A new order record is
      added whose \C{o\_id} is equal to the \C{d\_next\_o\_id} from
      Fig.~\ref{fig:tpcc_db1}. The district's \C{d\_next\_o\_id} is
      incremented. The order's \C{o\_ol\_cnt} is 2, reflecting the
      actual number of order line records whose \C{ol\_o\_id} is equal
      to the order's id (2).}
		\label{fig:tpcc_db2}
   	\end{subfigure}
\caption{Database schema of TPC-C's order management system.
  The naming
  convention indicates primary keys and foreign keys. For e.g.,
  \C{ol\_id} is the primary key column of the order line table,
  whereas \C{ol\_o\_id} is a foreign key that refers to the \C{o\_id}
  column of the order table.}
\label{fig:schema}
\end{figure}

TPC-C defines multiple invariants, called \emph{consistency
  conditions}, over the state of the application in the database. One
such consistency condition is the requirement that for a given order
\C{o}, the \emph{order-line-count} field (\C{o.o\_ol\_cnt}) should
reflect the number of order lines under the order; this is the number
of \C{Order\_line} records whose \C{ol\_o\_id} field is the same as
\C{o.o\_id}.  In a sequential execution, it is easy to see how this
condition is preserved.  A new \C{Order} record is added with its
\C{o\_id} distinct from existing order ids, and its \C{o\_ol\_cnt} is
set to be equal to the size of the \C{item\_reqs} set. The \C{foreach}
loop runs once for each \C{item\_req}, adding a new \C{Order\_line}
record for each requested item, with its \C{ol\_o\_id} field set to
\C{o\_id}. Thus, at the end of the loop, the number of \C{Order\_line}
records in the database (i.e., the number of records whose
\C{ol\_o\_id} field is equal to \C{o\_id}) is guaranteed to be equal
to the size of the \C{item\_reqs} set, which in turn is equal to the
\C{Order} record's \C{o\_ol\_cnt} field; these constraints ensure that
the transaction's consistency condition is preserved.

Because the aforementioned reasoning is reasonably simple to perform
manually, verifying the soundness of TPC-C's consistency conditions
would appear to be feasible.  Serializability aids the tractability of
verification by preventing any interference among concurrently
executing transactions while the \C{new\_order} transaction executes,
essentially yielding serial behaviors.  
\begin{wrapfigure}{l}{.4\textwidth}
\includegraphics[scale=0.45]{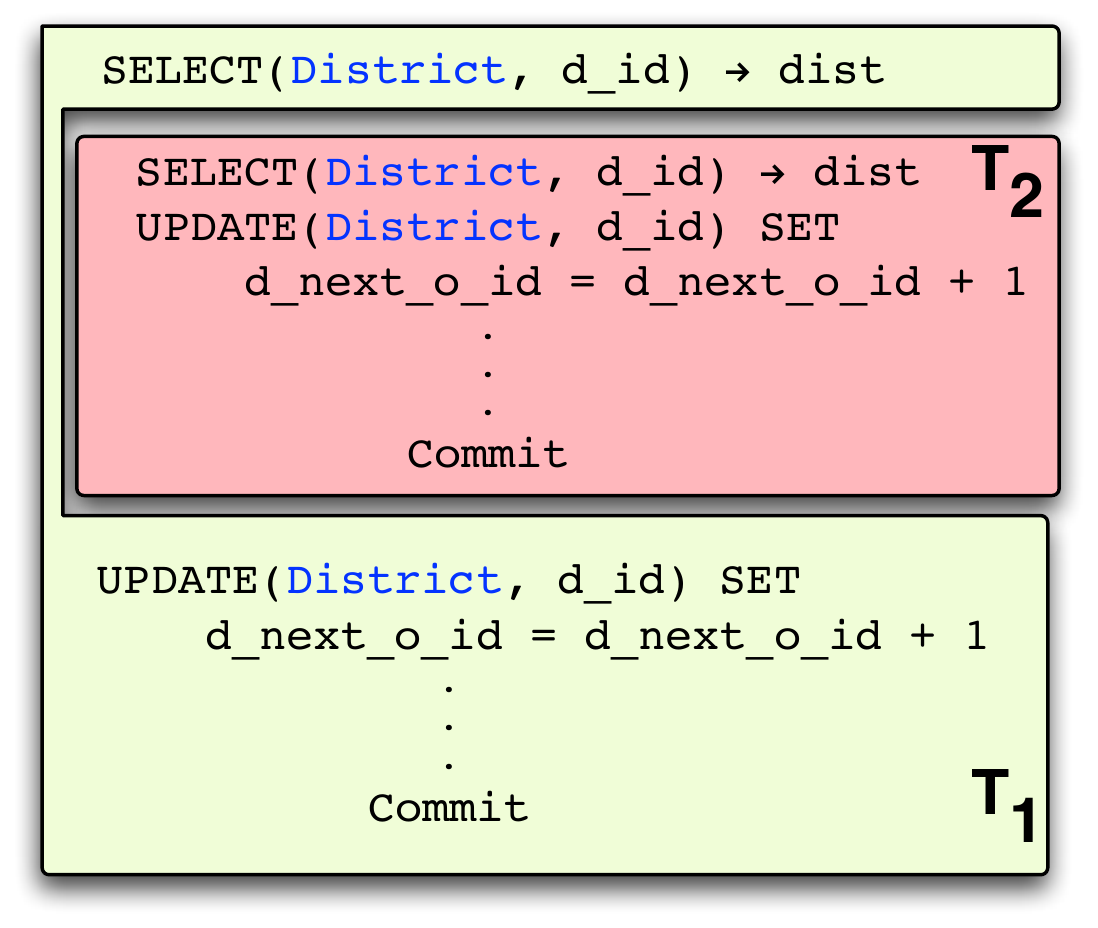}
\caption{\small An RC execution involving two instances ($T_1$ and
  $T_2$) of the \C{new\_order} transaction depicted in
  Fig.~\ref{fig:new_order_code}. 
  Both instances read the \C{d\_id} \C{District} record concurrently,
  because neither transaction is committed when the reads are
  executed.  The subsequent operations are effectively sequentialized,
  since $T_2$ commits before $T_1$. Nonetheless, both transactions read the same value for
  \C{d\_next\_o\_id} resulting in them  adding \C{Order} records
  with the same ids, which in turn triggers a violation of TPC-C's
  consistency condition.}
\label{fig:new_order_exec}
\vspace*{-10pt}
\end{wrapfigure}
Under weak isolation\footnote{Weak isolation does not violate
  atomicity as long as the witnessed effects are those of committed 
transactions}, however, interferences of various kinds are permitted,
leading to executions superficially similar to executions permitted by
concurrent (racy) programs~\cite{GHE15,HPQ+15}.  To illustrate,
consider the behavior of the \C{new\_order} transaction when executed
under a \emph{Read Committed} (RC) isolation level, the default
isolation level in 8 of the 18 databases studied
in~\cite{bailishotos}.  An executing RC transaction is isolated from
\emph{dirty writes}, i.e., writes of uncommitted transactions, but is
allowed to witness the writes of concurrent transactions as soon as
they are committed.  Thus, with two concurrent instances of the
\C{new\_order} transaction (call them $T_1$ and $T_2$), both
concurrently placing new orders for different customers under the same
district (\C{d\_id}), RC isolation allows the execution shown in
Fig.~\ref{fig:new_order_exec}.


The figure depicts an execution as a series of SQL operations. In the
execution, the \C{new\_order} instance $T_1$ (green) reads the
\C{d\_next\_o\_id} field of the district record for \C{d\_id}, but
before it increments the field, another \C{new\_order} instance
$T_2$ (red) begins its execution and commits. Note that $T_2$ reads the
same \C{d\_next\_o\_id} value as $T_1$, and inserts new \C{Order} and
\C{Order\_line} records with their \C{o\_id} and \C{ol\_o\_id} fields
(resp.) equal to \C{d\_next\_o\_id}. $T_2$ also increments the
\C{d\_next\_o\_id} field, which $T_1$ has already acccessed. This is
allowed because reads typically do not obtain a mutually exclusive
lock on most databases. After $T_2$'s commit, $T_1$ resumes execution
and adds new \C{Order} and \C{Order\_line} fields with the same order
id as $T_1$. Thus, at the end of the execution, \C{Order\_line}
records inserted by $T_1$ and $T_2$ all bear the same order id. There
are also two \C{Order} records with the same district id (\C{d\_id})
and order id, none of whose \C{o\_ol\_cnt} reflects the actual number
of \C{Order\_line} records inserted with that order id.  This clearly
violates TPC-C's consistency condition.

This example does not exhibit any of the anomalies that
\emph{characterize} RC isolation~\cite{berenson}\footnote{Berenson
\emph{et al.} characterize isolation levels in terms of the anomalies
they \emph{admit}. For example, RC is characterized by \emph{lost
writes} because it admits the anomaly.}. For instance, there are no
\emph{lost writes} since both concurrent transactions' writes are
present in the final state of the database.  Program analyses that aim
to determine appropriate isolation by checking for possible
manifestations of RC-induced anomalies would fail to identify grounds
for promoting the isolation level of \C{new\_order} to something
stronger.  Yet, if we take the semantics of the application into
account, it is quite clear that RC is not an appropriate isolation
level for \C{new\_order}.

While reasoning in terms of anomalies is cumbersome and inadequate,
reasoning about weak isolation in terms of
traces~\cite{adyaphd,gotsmanconcur15} on memory read and write actions
can complicate high-level reasoning.  A possible alternative would be
to utilize concurrent program verification methods where the
implementation details of weak isolation are interleaved within the
program, yielding a (more-or-less) conventional concurrent program.
But, considering the size and complexity of real-world transaction
systems, this strategy is unlikely to scale.

In this paper, we adopt a different approach that \emph{lifts}
isolation semantics (\emph{not} their implementations) to the
application layer, providing a principled framework to simultaneously
reason about application invariants and isolation properties.  To
illustrate this idea informally, consider how we might verify that
\C{new\_order} is sound when executed under \emph{Snapshot Isolation}
(SI), a stronger isolation level than RC. Snapshot isolation allows
transactions to be executed against a private snapshot of the
database, thus admitting concurrency, but it also requires that there
not be any write-write conflicts (i.e., such a conflict occurs if
concurrently executing transactions modify the same record) among
concurrent transactions when they commit. Write-write conflicts can be
eliminated in various ways, e.g., through conflict detection followed
by a rollback, or through exclusive locks, or a combination of both.
For instance, one possible implementation of SI, close to the one used
by PostgreSQL~\cite{postgresiso}, executes a transaction against its
private snapshot of the database, but obtains exclusive locks on the
actual records in the database before performing writes. A write is
performed only if the record that is to be written has not already
been updated by a concurrent transaction. Conflicts are resolved by
abort and roll back.



As this discussion hints, implementations of SI on real databases
such as PostgreSQL are highly complicated, often running into
thousands of lines of code.  Nonetheless, the semantics of SI, in
terms of how it effects transitions on the database state, can be
captured in a fairly simple model. First, effects induced by one
transaction (call it $T$) are not visible to another concurrently executing one
during $T$'s execution. Thus, from $T$'s perspective,
the global state does not change during its execution.  More formally,
for every operation performed by $T$, the global state $T$ witnesses
before ($\stg$) and after ($\stg')$ executing the operation is the
same ($\stg'=\stg$).  After $T$ finishes execution, it commits its
changes to the actual database, which may have already incorporated
the effects of concurrent transactions. In executions where $T$
successfully commits, concurrent transactions are guaranteed to not be
in write-write conflict with $T$. Thus, if $\stg$ is the global state
that $T$ witnessed when it finished execution (the snapshot state),
and $\stg'$ is the state to which $T$ commits, then the difference
between $\stg$ and $\stg'$ should not result in a write-write conflict
with $T$. To concretize this notion, let the database state be a map
from database locations to values, and let $\stl$ denote a
transaction-local log that maps the locations being written to their
updated values.  The absence of write-write conflicts between $T$ and
the diff between $\stg$ and $\stg'$ can be expressed as: $\forall
x\in\mathit{dom}(\delta)$, $\stg'(x) = \stg(x)$.  In other words, the
semantics of SI can be captured as an axiomatization over transitions
of the database state ($\Delta \longrightarrow \Delta'$) during a
transaction's ($T$) lifetime:
\begin{itemize}
  \item While $T$ executes, $\Delta' = \Delta$.
  \item After $T$ finishes execution, but before it commits its local
    state $\delta$, $\forall(x\in\mathit{dom}(\delta)).~\Delta'(x) = \Delta(x)$.
\end{itemize}

\noindent This simple characterization of SI isolation allows us to
verify the consistency conditions associated with the \C{new\_order}
transaction.  First, since the database does not change ($\Delta' =
\Delta$) during execution of the transaction's body, we can reason
about \C{new\_order} as though it executed in complete isolation until
its commit point, leading to a verification process similar to what
would have been applied when reasoning sequentially.  When
\C{new\_order} finishes execution, however, but before it commits, the
SI axiomatization shown above requires us to consider global state
transitions $\stg \longrightarrow \stg'$ that do not include changes
to the records ($\stl$) written by \C{new\_order}, i.e.,
$\forall(x\in\mathit{dom}(\delta)).~\Delta'(x) = \Delta(x)$.  The
axiomatization precludes any execution in which there are concurrent
updates to shared table fields (e.g., \C{d\_next\_o\_id} on the same
\C{District} table), but does not prohibit interferences that write to
different tables, or write to different records in the same table.  We
need to reason about the safety of such interferences with respect to
\C{new\_order}'s consistency invariants to verify \C{new\_order}.

We approach the verification problem by first observing that a
relational database is a significantly simpler abstraction than shared
memory. Its primary data structure is a table, with no primitive
support for pointers, linked data structures, or aliasing.  Although a
database essentially abstracts a mutable state, this state is managed
through a well-defined fixed number of interfaces (SQL statements),
each tagged with a logical formula describing what records are
accessed and updated.

This observation leads us away from thinking of a collection of
database transactions as a simple variant of a concurrent imperative
program.  Instead, we see value in viewing them as essentially
functional computations that manage database state abstractly,
mirroring the structure of our DSL.  By doing so, we can formulate the
semantics of database operations as state transformers that explicitly
relate an operation's pre- and post-states, defining the semantics of
the corresponding transformer algorithmically, just like classical
predicate transformer semantics (e.g., weakest pre-condition or
strongest post-condition).  In our case, a transformer interprets a
SQL statement in the set domain, modeling the database as a set of
records, and a SQL statement as a function over this set.  Among other
things, one benefit of this approach is that low-level loops can now
be substituted with higher-order combinators that automatically lift
the state transformer of its higher-order argument, i.e., the loop
body, to the state transformer of the combined expression, i.e., the
loop.  We illustrate this intuition on a simple example.


\begin{figure}[!h]
\begin{ocaml}
foreach item_reqs @@ fun item_req ->
  SQL.update Stock (fun s -> {s with s_qty = k1}) 
                   (fun s -> s.s_i_id = item_req.ol_i_id);
  SQL.insert Order_line {ol_o_id=k2; ol_d_id=k3; 
                         ol_i_id=item_req.ol_i_id; ol_qty=item_req.ol_qty}
\end{ocaml}
\caption{Foreach loop from Fig.~\ref{fig:new_order_code}}
\label{fig:foreach_code}
\end{figure}

Fig.~\ref{fig:foreach_code} shows a (simplified) snippet of code taken
from Fig.~\ref{fig:new_order_code}. Some irrelevant expressions have
been replaced with constants (\C{k1}, \C{k2}, and \C{k3}).  The body of
the loop executes a SQL update followed by an insert.  Recall that a
transaction reads from the global database ($\Delta$), and writes to a
transaction-local database ($\delta$) before committing these
updates. An update statement filters the records that match the search
criteria from $\Delta$ and computes the updated records that are to be
added to the local database. Thus, the state transformer for the
update statement (call it $\T_U$) is the following function on
sets\footnote{Bind ($\bind$) has higher precedence than union
($\cup$). Angle braces ($\langle \ldots \rangle$) are used to denote
records.}:
\begin{smathpar}
\begin{array}{l}
  \lambda(\stl,\stg).~ \stl \cup \stg \bind(\lambda \C{s}. 
      \itel{{\sf table}(\C{s}) = \C{Stock} \conj 
        \C{s}.\C{s\_i\_id} = \C{item\_req.ol\_i\_id}\\\hspace*{1.15in}}
           {\{ \langle \C{s\_i\_id}=\C{s.s\_i\_id};\, 
                       \C{s\_d\_id}=\C{s.s\_d\_id};\,
                       \C{s\_qty} = \C{k1}\rangle \}\\\hspace*{1.15in}}
           {\emptyset})
\end{array}
\end{smathpar}
Here, the set bind operator extracts record elements ($\C{s}$) from
the database, checks the precondition of the update action, and if
satisfied, constructs a new set containing a single record that is
identical to $\C{s}$ except that it binds field $\C{s\_qty}$ to value
$\C{k_1}$.  This new set is added (via set union) to the existing
local database state $\delta$.\footnote{For now, assume that the
  record being added is not already present in $\stl$.}

The transformer ($\T_I(\delta,\Delta)$) for the subsequent \C{insert}
statement can be similarly constructed:
\begin{smathpar}
\begin{array}{l}
  \lambda(\stl,\stg).~ \stl \cup
       \{\langle\C{ol\_o\_id}=\C{k2};\,
                 \C{ol\_d\_id}=\C{k3};\, \C{ol\_i\_id}=\C{item\_req.ol\_i\_id};\, 
                 \C{ol\_qty}=\C{item\_req.ol\_qty}\rangle\}
 
\end{array}
\end{smathpar}
Observe that both transformers are of the form $\T(\delta,\Delta) =
\stl \cup \F(\stg)$, where $\F$ is a function that returns the
set of records added  to the transaction-local database ($\stl$). Let
$\F_U$ and $\F_I$ be the corresponding functions for $\T_U$ and $\T_I$
shown above. The state transformation induced by the loop body in
Fig.~\ref{fig:new_order_code} can be expressed as the following
composition of $\F_U$ and $\F_I$:
\begin{smathpar}
\begin{array}{l}
  \lambda(\stl,\stg).~ \stl \cup \F_U(\stg) \cup \F_I(\stg)
\end{array}
\end{smathpar}
The transformer for the loop itself can now be computed to be:
\begin{smathpar}
\begin{array}{l}
  \lambda(\stl,\stg).~ \stl \cup \C{item\_reqs}\bind
      (\lambda\C{item\_req}.~ \F_U(\stg) \cup \F_I(\stg))
\end{array}
\end{smathpar}
Observe that the structure of the transformer mirrors the structure of
the program itself. In particular, SQL statements become set
operations, and the \C{foreach} combinator becomes set monad's bind
($\bind$) combinator.  As we demonstrate, the advantage of inferring
such transformers is that we can now make use of a
semantics-preserving translation from the domain of sets equipped with
$\bind$ to a decidable fragment of first-order logic, allowing us to
leverage SMT solvers for automated proofs without having to infer
potentially complex thread-local invariants or intermediate
assertions.  Sec.~\ref{sec:inference} describes this translation. In
the exposition thus far, we assumed $\Delta$ remains invariant, which
is clearly not the case when we admit concurrency.  Necessary
concurrency extensions of the state transformer semantics to deal with
interference is also covered in Sec.~\ref{sec:inference}.  Before
presenting the transformer semantics, we first focus our attention in
the following two sections on the theoretical foundations for weak
isolation, upon which this semantics is based.

\section{\txnimp: Syntax and Semantics}
\label{sec:opsem}

\label{sec:syntax}

\begin{figure*}[!t]
\raggedright
\textbf{Syntax}\\
\begin{smathpar}
\renewcommand{\arraystretch}{1.2}
\begin{array}{lclcl}
\multicolumn{5}{c} {
  {x,y} \in \mathtt{Variables}\;\;\spc
  {f} \in \mathtt{Field\;Names} \;\;\spc
  {i,j} \in \mathbb{N} \;\;\spc
  {\odot} \in \{+,-,\le,\ge,=\}\;\;\spc
  {k} \in \mathbb{Z}\cup\mathbb{B} \;\;\spc
  {\rec} \in \langle \bar{f}=\bar{k}\rangle
}\\
{\stl,\stg,s} & \in & \mathtt{State} & \coloneqq &  \Pow{\langle
  \bar{f}=\bar{k} \rangle} \\
{\I_e, \I_c }  & \in & \mathtt{Isolation Spec} & \coloneqq & (\stl,\stg,\stg') \rightarrow \Prop\\
v & \in & \mathtt{Values} & \coloneqq & k \ALT \rec \ALT s\\
e & \in & \mathtt{Expressions} & \coloneqq & v \ALT x \ALT x.f 
    \ALT \langle \bar{f}=\bar{e} \rangle \ALT e_1 \odot e_2\\ 
c & \in & \mathtt{Commands} & \coloneqq & \lete{x}{e}{c}
    \ALT \ite{e}{c_1}{c_2}\ALT c_1;c_2 \ALT \inserte{x}  \\
&&&&\ALT \deletee{\lambda x.e}
    \ALT \lete{y}{\selecte{\lambda x.e}}{c}
    \ALT \updatee{\lambda x.e_1}{\lambda x.e_2}\\
&&&&\ALT \foreache{x}{\lambda y.\lambda z. c} 
    \ALT \foreachr{s_1}{s_2}{\lambda x.\lambda y. e}\\
&&&&\ALT \ctxn{i}{\I}{ c } \ALT \ctxnr{i}{\I,\stl,\stg}{c} \ALT c_1 ||
  c_2 \ALT \cskip \\
\ectx & \in & \mathtt{Eval\;Ctx} & ::= & \bullet \ALT  
  \bullet || c_2 \ALT c_1 || \bullet \ALT \bullet;\,c_2 
  \ALT \ctxnr{i}{\I,\stl,\stg}{\bullet} \\
\end{array}
\end{smathpar}
\bigskip

\renewcommand{\arraystretch}{1.2}

\textbf{Local Reduction} \quad 
\fbox {\(\stg \vdash (\tbox{c}_i,\stl) \stepsto (\tbox{c'}_i,\stl')\)}\\
\bigskip

\begin{minipage}{2.45in}
\rulelabel{E-Insert}
\begin{smathpar}
\begin{array}{c}
\RULE
{
  r.\idf \not\in \dom(\stl \cup \stg)\\
  r' = \langle r \;\C{with}\; \txnf=i;\,\delf=\C{false} \rangle
}
{
  \stg \vdash (\tbox{\inserte{r}}_i,\stl) \stepsto
  (\tbox{\cskip}_i,\stl \cup \{r'\})
}
\end{array}
\end{smathpar}
\end{minipage}
\begin{minipage}{2.6in}
\rulelabel{E-Select}
\begin{smathpar}
\begin{array}{c}
\RULE
{
  \\
  s = \{r\in\Delta \,|\, \eval([r/x]e)=\C{true}\}\spc
  c' = [s/y]c
}
{
  \stg \vdash (\tbox{\lete{y}{\selecte{\lambda x.e}}{c}}_i, \stl) \stepsto 
              (\tbox{c'}_i,\stl)
}
\end{array}
\end{smathpar}
\end{minipage}
\bigskip

\begin{minipage}{2.5in}
\rulelabel{E-Delete}
\begin{smathpar}
\begin{array}{c}
\RULE
{
  \dom(\stl)\cap\dom(s) = \emptyset \\
  \hspace*{-0.4in}
  s = \{r' \,|\, \exists(r\in\Delta).~ \eval([r/x]e)=\C{true} \\
      \hspace*{0.2in}\conj r'=\langle r \with \delf=\C{true};\,
      \txnf=i \rangle\}\\
}
{
  \stg \vdash (\tbox{\deletee{\lambda x.e}}_i,\stl) \stepsto 
  (\tbox{\cskip}_i,\stl \cup s)
}
\end{array}
\end{smathpar}
\end{minipage}
\begin{minipage}{2.8in}
\rulelabel{E-Update}
\begin{smathpar}
\begin{array}{c}
\RULE
{
  \dom(\stl) \cap \dom(s) = \emptyset\\
  \hspace*{-0.8in}s = \{r' \,|\, \exists(r\in\Delta).~ 
    \eval([r/x]e_2)=\C{true} \conj \\
  r'= \langle [r/x]e_1 \;\C{with}\;
    \idf=r.\idf;\,\txnf=i;\,\delf = r.\delf \rangle \}
}
{
  \stg \vdash (\tbox{\updatee{\lambda x.e_1}{\lambda x.e_2}}_i,\stl) 
      \stepsto (\tbox{\cskip}_i,\stl \cup s)
}
\end{array}
\end{smathpar}
\end{minipage}
\bigskip

\begin{smathpar}
\begin{array}{ll}
  \rulelabel{E-Foreach1} & \stg \vdash (\tbox{\foreache{s}{\lambda y.\lambda
    z.c}}_i,\stl) \stepsto (\tbox{\foreachr{\emptyset}{s}{\lambda
    y.\lambda z. c}}_i,\stl)\\
  \rulelabel{E-Foreach2} & \stg \vdash (\tbox{\foreachr{s_1}{\{r\} \uplus s_2}
    {\lambda y.\lambda z.c}}_i,\stl) \stepsto (\tbox{[r/z][s_1/y]c;\,\\
    & \hspace*{2.5in}\foreachr{s_1 \cup \{r\}}{s_2}{\lambda y.\lambda z. c}}_i,\stl)\\
  \rulelabel{E-Foreach3} & \stg \vdash (\tbox{\foreachr{s}{\emptyset}
    {\lambda y.\lambda z.c}}_i,\stl) \stepsto (\tbox{\cskip}_i,\stl)\\
\end{array}
\end{smathpar}
\bigskip

\textbf{Top-Level Reduction} \quad 
\fbox {\((c,\stg) \stepsto (c',\stg')\)}\\
\bigskip

\begin{minipage}{2.5in}
  \rulelabel{E-Txn-Start}
  \begin{smathpar}
  \begin{array}{c}
    \RULE{\\}
         {(\ctxn{i}{\I}{c},\stg) \stepsto (\ctxnr{i}{\I,\emptyset,\stg}{c},\stg)}
  \end{array}
  \end{smathpar}
\end{minipage}%
\begin{minipage}{2.3in}
\rulelabel{E-Txn}
\begin{smathpar}
\begin{array}{c}
\RULE
{
  \I_e\,\,(\stl,\stg,\stg')\spc
  \stg \vdash (\tbox{c}_i,\stl) \stepsto (\tbox{c'}_i,\stl')
}
{
  (\ctxnr{i}{\I,\stl,\stg}{c},\stg') \stepsto
  (\ctxnr{i}{\I,\stl',\stg'}{c'},\stg')
}
\end{array}
\end{smathpar}
\end{minipage}\\
\bigskip

\begin{center}
\begin{minipage}{3in}
\rulelabel{E-Commit}
\begin{smathpar}
\begin{array}{c}
\RULE
{
  \I_c\,\,(\stl,\stg,\stg')
}
{
  (\ctxnr{i}{\I,\stl,\stg}{\cskip},\stg') \stepsto (\cskip,\stl \rhd \stg')
}
\end{array}
\end{smathpar}
\end{minipage}
\end{center}
\hfill
\caption{\small \txnimp: Syntax and Small-step semantics}
\label{fig:txnimp}
\end{figure*}

Fig.~\ref{fig:txnimp} shows the syntax and small-step semantics of
\txnimp, a core language that we use to formalize our intuitions about
reasoning under weak isolation. Variables ($x$), integer and boolean
constants ($k$), records ($r$) of named constants, sets ($s$) of such
records, arithmetic and boolean expressions ($e_1 \odot e_2$), and
record expressions ($\langle \bar{f}=\bar{e} \rangle$) constitute the
syntactic class of expressions ($e$). Commands ($c$) include $\cskip$,
conditional statements, \C{LET} constructs to bind names, \C{FOREACH}
loops, SQL statements, their sequential composition ($c_1;c_2$),
transactions ($\ctxn{i}{\I}{c}$) and their parallel composition
($c_1\,||\,c_2$). Each transaction is assumed to have a unique
identifier $i$, and executes at the top-level; our semantics does not
support nested transactions. The $\I$ in the \C{TXN} block syntax is
the transaction's isolation specification, whose purpose is explained
below.  Certain terms that only appear at run-time are also present in
$c$.  These include a {\sf txn} block tagged with sets ($\stl$ and
$\stg$) of records representing local and global database state, and a
runtime {\sf foreach} expression that keeps track of the set ($s_1$)
of items already iterated, and the set ($s_2$) of items yet to be
iterated. Note that the surface-level syntax of the \C{FOREACH}
command shown here is slightly different from the one used in previous
sections; its higher-order function has two arguments, $y$ and $z$,
which are invoked (during the reduction) with the set of
already-iterated items, and the current item, respectively. This form
of \C{FOREACH} lends itself to inductive reasoning that will be useful
for verification (Sec.~\ref{sec:reasoning}). Our language ensures
that all effectful actions are encapsulated within database commands,
and that all shared state among processes are only manipulated via
transactions and its supported operations.  In particular, we do not
consider programs in which objects resident on e.g., the OCaml heap
are concurrently manipulated by OCaml expressions as well as database
actions.


We define a small-step operational semantics for this language in
terms of an abstract machine that executes a command, and updates
either a transaction-local ($\stl$), or global ($\stg$) database, both
of which are modeled as a set of records of a pre-defined type, i.e.,
they all belong to a single table.  The generalization to multiple
tables is straightforward, e.g., by having the machine manipulate a
set of sets, one for each table.  The semantics assumes that records
in $\stg$ can be uniquely identified via their $\idf$ field, and
enforces this property wherever necessary. Certain hidden fields are
treated specially by the operational semantics, and are hidden
from the surface language. These include a $\txnf$ field that tracks
the identifier of the transaction that last updated the record, and a
$\delf$ field that flags deleted records in $\stl$.  For a set $S$ of
records, we define $\dom(S)$ as the set of unique ids of all records
in $S$. Thus $|\dom(\stg)| = |\stg|$. During its execution, a
transaction may write to multiple records in $\stg$. Atomicity
dictates that such writes should not be visible in $\stg$ until the
transaction commits. We therefore associate each transaction with a
local database ($\stl$) that stores such uncommitted
records\footnote{While SQL's \C{UPDATE} admits writes at the
  granularity of record fields, most popular databases enforce
  record-level locking, allowing us to think of ``uncommitted writes''
  as ``uncommitted records''. }. Uncommitted records include deleted
records, whose $\delf$ field is set to \C{true}. When the transaction
commits, its local database is atomically \emph{flushed} to the global
database, committing these heretofore uncommitted records. The flush
operation ($\rhd$) is defined as follows:
\begin{smathpar}
\begin{array}{c}
\forall r.~ r \in (\stl\rhd\stg) ~\Leftrightarrow~ 
  (r.\idf \notin \dom(\stl) \conj r \in \stg)
\disj (r \in \stl \conj \neg r.\delf) 
\end{array}
\end{smathpar}
Let $\stg' = \stl\rhd\stg$. A record $r$ belongs to $\stg'$ iff it
belongs to $\stg$ and has not been updated in $\stl$, i.e., $r.\idf
\notin \dom(\stl)$, or it belongs to $\stl$, i.e., it is either a new
record, or an updated version of an old record, provided the update is
not a deletion ($\neg r.\delf$).  Besides the commit, flush also helps
a transaction read its own writes. Intuitively, the result of a read
operation inside a transaction must be computed on the database
resulting from flushing the current local state ($\stl$) to the global
state ($\stg$). The abstract machine of Fig.~\ref{fig:txnimp},
however, does not let a transaction read its own writes. This
simplifies the semantics, without losing any generality, since
substituting $\stl\rhd\stg$ for $\stg$ at select places in the
reduction rules effectively allows reads of uncommitted transaction
writes to be realized, if so desired.

The small-step semantics is stratified into a transaction-local
reduction relation, and a top-level reduction relation. The
transaction-local relation ($\stg \vdash (c,\stl) \stepsto
(c',\stl')$) defines a small-step reduction for a command inside a
transaction, when the database state is $\stg$; the command $c$
reduces to $c'$, while updating the transaction-local database $\stl$
to $\stl'$. The definition assumes a meta-function $\eval$ that
evaluates closed terms to values. The reduction relation for SQL
statements is defined straightforwardly. \C{INSERT} adds a new record
to $\stl$ after checking the uniqueness of its id. \C{DELETE} finds
the records in $\stg$ that match the search criteria defined by its
boolean function argument, and adds the records to $\stl$ after
marking them for deletion. \C{SELECT} bounds the name introduced by
\C{LET} to the set of records from $\stg$ that match the search
criteria, and then executes the bound command $c$. \C{UPDATE} uses its
first function argument to compute the updated version of the records
that match the search criteria defined by its second function
argument. Updated records are added to $\stl$.

The reduction of \C{FOREACH} starts by first converting it to its
run-time form to keep track of iterated items ($s_1$), as well as
yet-to-be-iterated items ($s_2$).  Iteration involves invoking its
function argument with $s_1$ and the current element $x$ (note:
$\uplus$ in $\{x\} \uplus s_2$ denotes a disjoint union). The
reduction ends when $s_2$ becomes empty. The reduction rules for
conditionals, \C{LET} binders, and sequences are standard, and
omitted for brevity.

The top-level reduction relation defines the small-step semantics of
transactions, and their parallel composition. A transaction comes
tagged with an \emph{isolation specification} $\I$, which has two components
$\I_e$ and $\I_c$, that dictate the timing and nature of interferences
that the transaction can witness, during its execution ($\I_e$), and
when it is about to commit ($\I_c$).  Formally, $\I_e$ and $\I_c$ are
predicates over the (current) transaction-local database state
($\stl$), the state ($\stg$) of the global database when the
transaction last took a step, and the current state ($\stg'$) of the
global database.  Intuitively, $\stg'\neq\stg$ indicates an
interference from another concurrent transaction, and the predicates
$\I_e$ and $\I_c$ decide if this interference is allowed or not,
taking into account the local database state ($\stl$). For instance,
as described in \S\ref{sec:motivation}, an SI transaction on
PostgreSQL defines $\I$ as follows:
\begin{smathpar}
\begin{array}{lcl}
\I_e\,\,(\stl,\stg,\stg') & = & \stg' = \stg\\
\I_c\,\,(\stl,\stg,\stg') & = & \forall(r\in\stl)(r'\in\stg).~ r'.\idf = r.\idf \Rightarrow r'\in\stg'
\end{array}
\end{smathpar}
This definition dictates that no change to the global database state
can be visible to an SI transaction while it executes ($\I_e$), and
there should be no concurrent updates to records written by the
transaction by other concurrently executing ones ($\I_c$).
To simplify the presentation, we use $\I$ instead of $\I_e$ and $\I_c$
when its destructed form is not required.

The reduction of a $\ctxn{i}{\I}{c}$ begins by first converting it to
its run-time form $\ctxnr{i}{\I,\stl,\stg}{c}$, where $\stl =
\emptyset$, and $\stg$ is the current (global) database.  Rule
\rulelabel{E-Txn} reduces $\ctxnr{i}{\I,\stl,\stg}{c}$ under a
database state ($\stg'$), only if the transaction-body isolation
specification ($\I_e$) allows the interference between $\stg$ and
$\stg'$.   Rule \rulelabel{E-Commit} commits the
transaction $\ctxnr{i}{\I,\stl,\stg}{c}$ by flushing its uncommitted
records to the database. This is done only if the interference between
$\stg$ and $\stg'$ is allowed at the commit point by the isolation
specification ($\I_c$).  The distinction between $\I_e$ and $\I_c$
allows us to model the snapshot semantics of realistic isolation
levels that isolate a transaction from interference during its
execution, but expose interferences at the commit point.

\textbf{Local Context Independence} As mentioned previously, our
operational semantics does not let a transaction read its own writes.
It also does not let a transaction overwrite its own writes, due to
the premise $\dom(\stl)\cap\dom(s) = \emptyset$ on the
\rulelabel{E-Delete} and \rulelabel{E-Update} rules. We refer to this
restriction as \emph{local context independence}.  This restriction is
easy to relax in the operational semantics and the reasoning framework
presented in the next section; our inference procedure described in
\S\ref{sec:inference}, however, has a non-trivial dependence on this
assumption.  Nonetheless, we have encountered few instances in
practice where enforcing local context independence turns out to be a
severe restriction. Indeed, all of the transactions we have considered
in our benchmarks (e.g., TPC-C) satisfy this assumption.

\subsection{Isolation Specifications}
\label{sec:isolation}

A distinctive characteristic of our development is that it is
parameterized on a weak isolation specification $\I$ that can be
instantiated with the declarative characterization of an isolation
guarantee or a concurrency control mechanism, regardless of the actual
implementation used to realize it. This allows us to model a range of
isolation properties that are relevant to the theory and practice of
transaction processing systems without appealing to specific
implementation artifacts like locks, versions, logs, speculation, etc.
A few well-known properties are discussed below:


\textbf{Unique Ids}. As the \C{new\_order} example
(\S\ref{sec:motivation}) demonstrates, enforcing global uniqueness
of ordered identifiers requires stronger isolation levels than the
ones that are default on most databases (e.g., Read
Committed). Alternatively, globally unique sequence numbers,
regardless of the isolation level, can be requested from a relational
database via SQL's \C{UNIQUE} and \C{AUTO\_INCREMENT} keywords. Our
development crucially relies on the uniqueness of record
identifiers\footnote{The importance of unique ids is recognized in
  real-world implementations.  For example, MySQL's InnoDB engine
  automatically adds a 6-byte unique identifier if none exists for a
  record.}, which are checked locally for uniqueness by the
\rulelabel{E-Insert} rule.  The global uniqueness of locally unique
identifiers can be captured as an isolation property thus:
\begin{smathpar}
\begin{array}{lcl}
  \I_{id}(\stl,\stg,\stg') & = & \forall(r\in\stl).~
      r.\idf\notin \dom(\stg) \Rightarrow r.\idf\notin \dom(\stg')
\end{array}
\end{smathpar}
$\I_{id}$ ensures that if the id of a record is globally unique when
it is added to a transaction's $\stl$, it remains globally unique
until the transaction commits. This would be achieved within our
semantic framework by prohibiting the interference from a concurrent
transaction that adds the same id. The axiom thus simulates a global
counter protected by an exclusive lock without explicitly appealing to
an implementation artifact.

\textbf{Write-Write Conflicts}. Databases often employ a combination
of concurrency control methods, both optimistic (e.g., speculation and
rollback) and pessimistic (e.g., various degrees of locking), to
eliminate write-write (\emph{ww}) conflicts among concurrent
transactions. We can specify the absence of such conflicts using our
tri-state formulation thus:
\begin{smathpar}
\begin{array}{lcl}
  \I_{ww}(\stl,\stg,\stg') & = & \forall(r'\in\stl)(r \in \stg).~
      r.\idf = r'.\idf  \Rightarrow r\in\stg'
\end{array}
\end{smathpar}
That is, given a record $r'\in\stl$, if there exists an $r\in\stg$
with the same id (i.e., $r'$ is an updated version of $r$), then $r$
must be present unmodified in $\stg'$. This prevents a concurrent
transaction from changing $r$, thus simulating the behavior of an
exclusive lock or a speculative execution that succeeded (Note: a
transaction writing to $r$ always changes $r$ because its $\txnf$
field is updated). 

\textbf{Snapshots} Almost all major relational databases implement
isolation levels that execute transactions against a static snapshot
of the database that can be axiomatized thus:
\begin{smathpar}
\begin{array}{lcl}
  \I_{ss}(\stl,\stg,\stg') & = & \stg' = \stg
\end{array}
\end{smathpar}

\textbf{Read-Only Transactions}. Certain databases implement special
privileges for read-only transactions. Read-only behavior can be
enforced on a transaction by including the following proposition as
part of its isolation invariant:
\begin{smathpar}
\begin{array}{lcl}
  \I_{ro}(\stl,\stg,\stg') & = & \stl = \emptyset\\
\end{array}
\end{smathpar}

In addition to these properties, various specific isolation levels
proposed in the database or distributed systems literature, or
implemented by commercial vendors can also be specified within this
framework:

\textbf{Read Committed (RC) and Monotonic Atomic View (MAV).} RC
isolation allows a transaction to witness writes of committed
transactions at any point during the transaction's execution.
Although it offers only weak isolation guarantees, it nonetheless
prevents witnessing \emph{dirty writes} (i.e., writes performed by
uncommitted transactions).  Monotonic Atomic View
(MAV)~\cite{bailishat} is an extension to RC that guarantees the
continuous visibility of a committed transaction's writes once they
become visible in the current transaction. That is, a MAV transaction
does not witness \emph{disappearing writes}, which can happen on a
weakly consistent machine. Due to the SC nature of our abstract
machine (there is always a single global database state $\stg$; not a
vector of states indexed by vector clocks), and our choice to never
violate atomicity of a transaction's writes, both RC and MAV are
already guaranteed by our semantics.  Thus, defining $\I_e$ and $\I_c$
to \emph{true} ensures RC and MAV behavior under our semantics.

\textbf{Repeatable Read (RR)} By definition, multiple reads to a
transactional variable in a Repeatable Read transaction are required
to return the same value.  RR is often implemented (for e.g., in
~\cite{mysqliso,bailishat}) by executing the transaction against a
(conceptual) snapshot of the database, but committing its writes to
the actual database. This implementation of RR can be axiomatized as
$\I_e = \I_{ss}$ and $\I_{c}=true$. However, this specification of RR
is stronger than the ANSI SQL specification, which requires no more
than the invariance of already read records. In particular, ANSI SQL
RR allows \emph{phantom reads}, a phenomenon in which a repeated
\C{SELECT} query might return newly inserted records that were not
previously returned. This specification is implemented, for e.g., in
Microsoft's SQL server, using record-level exclusive read locks, that
prevent a record from being modified while it is read by an
uncommitted transaction, but which does not prohibit insertion of new
records. The ANSI SQL RR specification can be axiomatized in our
framework, but it requires a minor extension to our operational
semantics to track a transaction's reads. In particular, the records
returned by \C{SELECT} should be added to the local database $\stl$,
but without changing their transaction identifiers ($\txnf$ fields),
and flush ($\rhd$) should only flush the records that bear the current
transaction's identifier. With this extension, ANSI SQL RR can be
axiomatized thus:
\begin{smathpar}
\begin{array}{lcl}
  \I_e(\stl,\stg,\stg') & \Leftrightarrow & \forall(r\in\stl).
      r \in \Delta \Rightarrow r \in \Delta'\\
  \I_c(\stl,\stg,\stg') & \Leftrightarrow & true\\
\end{array}
\end{smathpar}
If a record $r$ belongs to both $\stl$ and $\stg$, then it must be a
record written by a different transaction and read by the current
transaction (since the current transaction's records are not yet
present in $\stg$). By requiring $r\in\stg'$, $\I_e$ guarantees the
invariance of $r$, thus the repeatability of the read. 

\textbf{Snapshot Isolation (SI)} The concept of executing a
transaction against a consistent snapshot of the database was first
proposed as Snapshot Isolation in~\cite{berenson}. SI doesn't admit
write-write conflicts, and the original proposal, which is implemented
in Microsoft SQL Server, required the database to roll-back an SI
transaction if conflicts are detected during the commit. This behavior
can be axiomatized as $\I_e = \I_{ss}$ (execution against a snapshot),
and $\I_c = \I_{ww}$ (avoiding write-write conflicts during the
commit).
Note that the same axiomatization applies to PostgreSQL's RR,
although its implementation (described in Sec.~\ref{sec:motivation})
differs considerably from the original proposal. Thus, reasoning done
for an SI transaction on MS SQL server carries over to PostgreSQL's RR
and vice-versa, demonstrating the benefits of reasoning axiomatically
about isolation properties.

\textbf{Serializability (SER)} The specification of serializability is
straightforward:
\begin{smathpar}
\begin{array}{lcl}
\I_e\,\,(\stl,\stg,\stg') & = & \stg' = \stg\\
\I_c\,\,(\stl,\stg,\stg') & = & \stg' = \stg
\end{array}
\end{smathpar}

\section{The Reasoning Framework}
\label{sec:reasoning}


We now describe a proof system that lets us prove the correctness of a
\txnimp program $c$ w.r.t its high-level consistency conditions $I$,
on an implementation that satisfies the isolation specifications
($\I$) of its transactions\footnote{Note the difference between $I$
  and $\I$. The former constitute \emph{proof} \emph{obligations} for
  the programmer, whereas the latter describes a transaction's
  \emph{assumptions} about the operational characteristics of the
  underlying system.}.  Our proof system is essentially an adaptation
of a rely-guarantee reasoning framework~\cite{rgjones} to the setting
of weakly isolated database transactions.  The primary challenge in
the formulation deals with how we relate a transaction's isolation
specification ($\I$) to its rely relation ($R$) that describes the
transaction's environment, so that interference is considered only
insofar as allowed by the isolation level.  Another characteristic of
the transaction setting that affects the structure of the proof system
is atomicity; we do not permit a transaction's writes to be visible
until it commits.  In the context of rely-guarantee, this means that
the transaction's guarantee ($G$) should capture the aggregate effect
of a transaction, and not its individual writes.  While shared memory
\C{atomic} blocks also have the same characteristic, the fact that
transactions are weakly-isolated introduces non-trivial complexity.
Unlike an \C{atomic} block, the effect of a transaction is \emph{not}
a sequential composition of the effects of its statements because each
statement can witness a potentially different version of the state.

\subsection{The Rely-Guarantee Judgment}
\label{sec:rely-guarantee}

\begin{figure}[t]
\raggedright
\fbox {\( \R \vdash \hoare{P}{c}{Q} \)} 
\quad \fbox {\( \rg{I,R}{c}{G,I} \)}\\[4pt]
%
\bigskip

\rulelabel{RG-Select}
\begin{center}
\begin{smathpar}
\begin{array}{c}
\RULE
{
  P(\stl,\stg) \conj
    x = \{r \,|\, r\in\Delta \wedge [r/y]e\} \Rightarrow
  P'(\stl,\stg)\spc
  \R \vdash \hoare{P'}{c}{Q}\spc
  \stable(\R,P')
}
{
  \R \vdash \hoare{P}{\lete{x}{\selecte{\lambda y.e}}{c}}{Q}
}
\end{array}
\end{smathpar}
\end{center}
\rulelabel{RG-Insert}
\begin{center}
\begin{smathpar}
\begin{array}{c}
\RULE
{
  \stable(\R,P)\\
  \forall\stl,\stl',\stg,i.~P(\stl,\stg) \conj j \not\in
  \dom(\stl\cup\stg) \conj 
  \stl'=\stl \cup 
  \{\langle x \with \idf=j;\,\txnf=i;\,\delf=\C{false}\rangle\} \Rightarrow Q(\stl',\stg)
}
{
  \R \vdash \hoare{P}{\inserte{x}}{Q}
}
\end{array}
\end{smathpar}
\end{center}
\rulelabel{RG-Update}
\begin{center}
\begin{smathpar}
\begin{array}{c}
\RULE
{
  \hspace*{-1.2in} \stable(\R,P)\spc
  \forall\stl,\stl',\stg.~P(\stl,\stg) \conj 
  \stl' = \stl \cup \{r' \,|\, \exists(r\in\Delta).[r/x]e_2 \conj\\
   \hspace*{1.3in}r'=\langle[r/x]e_1 \with \idf=r.\idf;\,\txnf=i;\,\delf=\C{false}\rangle\} \Rightarrow   Q(\stl',\stg)
}
{
  \R \vdash \hoare{P}{\updatee{\lambda x.e_1}{\lambda x.e_2}}{Q}
}
\end{array}
\end{smathpar}
\end{center}
\begin{minipage}{3.2in}
\rulelabel{RG-Delete}
\begin{smathpar}
\begin{array}{c}
\RULE
{
  \stable(\R,P)\\
  \forall\stl,\stl',\stg.~P(\stl,\stg) \conj 
  \stl' = \stl \cup \{r' \,|\, \exists(r\in\Delta).~ [r/x]e
        \conj r'=\langle r \with \txnf=i; \delf=\C{true}\rangle\}
  \Rightarrow 
  Q(\stl',\stg)
}
{
  \R \vdash \hoare{P}{\deletee{\lambda x.e}}{Q}
}
\end{array}
\end{smathpar}
\end{minipage}
\bigskip

\begin{minipage}{2.8in}
\rulelabel{RG-Foreach}
\begin{smathpar}
\begin{array}{c}
\RULE
{
  \stable(\R,Q)\spc
  \stable(\R,\psi) \spc \stable(\R, P)\\
  P \Rightarrow [\emptyset/y]\psi\spc
  \R \vdash \hoare{\psi \wedge z\in x}{c}{Q_c}\\
  Q_c \Rightarrow [y \cup \{z\}/y]\psi\spc
  [x/y]\psi \Rightarrow Q
}
{
  \R \vdash \hoare{P}{\foreache{x}{\lambda y.\lambda z.c}}{Q}
}
\end{array}
\end{smathpar}
\end{minipage}
\begin{minipage}{2in}
\rulelabel{RG-Conseq}
\begin{smathpar}
\begin{array}{c}
\RULE
{
  \rg{I,R}{\ctxn{i}{\I}{c}}{G,I}\\
  \I' \Rightarrow \I \spc 
  R' \subseteq R \spc
  G \subseteq G' \\
  \stable(R',\I')\spc
  \forall \stg,\stg'.~I(\stg) \wedge G'(\stg,\stg') \Rightarrow I(\stg')\\
}
{
  \rg{I,R'}{\ctxn{i}{\I'}{c}}{G',I}
}
\end{array}
\end{smathpar}
\end{minipage}
\bigskip

\rulelabel{RG-Txn}
\begin{center}
\begin{smathpar}
\begin{array}{c}
\RULE
{
  \stable(R,\I)\spc
  \stable(R,I)\spc
  \R_e = R \backslash \I_e \spc 
  \R_c = R \backslash \I_c \spc 
  P(\stl,\stg) \Leftrightarrow \stl=\emptyset \wedge I(\stg)\\
   \R_e \vdash \rg{P}{c}{Q} \spc \stable(\R_c,Q) \spc
  \forall \stl,\stg.~ Q(\stl,\stg) \Rightarrow 
    G(\stg, \stl \rhd \stg)\spc
  \forall \stg,\stg'.~I(\stg) \wedge G(\stg,\stg') \Rightarrow I(\stg')\\
}
{
  \rg{I,R}{\ctxn{i}{\I}{c}}{G,I}
}
\end{array}
\end{smathpar}
\end{center}

\caption{\small \txnimp: Rely-Guarantee rules}
\label{fig:rg-rules}
\vspace*{-12pt}
\end{figure}

Fig.~\ref{fig:rg-rules} shows an illustrative subset of the
rely-guarantee (RG) reasoning rules for $\txnimp$. We define two RG
judgments: top-level ($\rg{I,R}c{G,I}$), and transaction-local ($\R
\vdash \hoare{P}c{Q}$).  Recall that the standard RG judgment is the
quintuple $\rg{P,R}{c}{G,Q}$. Instead of separate $P$ and $Q$
assertions, our top-level judgment uses $I$ as both a pre- and
post-condition, because our focus is on verifying that a
\txnimp\ program \emph{preserves} a databases' consistency
conditions\footnote{The terms \emph{consistency condition},
  \emph{high-level invariant}, and \emph{integrity constraint} are
  used interchangeably throughout the paper.}.  A transaction-local RG
judgment does not include a guarantee relation because
transaction-local effects are not visible outside a transaction. Also,
the rely relation ($\R$) of the transaction-local judgment is not the
same as the top-level rely relation ($R$) because it must take into
account the transaction's isolation specification ($\I$). Intuitively,
$\R$ is $R$ modulo $\I$.  Recall that a transaction writes to its
local database ($\stl$), which is then flushed when the transaction
commits. Thus, the guarantee of a transaction depends on the state of
its local database at the commit point. The pre- and post-condition
assertions ($P$ and $Q$) in the local judgment facilitate tracking the
changes to the transaction-local state, which eventually helps us
prove the validity of the transaction's guarantee.  Both $P$ and $Q$
are bi-state assertions; they relate transaction-local database state
($\stl$) to the global database state ($\stg$). Thus, the
transaction-local judgment effectively tracks how transaction-local
and global states change in relation to each other.

\subsubsection{Stability}

A central feature of a rely-guarantee judgment is a stability
condition that requires the validity of an assertion $\phi$ to be
unaffected by interference from other concurrently executing
transactions, i.e., the rely relation $R$. In conventional RG,
stability is defined as follows, where $\sigma$ and $\sigma'$ denote
states:
\begin{smathpar}
\begin{array}{lcl}
\stable(R,\phi) & \Leftrightarrow & \forall \sigma,\sigma'.~
\phi(\sigma) \conj R(\sigma,\sigma') \Rightarrow \phi(\sigma')\\
\end{array}
\end{smathpar}
Due to the presence of local and global database states, and the
availability of an isolation specification, we use multiple
definitions of stability in Fig.~\ref{fig:rg-rules}, but they all
convey the same intuition as above. In our setting, we only need to
prove the stability of an assertion ($\phi$) against those environment
steps which lead to a global database state on which the transaction
itself can take its next step according to its isolation specification
($\I$). 
\begin{smathpar}
\begin{array}{lcl}
\stable(R, \phi) & \Leftrightarrow & \forall \stl, \stg, \stg'. \phi(\stl, \stg) 
  \wedge R^{*}(\stg, \stg') \wedge \I(\stl, \stg, \stg') \Rightarrow \phi(\stl, \stg')
\end{array}
\end{smathpar}
\noindent A characteristic of RG reasoning is that stability of an
assertion is always proven w.r.t to $R$, and not $R^{*}$, although
interference may include multiple environment steps, and $R$ only
captures a single step. This is nonetheless sound due to inductive
reasoning: if $\phi$ is preserved by every step of $R$, then $\phi$ is
preserved by $R^{*}$, and vice-versa.  However, such reasoning does
not extend naturally to isolation-constrained interference because
$R^{*}$ modulo $\I$ is not same as $\R^{*}$; the former is a
transitive relation constrained by $\I$, whereas the latter is the
transitive closure of a relation constrained by $\I$. This means,
unfortunately, that we cannot directly replace $R^{*}$ by $R$ in the
above condition.

To obtain an equivalent form in our setting, we require an additional
condition on the isolation specification, which we call the
\emph{stability condition on $\I$}.  The condition requires $\I$ to
admit the interference of multiple $R$ steps (i.e.,
$R^{*}(\stg,\stg'')$, for two database states $\stg$ and $\stg''$),
only if it also admits interference of each $R$ step along the way.
Formally:
\begin{smathpar}
\begin{array}{lcl}
  \stable(R,\I) & \Leftrightarrow & \forall \stl,\stg,\stg',\stg''.~
  \I(\stl,\stg,\stg'') \conj R(\stg',\stg'') \Rightarrow
  \I(\stl,\stg,\stg') \conj \I(\stl,\stg',\stg'')
\end{array}
\end{smathpar}
It can be easily verified that the above stability condition is
satisfied by the isolation axioms from Sec.~\ref{sec:isolation}. For
instance, $\I_{ss}$, the snapshot axiom, is stable because if a the
state is unmodified between $\stg$ and $\stg''$, then it is clearly
unmodified between $\stg$ and $\stg'$, and also between $\stg'$ and
$\stg''$, where $\stg'$ is an intermediary state.  Modifying and
restoring the state $\stg$ is not possible because each new commit
bears a new transaction id different from the transaction ids (\C{txn}
fields) present in $\stg$. 

The stability condition on $\I$ guarantees that an interference from
$R^*$ is admissible only if the interference due to each individual
$R$ step is admissible. In other words, it makes isolation-constrained
$R^*$ relation equal to the transitive closure of the
isolation-constrained $R$ relation. We call $R$ constrained by
isolation $\I$ as $R$ modulo $\I$ ($R\backslash \I$; written equivalently
as $\R$), which is the following ternary relation:
\begin{smathpar}
  \begin{array}{lcl}
    (R \backslash \I)(\stl, \stg, \stg') & \Leftrightarrow & R(\stg,
      \stg') \wedge \I(\stl, \stg, \stg')\\
  \end{array}
\end{smathpar}
It is now enough to prove the stability of an RG assertion $\phi$
w.r.t $R\backslash \I$:
\begin{smathpar}
  \begin{array}{lcl}
    \stable((R \backslash \I), \phi) & \Leftrightarrow & \forall
      \stl,\stg,\stg'.~ \phi(\stl,\stg) \wedge (R \backslash \I)(\stl,\stg,\stg') \Rightarrow \phi(\stl,\stg')\\
  \end{array}
\end{smathpar}
This condition often significantly simplifies the form of $R
\backslash \I$ irrespective of $R$. For example, when a transaction is
executed against a snapshot of the database (i.e. $\mathbb{I}_{ss}$),
$R \backslash \I_{ss}$ will be the identity function, since any
non-trivial interference will violate the $\stg' = \stg$ condition
imposed by $\I_{ss}$.

\subsubsection{Rules}

\rulelabel{RG-Txn} is the top-level rule that lets us prove a
transaction preserves the high-level invariant $I$ when executed under
the required isolation as specified by $\I$. It depends on a
transaction-local judgment to verify the body ($c$) of a transaction
with id $i$. The precondition $P$ of $c$ must follow from the fact
that the transaction-local database ($\stl$) is initially empty, and
the global database satisfies the high-level invariant $I$. The rely
relation ($\R_e$) is obtained from the global rely relation $R$ and
the isolation specification $\I_e$ as explained above. Recall that
$\I_e$ constrains the global effects visible to the transaction while
it is executing but has not yet committed, and $P$ and $Q$ of the
transaction-local RG judgment are binary assertions; they relate local
and global database states. The local judgment rules require one or
both of them to be stable with respect to the constrained rely
relation $\R_e$.

\raggedbottom
For the guarantee $G$ of a transaction to be valid, it must follow
from the post-condition $Q$ of the body, provided that $Q$ is stable
w.r.t the commit-time interference captured by $\R_c$. $\R_c$, like
$\R_e$, is computed as a rely relation modulo isolation, except that
commit-time isolation ($\I_c$) is considered. The validity of
$G$ is captured by the following implication:
\begin{smathpar}
\begin{array}{c}
  \forall \stl,\stg.~ Q(\stl,\stg) \Rightarrow G(\stg, \stl \rhd \stg)\spc
\end{array}
\end{smathpar}
In other words, if $Q$ relates the transaction-local database state
($\stl$) to the state of the global database ($\stg$) before a transaction
commits, then $G$ must relate the states of the global database before
and after the commit. The act of commit is captured by the flush
action ($\stl\rhd\stg$). Once we establish the validity of $G$ as a
faithful representative of the transaction, we can verify that the
transaction preserves the high-level invariant $I$ by checking the
stability of $I$ w.r.t $G$, i.e., $\forall \stg,\stg'.~I(\stg) \wedge
G(\stg,\stg') \Rightarrow I(\stg')$.

The \rulelabel{RG-Conseq} rule lets us safely weaken the guarantee
$G$, and strengthen the rely $R$ of a transaction. Importantly, it
also allows its isolation specification $\I$ to be strengthened (both
$\I_e$ and $\I_c$). This means that a transaction proven correct under
a weaker isolation level is also correct under a stronger level.
Parametricity over the isolation specification $\I$, combined with the
ability to strengthen $\I$ as needed, admits a flexible proof strategy
to prove database programs correct. For example, programmers can
declare isolation requirements of their choice through $\I$, and then
prove programs correct assuming the guarantees hold. The soundness of
strengthening $\I$ ensures that a program can be safely executed on
any system that offers isolation guarantees at least as strong as
those assumed.

Salient rules of transaction-local RG judgments are shown in
Fig.~\ref{fig:rg-rules}. These rules (\rulelabel{RG-Update},
\rulelabel{RG-Select}, \rulelabel{RG-Delete}, and
\rulelabel{RG-Insert}) reflect the structure of the corresponding
reduction rule from Fig.~\ref{fig:txnimp}.  The rule
\rulelabel{RG-Foreach} defines the RG judgment for a \C{FOREACH} loop.
As is characteristic of loops, the reasoning is pivoted on a loop
invariant $\psi$ that needs to be stable w.r.t $\R$. $\psi$ must be
implied by $P$, the pre-condition of \C{FOREACH}, when no elements
have been iterated, i.e, when $y=\emptyset$. The body of the loop can
assume the loop invariant, and the fact that $z$ is an element from
the set $x$ being iterated, to prove its post-condition $Q_c$. The
operational semantics ensures that $z$ is added to $y$ at the end of
the iteration, hence $Q_c$ must imply $[y\cup\{z\}/y]\psi$. When the
loop has finished execution, $y$, the set of iterated items, is the
entire set $x$. Thus $[x/y]\psi$ is true at the end of the loop, from
which the post-condition $Q$ must follow. As with the other rules, $Q$
needs to be stable. The rules for conditionals, sequencing etc., are
standard, and hence elided.

\subsection{Semantics and Soundness}

We now formalize the semantics of the RG judgments defined in
Fig.~\ref{fig:rg-rules}, and state their soundness guarantees.

\begin{definition}[\bfseries Interleaved step and multi-step relations]
Interleaved step relations interleave global and transaction-local
reductions with interference as captured by the corresponding rely
relations. They are defined thus:
\begin{smathpar}
\begin{array}{lclr}
(c,\stg) \rstepsto (c',\stg') & \defeq &  
  (c,\stg) \stepsto (c',\stg') \disj (c' = c \conj R(\stg, \stg'))&
  \texttt{global}\\
(\tbox{c}_i,\stl,\stg) \rstepsto (\tbox{c'}_i,\stl',\stg') & \defeq & \stg \vdash 
  (\tbox{c}_i,\stl) \stepsto (\tbox{c'}_i,\stl') \conj \stg'=\stg& \texttt{transaction-local}\\
  &   & \disj (c' = c \conj \stl'=\stl \conj \R(\stl, \stg, \stg'))
\end{array}
\end{smathpar}

\noindent An interleaved multi-step relation ($\stepssto{n}$) is the
reflexive transitive closure of the interleaved step relation.  
\end{definition}

\begin{definition}[\bfseries Semantics of RG judgments]
\label{def:rg-semantics}
The semantics of the global and transaction-local RG judgments are
defined thus:
\begin{smathpar}
\begin{array}{lclr}
\R \vdash \hoare{P}{c}{Q} & \defeq & \forall
  \stl,\stl',\stg,\stg'.~ P(\stl,\stg) \conj (\tbox{c}_i,\stl,\stg) \rstepssto{n}
  (\tbox{\cskip}_i, \stl',\stg')
  \Rightarrow Q(\stl',\stg') &\\
\rg{I,R}{c}{G,I} & \defeq &  \forall \stg.\, I(\stg)
  \Rightarrow (\forall \stg'.\; (c,\stg) 
    \rstepssto{n} (\cskip,\stg') \Rightarrow I(\stg')) \\
&&\hspace*{0.6in}\conj \texttt{txn-guaranteed}(R,G,c,\stg)\\
\end{array}
\end{smathpar}

\noindent The
$\texttt{txn-guaranteed}$ predicate used in the semantics of the
global RG judgment is defined below:\vspace*{-10pt}

\begin{smathpar}
\begin{array}{lcl}
\texttt{txn-guaranteed}(R,G,c,\stg) &\defeq& \forall c',c''\stg',\stg''.
(c,\stg) \rstepssto{n} (c',\stg') \conj (c',\stg') \stepsto
  (c'',\stg'') \Rightarrow G(\stg',\stg'')\\
\end{array}
\end{smathpar}
\end{definition}

\noindent Thus, if $\rg{I,R}{c}{G,I}$ is a valid RG judgment, then (a)
every interleaved multi-step reduction of $c$ preserves the database
integrity constraint (consistency condition) $I$, and (b) the effect
that every transaction in $c$ has on the database state is captured by
$G$.
\noindent We can now assert the soundness of the RG judgments in
Fig.~\ref{fig:rg-rules} as follows\footnote{Full proofs for the major
  theorems and lemmas defined in this paper are available from
  ~\cite{KN+18_arxiv}.}:

\begin{theorem}[\bfseries Soundness] 
The rely-guarantee judgments defined by the rules in
Fig.~\ref{fig:rg-rules} are sound with respect to the semantics of
Definition~\ref{def:rg-semantics}.
\end{theorem}

\paragraph{{\sc Proof Sketch.}}  The most important rule is the top-level
rule \rulelabel{RG-Txn}, which proves that a transaction $c$ which
begins its execution in global database state satisfying $I$ and
encountering interference $R$ while executing under isolation
specification $\I$ finishes its execution in a database state also
satisfying $I$, and also guarantees that its commit step satisfies
$G$. The rule uses the transaction-local RG judgment $\R_e \vdash
\rg{P}{c}{Q}$. By \rulelabel{E-Txn-Start}, the local and global
database states at the start of a transaction satisfy $P$, and the
only challenge is that environment steps in an execution covered by
$\R_e \vdash \rg{P}{c}{Q}$ are in $\R_e$, while the top-level
judgment requires environment steps in $R$. We show that it is enough
to consider only those environment steps in $\R_e$. First, we use an
inductive argument and stability of $\I_e$ ($\stable(R, \I_e)$) to
show that any execution in which the transaction completes all its
steps must always preserve the isolation specification $\I_e$ after
every environment step. Intuitively, this is because once $\I_e$ gets
broken after some environment step, it will continue to remain broken
and the transaction would not be able to proceed (according to
\rulelabel{E-Txn}). Since $\R_e$ contains exactly those environment
steps which preserve $\I_e$, the local-level RG judgment can be
soundly used, which guarantees that after the transaction finishes its
execution, its local state $\stl$ and global state $\stg$ will satisfy
the assertion $Q$. Environment steps between the last step of the
transaction and its commit step can modify the global state, and hence
we also require $Q$ to be stable against $R$. Again, we use an
inductive argument, the stability of $\I_c$, and the fact that the
transaction must execute its commit step to show that all environment
steps must preserve $\I_c$, and hence it is enough to require
$\stable(\R_c, Q)$. $Q$ guarantees that the commit step is in $G$, and
$G$ in turn guarantees that after execution, the global database state
will obey the invariant $I$.



\section{Inference}
\label{sec:inference}

The rely-guarantee framework presented in the previous section
facilitates modular proofs for weakly-isolated transactions, but
imposes a non-trivial annotation burden.  In particular, it requires
each statement ($c$) of the transaction to be annotated with a stable
pre- ($P$) and post-condition ($Q$), and loops to be annotated with
stable inductive invariants ($\psi$). While weakest pre-condition
style predicate transformers can help in inferring intermediate
assertions for regular statements, loop invariant inference remains
challenging, even for the simple form of loops considered here.  As an
alternative, we present an inference algorithm based on state
transformers that alleviates this burden.  The idea is to infer the
logical effect that each statement has on the transaction-local
database state $\stl$ (i.e., how it transforms $\stl$), and compose
multiple such effects together to describe the effect of the
transaction as a whole.  Importantly, this approach generalizes to
loops, where the effect of a loop can be inferred as a well-defined
function of the effect of its body, thanks to certain pleasant
properties enjoyed by the database programs modeled by our core
language.  Interpreting database semantics as functional
transformations on sets (described in terms of their logical effects)
enables an inference mechanism that can leverage off-the-shelf SMT
solvers for automated verification.

\begin{figure}[t]
\begin{smathpar}
\renewcommand{\arraystretch}{1.2}
\begin{array}{lcl}
\multicolumn{3}{c}{
  x,y,\stl,\stg \in \texttt{Variables}\spc
  \varphi \in \Prop^{0}\spc
  \phi \in \Prop^{1}
}\\
s & \coloneqq & x \ALT \stl \ALT \stg \ALT \{x \,|\, \varphi\} 
  \ALT \existsl(\stg,\phi,s) 
  \ALT s_1 \bind \lambda x. s_2 \ALT \itel{\varphi}{s_1}{s_2} 
  \ALT s_1 \cup s_2 \\
\end{array}
\end{smathpar}
\caption{Syntax of the set language $\SL$}
\label{fig:logic-syntax}
\end{figure}

At the core of our approach is a simple language ($\SL$) to express
set transformations (see Fig.~\ref{fig:logic-syntax}). The language
admits set expressions that include variables ($x$), literals of the
form $\{x \,|\, \varphi\}$ where $\varphi$ is a propositional
(quantifier-free) formula on $x$, a restricted form of existential
quantification that binds a set $\stg$ satisfying proposition $\phi$
in a set expression $s$, a monadic composition of two set expressions
($s_1$ and $s_2$) composed using a bind ($\bind$) operation, a
conditional set expression where the condition is a propositional
formula, and a union of two set expressions. Symbols $\stl$ and $\stg$
are also variables in $\SL$, but are used to denote local and database
states (also represented as sets), respectively. Constant sets can be
written using set literal expressions. For example, the set $\{1,2\}$
can be written as $\{x \,|\, x=1 \vee x=2\}$. The language is
carefully chosen to be expressive enough to capture the semantics of
$\txnimp$ statements (as well as SQL operations more generally), yet
simple enough to have a semantics-preserving translation amenable for
automated verification.

\begin{figure}
\raggedright
\fbox {\(\Fx \vdash c \elabsto \F \)}\\
%


\begin{center} 
%
\begin{smathpar}
\begin{array}{c}
\RULE
{
}
{
  \Fx \vdash \inserte{x} ~\elabsto~ \stabilize
    {\inctxt{\Fx}
            {\lambda(\stg).~ \{ r \,|\, r = 
              \{\langle x \with \delf=\mathtt{false};\, 
                \txnf = i\rangle \}}}
}
\end{array}
\end{smathpar}

\begin{smathpar}
\begin{array}{c}
\RULE
{
  G = \lambda r.~ \itel{[r/x]e_2}
      {\{r' \,|\, r' = \langle[r/x]e_1 \with \idf=r.\idf;\,\delf = r.\delf;\,\txnf=i\rangle \}}
      {\emptyset}
}
{
  \Fx \vdash \updatee{\lambda x.e_1}{\lambda x.e_2}  ~\elabsto~
  \stabilize{\inctxt{\Fx}{\lambda(\stg).~ \stg \bind G }}
}
\end{array}
\end{smathpar}

\begin{smathpar}
\begin{array}{c}
\RULE
{
  G = \lambda r.~ \itel{[r/x]e}
      {\{r' \,|\, r' = \langle r \with \delf = \mathtt{true};\,\txnf=i\rangle \}}
      {\emptyset}
}
{
  \Fx \vdash \deletee{\lambda x.e}  ~\elabsto~
  \stabilize{\inctxt{\Fx}{\lambda(\stg).~ \stg \bind G }}
}
\end{array}
\end{smathpar}

\begin{smathpar}
\begin{array}{c}
\RULE
{
  \Fx \vdash c \elabsto \F \spc
}
{
  \Fx \vdash \lete{x}{e}{c} ~\elabsto~  
    {\lambda(\stg).~ [e/x]\,\F(\stg)}\\
}
\end{array}
\end{smathpar}

\begin{smathpar}
\begin{array}{c}
\RULE
{
  \Fx \vdash c \elabsto \F \\
  G = \lambda r.~ \itel{[r/x]e}{\{r' \,|\, r' = r \}}{\emptyset}\spc
  \F' = \stabilize{\inctxt{\Fx}{\lambda(\stg).~\stg \bind G}}

}
{
  \Fx \vdash \lete{y}{\selecte{\lambda x.e}}{c} ~\elabsto~  
    \lambda (\stg).~ [\F'(\stg)/y]\,\F(\stg)\\
}
\end{array}
\end{smathpar}

\begin{smathpar}
\begin{array}{c}
\RULE
{
  \Fx \vdash c_1 \elabsto \F_1 \spc
  \Fx \vdash c_2 \elabsto \F_2 
}
{
  \Fx \vdash \ite{e}{c_1}{c_2} ~\elabsto~
    \lambda(\stg).~\itel{e}{\F_1(\stg)}{\F_2(\stg)}\\
}
\end{array}
\end{smathpar}

\begin{smathpar}
\begin{array}{c}
\RULE
{
  \Fx \vdash c_1 \elabsto \F_1 \spc
  \Fx \cup \F_1 \vdash c_2 \elabsto \F_2 
}
{
  \Fx \vdash c_1;c_2 ~\elabsto~ \F_1 \cup \F_2 
}
\end{array}
\end{smathpar}
%

%
\begin{smathpar}
\begin{array}{c}
\RULE
{
  \Fx \vdash c \elabsto \F \spc
}
{
  \Fx \vdash \foreache{x}{\lambda y.\lambda z.~c} ~\elabsto~
  \lambda(\stg).~ x\bind(\lambda z.~\F(\stg))
}
\end{array}
\end{smathpar}

\end{center}

\caption{\txnimp: State transformer semantics. }
\label{fig:inference-rules}
\end{figure}

Fig.~\ref{fig:inference-rules} shows the syntax-directed state
transformer inference rules for $\txnimp$ commands inside a
transaction $\C{TXN}_i$. The rules compute, for each command $c$, a
(meta) function $\F$ that returns a set of records as an expression in
$\SL$, given a global database $\stg$. Intuitively, $\F(\stg)$
abstracts the set of records added to the local database $\stl$ as a
result of executing $c$ under $\stg$ (i.e., $\stg \vdash
(\tbox{c}_i,\stl) \stepsto^{*}_{R} (\tbox{\cskip}_i, \stl \cup
\F(\stg))$)\footnote{Recall that the operational semantics treats 
  deletion of records as the addition of the deleted record with its
  \C{del} field set to true in the local store.}. Note that the
  function $\F$ we call state transformer here is actually the
  \emph{effect} part of the state transformer introduced in
  Sec.~\ref{sec:motivation}, which is a function $\T$ of form
  $\lambda(\stl,\stg).~\stl \cup \F(\stg)$. Nonetheless, for
  simplicity, we will continue to refer to $\F$ as state transformer.
  Since the execution is subject to isolation-constrained
  interference, the inference judgment depends on the
  isolation-constrained rely relation $\R$, which is used to enforce
  the stability of the state transformer $\F$.  Recall that $\R$ is a
  tri-state rely relation over $\stl$, $\stg$ and $\stg'$, that admits
  an interference from $\stg$ and $\stg'$ depending on the local
  database state $\stl$. Thus, the stability of the state transformer
  $\F$ of $c$ with respect to $\R$ needs to take into account the
  (possible) prior state of the local database $\stl$, which depends
  on the context (sequence of previous commands) of $c$, and computed
  by the corresponding state transformer $\Fx$. Thus, the semantics of
  the state transformer can be understood in terms of the RG judgment
  as following (formalized as Theorem~\ref{thm:inference-sound} in
  Sec.~\ref{sec:inference-sound}):
\begin{smathpar}
  \begin{array}{c}
    \R \vdash \hoare{\lambda(\stl,\stg).~ \stl = \Fx(\stg)}{c}
    {\lambda(\stl,\stg).~ \stl = \Fx(\stg) \cup \F(\stg)}
  \end{array}
\end{smathpar}
In the above RG judgment, let $P$ denote the pre-condition
$\lambda(\stl,\stg).~ \stl = \Fx(\stg)$, and let $Q$ denote the
post-condition $\lambda(\stl,\stg).~ \stl = \Fx(\stg) \cup \F(\stg)$.
The stability condition on the state transformer $\F$ can be derived
from the stability condition on $Q$. Observe that for $Q$ to be
stable, $\Fx(\stg') \cup \F(\stg')$ must be equal to $\Fx(\stg) \cup
\F(\stg)$, where $\stg$ and $\stg'$ are related by $R$ (ignore $\I$
for the moment). Assuming that $P$ is stable, $\Fx(\stg')=\Fx(\stg)$
is already given, leaving $\F(\stg')=\F(\stg)$ to be enforced. Thus,
the stability of $\F$ in in the context of $\Fx$ (written
$\inctxt{\Fx}{\F}$) is defined as following:
\begin{smathpar}
\begin{array}{lcl}
  \stable(\R,\inctxt{\Fx}{\F}) & \Leftrightarrow & \forall \stg,\stg',\nubar.~
  \R(\Fx(\stg) \cup \F(\stg),\stg,\stg') \Rightarrow \F(\stg) = \F(\stg')
\end{array}
\end{smathpar}
where $\nubar$ are the variables that occur free in $\F$; this is
possible because of how the inference rules are structured. The
equality in $\SL$ translates to equivalence in first-order logic, as
we describe later. In the inference rules, stability is enforced
constructively by a meta-function $\stabilize{\cdot}$, which accepts
a transformer $\F$ (in its context $\Fx$) and returns a new
transformer that is guaranteed to be stable under $\R$.
$\stabilize{\cdot}$ achieves the stability guarantee by abstracting
away the bound global state ($\stg$) in an unstable $\F$ to an
existentially bound $\stg'$ as described below:
\begin{smathpar}
\begin{array}{lcll}
  \stabilize{\inctxt{\Fx}{\F}} & = & \F & \texttt{if }
  \stable(\R,\inctxt{\Fx}{\F}).\\
  & = & \lambda (\stg).~\existsl(\stg',I(\stg'),\F(\stg')) 
      & \texttt{otherwise. }\stg'\texttt{ is a fresh name.}\\
\end{array}
\end{smathpar}
Observe that when $\F$ is not stable, $\stabilize{\F}$ returns a
transformer $\F'$ that simply ignores its $\stg$ argument in favor of
a generic $\stg'$, making $\F'$ trivially stable. It is safe to assume
$I(\stg')$ because all verified transactions must preserve the
invariant, and hence only valid database states will ever be
witnessed. From the perspective of RG reasoning, $\stabilize{\cdot}$
effectively weakens the post-condition of a statement, as done by the
\rulelabel{RG-Conseq} rule for transaction-bound commands.  The
weakening semantics chosen by $\stabilize{\cdot}$, while being simple,
is nonetheless useful because of the $I(\stg')$ assumption imposed on
an existentially bound $\stg'$. The example in
Fig.~\ref{fig:weakening-example} demonstrates.
\begin{figure}[h]
\begin{center}
\begin{ocaml}
let add_interest acc_id pc = atomically_do @@ fun () ->
  let a = SQL.select1 BankAccount (fun acc -> acc.id = acc_id) in
  let y = a.bal + pc*a.bal in
  SQL.update BankAccount (fun acc -> {acc with bal = acc.bal + y})
                         (fun acc -> acc.id = acc_id)
\end{ocaml}
\end{center}
\caption{A transaction that deposits an interest to a bank account.}
\label{fig:weakening-example}
\end{figure}
Here, an \C{add\_interest} transaction adds a positive interest
(determined by \C{pc}) to the balance of a bank account, which is
required to be non-negative ($I(\stg) \Leftrightarrow
\forall(r\in\stg).~r.\C{bal}\ge 0$). The transaction starts by issuing
a \C{select1} query, whose transformer $\F$ is essentially a singleton
set containing a record $r$ whose id is \C{acc\_id} (i.e.,
$\F(\stg) = \{r \,|\, r\in\stg \wedge r.\C{id} = \C{acc\_id}\}$).
However, $\F$ is unstable because $\F(\stg')$ may not be
the same set as $\F(\stg)$ when $\stg'\neq\stg$. A record
$r\in\stg$ whose $\C{id}=\C{acc\_id}$ may have its balance updated by
a concurrent \C{withdraw} or \C{deposit} transaction in $\stg'$,
making the record in $\stg'$ different from the record in $\stg$.
Hence the stability check fails.  Fortunately, the weakening operator
($\stabilize{\cdot}$) allows us to weaken the effect to
$\existsl(\stg, I(\stg), \{ r \,|\, r\in\stg \wedge
r.\C{id}=\C{acc\_id}\})$, which effectively asserts that the
\C{select1} query returns a record with $\C{id}=\C{acc\_id}$ from
\emph{some} database state that satisfies the non-negative balance
invariant $I$.  This weakened assertion is nonetheless enough to
deduce that $\C{a.bal}\ge0$, and subsequently prove that $\C{a.bal +
pc*a.bal}\ge 0$, allowing us to verify the \C{add\_interest}
transaction.

The state transformer rules, like the earlier RG rules, closely follow
the corresponding reduction rules in Fig.~\ref{fig:txnimp}, except
that their language of expression is $\SL$. For instance, while the
reduction rule for \C{UPDATE} declaratively specifies the set of
updated records, the state transformer rule uses $\SL$'s bind
operation to \emph{compute} the set. Other SQL rules do likewise. The
rules for \C{LET} binders, conditionals, and sequences compose the
effects inferred for their subcommands. Thus, the effect of a sequence
of commands $c_1;c_2$ is the union of effects $\F_1$ and $\F_2$ of
$c_1$ and $c_2$, respectively, except that $\F_2$ is computed in a
context that includes $\F_1$ (we write $\F_1 \cup \F_2$ as a shorthand
for $\lambda(\stg).~\F_1(\stg) \cup \F_2(\stg)$). The inference rule
for \C{FOREACH} takes advantage of the $\SL$'s bind operator to lift
the effect inferred for the loop body to the level of the loop. Since
records added to $\stl$ in each iteration of \C{FOREACH} are
independent of the previous iteration (recall that we make a local
context independence assumption about database programs; Sec.
~\ref{sec:opsem}), sequential composition of the effects of different
iterations is the same as their parallel composition. Since the loop
body is executed once per each $z\in x$, the effect of the the loop is
a union of effects ($\F$) for all $z\in x$, all applied to the same
state ($\stg$).  That is, $\F_{loop}(\stg) = \bigcup_{z\in
x}\F_{body}(\stg)$. From the definition of the set monad's bind
operator, $\F_{loop}(\stg) = x \bind (\lambda z.~F_{body}(\stg))$,
which mirrors the definition of the rule.

\subsection{Soundness of Inference}
\label{sec:inference-sound}

We now formally state the correspondence between the inference rules
given above and the RG judgment of \S\ref{sec:reasoning}:
\begin{theorem}
  \label{thm:inference-sound}
  For all $i$,$R$,$I$,$c$,$\Fx$, $\F$, if $\stable(\R,I)$, $\stable(\R, \Fx)$ and $\Fx
  \vdash c \elabsto \F$, then:\\\vspace*{-0.2cm}
  \begin{smathpar}
  \begin{array}{c}
    \R \vdash \hoare{\lambda(\stl,\stg).~\stl=\Fx(\stg) \conj
    I(\stg)}{c}{\lambda(\stl,\stg).~\stl = \Fx(\stg) \cup \F(\stg)}
  \end{array}
  \end{smathpar}
\end{theorem}

\paragraph{{\sc Proof Sketch.}} 
The proof follows by structural induction on $c$. Let $P
=\lambda(\stl,\stg).~\stl=\Fx(\stg) \conj I(\stg)$ and $Q
=\lambda(\stl,\stg).\stl = \Fx(\stg) \cup \F(\stg)$. The base cases
correspond to \C{INSERT}, \C{UPDATE} and \C{DELETE} statements, where
the proof is straightforward. The proofs for \C{SELECT}, sequencing,
and conditionals use the inductive hypothesis to infer the
RG-judgments present in the premises of their corresponding RG-rules.
The interesting case is the \C{FOREACH} statement, for which we use
the loop invariant $\psi(\stl, \stg) \Leftrightarrow \stl = \Fx(\stg)
\cup (y \bind (\lambda z.~\F(\stg)))$, (where assuming that $c$ is the
body of the loop, $c \elabsto \F$) to prove all the premises of
\rulelabel{RG-Foreach}. Using the same notation as the rule
$\rulelabel{RG-Foreach}$, $y$ refers to the records already processed
in previous iterations of the loop, while $z$ refers to the record
being processed in the current iteration.  At the beginning of the
loop $[\phi/y]\psi(\stl, \stg)$ just reduces to $\stl = \Fx(\stg)$
which is implied by the pre-condition $P$. From the inductive
hypothesis, we can infer that each iteration corresponds to the
application of $\F$. Since all iterations are assumed to be
independent of each other, and $z$ is bound to a record in $x$ for
each iteration, we conclude that at the end of every iteration, the
loop invariant $[y \cup \{z\}/y]\psi$ will be satisfied.

\subsection{From $\SL$ to the first-order logic}

\begin{figure}
\begin{smathpar}
\begin{array}{lclc@{\hspace*{-30pt}}r}
  \mssemof{\nubar}{\stl \ALT \stg \ALT \ldots} & = &  
  (\top, \lambda (\vbar,r).\,r\in\stl) \ALT 
  (\top, \lambda (\vbar,r).\,r\in\stg) \ALT \ldots & \texttt{  }
  & |\vbar|=|\nubar|\\
%
  \mssemof{\nubar}{\{x \,|\, \varphi\}} & = & (\top, 
    \lambda (\vbar,r).\,[r/x]\varphi) & \texttt{  }
  & |\vbar|=|\nubar|\\
\mssemof{\nubar}{\itel{\varphi}{s_1}{s_2}} & = & (\phi_1 \wedge \phi_2,~
    \lambda (\vbar,r).\, \C{if}\;\varphi\;\C{then}\; \G_1(\vbar,r) 
  & \texttt{  }& (\phi_1,\G_1)=\mssemof{\nubar}{s_1}\\
& & \hspace*{1in}\;\C{else}\; \G_2(\vbar,r) & \texttt{  } 
  & (\phi_2,\G_2)=\mssemof{\nubar}{s_2} \\
  \mssemof{\nubar}{s_1 \cup s_2} & = & (\phi_1 \wedge \phi_2,~ &
    \texttt {  } & (\phi_1,\G_1)=\mssemof{\nubar}{s_1} \\
& & \hspace*{2em}\lambda (\vbar,r).\, \G_1(\vbar,r) \vee \G_2(\vbar,r)) & \texttt{ } 
  & (\phi_2,\G_2)=\mssemof{\nubar}{s_2} \\
\mssemof{\nubar}{s_1 \bind \lambda x. s_2} & = & (\phi_1 \wedge \phi_2
  \wedge \forall \nubar.\forall a.\forall b. ~\pi_1(\nubar)
  \Leftrightarrow & \texttt{  } &  
  \fresh(\pi_1) \spc \fresh(\pi_2) \spc \fresh(g)\\
& & \hspace*{0.65in}\G_1(\nubar,a) \wedge \G_2(\nubar,a,b)
    \Rightarrow g(\nubar,b) & \texttt{  }
  & (\phi_1,\G_1)=\mssemof{\nubar}{s_1} \\
& & \hspace*{0.5in}\wedge \forall \nubar.\forall b.\exists a.~ 
    \pi_2(\nubar) \Leftrightarrow  & \texttt{  } 
  & (\phi_2,\G_2)=\mssemof{\nubar,a}{[a/x]s_2}\\
& & \hspace*{0.65in}g(\nubar,b) \Rightarrow 
    \G_1(\nubar,a) \wedge \G_2(\nubar,a,b), & \texttt{  } 
  & \fresh(a) \spc \fresh(b) \spc\\
& & ~\lambda (\vbar,r).\,\pi_1(\vbar) \conj \pi_2(\vbar) \conj g(\vbar,r))
  & \texttt{  } & |\vbar| = |\nubar|\\
\mssemof{\nubar}{\existsl(\stg,\phi,s)} & = & (\phi_s \wedge 
  \forall \nubar.\forall a.\forall b.~ f(\nubar,a) \wedge 
        f(\nubar,b) \Rightarrow a = b& \texttt{  }
  & \fresh(a) \spc \fresh(b)\\
& &\hspace*{0.2in} \wedge \forall \nubar.\exists a.~ f(\nubar,a) & \texttt{  }
  & \fresh(f) \\
& &\hspace*{0.2in}\wedge \forall \nubar.\forall a.\forall b.~ \pi(\nubar) 
  \Leftrightarrow f(\nubar,a) \wedge [a/\stg]\phi & \texttt{  } 
  & \fresh(\pi)\spc\fresh(g)\\
  & & \hspace*{1.25in}\wedge
  g(\nubar,b)=\G_s(\nubar,b), & \texttt{  } 
  & (\phi_s,\G_s)=\mssemof{\nubar}{[a/\stg]s} \\
& & ~\lambda (\vbar,r).\,\pi(\vbar) \conj g(\vbar,r)) & \texttt{  }
  & |\vbar| = |\nubar|\\
%
\end{array}
\end{smathpar}

\caption{Encoding $\SL$ in first-order logic}
\label{fig:logic}
\end{figure}

Theorem~\ref{thm:inference-sound} lets us replace the local judgment
of the \rulelabel{RG-Txn} rule (Fig.~\ref{fig:rg-rules}) by a state
transformer inference judgment. The soundness of a transaction's
guarantee can now be established w.r.t the effect $\F$ of the body.
The \rulelabel{RG-Txn} rule so updated is shown below ($\Fempty =
\lambda(\stg).~\emptyset$ denotes an empty context):
\begin{smathpar}
\begin{array}{c}
\RULE
{
   \stable(R,\I)\spc
   \stable(R,I)\spc
   \R_e = R \backslash \I_e \spc 
   \R_c = R \backslash \I_c \spc
   \Fempty \vdash c \Longrightarrow_{\langle i,\R_e,I \rangle}\F \\
   \stable(\R_c,\inctxt{\Fempty}{\F}) \spc
   \forall \stg.~ G(\stg,\F(\stg)) \spc
   \forall \stg,\stg'.~I(\stg) \wedge G(\stg,\stg') \Rightarrow I(\stg')\\
}
{
  \rg{I,R}{\ctxn{i}{\I}{c}}{G,I}
}
\end{array}
\end{smathpar}
Automating the application of the \rulelabel{RG-Txn} rule for a
transaction requires automating the multiple implication checks in
the premise. While $R$, $G$, $\I$ and $I$ are formulas in
first-order logic (FOL) with a relatively simple structure, $\F$
is an expression in the set language $\SL$
(Fig.~\ref{fig:logic-syntax}) with a possibly complex structure.
Fortunately, however, there exists a semantics-preserving translation
from $\SL$ to a restricted subset of first-order logic (FOL) that
lends itself to automatic reasoning. 

The algorithm ($\mssemof{\cdot}{\cdot}$) shown in Fig.~\ref{fig:logic}
translates an $\SL$ expression ($s$) to FOL. The translation is based
on encoding a set of element type $T$ as a unary predicate on $T$.
The predicate is represented as a meta function that accepts an $x:T$
and returns a quantifier-free proposition that evaluates to true
($\top$) if and only if $x$ is present in the set. Alternatively,
the translation may also encode the set as a predicate in the logic
itself, in which case a quantified proposition constraining the
predicate is also generated. For instance, consider the set $\{1,2\}$.
The predicate describing the set can be encoded as the function
$\lambda \v. \v=1 \vee \v=2$, with no further constraints, or it can
be encoded as the function $\lambda \v. g(\v)$ with an associated
constraint, $\phi\in\Prop^1 = \forall \nu.~g(\nu) \Leftrightarrow
\nu=1 \vee \nu=2$, defining the uninterpreted predicate $g$.
The translation adopts one or the other approach, depending on the
need. For uniformity, we consider the encoding of a set as pair
($\phi$,$\G$), where $\G$ is a meta function, and $\phi$ is a FOL
formula constraining any uninterpreted predicates used in $G$.

Due to the presence of bind ($\bind$) in $\SL$, a set expression $s$
may contain free variables introduced by an enclosing binder. For
instance, consider the $\SL$ expression $s_1\bind(\lambda x. \{ y\,|\,
y=x+1\})$, where $s_1$ is an integer set (expression). The
subexpression $\{ y\,|\, y=x+1\}$ (call it $s_2$) contains $x$ as a
free variable. In such cases, the predicate associated with the
subexpression should also be indexed by its free variables so that a
unique set exists for each instantiation of the free variables. Thus,
the predicate ($\G$) associated with the subexpression from the
above example should be $\lambda\v_1.\lambda\v_2.~\v_2 = \v_1 + 1$, so
that the set $\G\; x_1$ is different from the set $\G\; x_2$ for
distinct $x_1,x_2 \in s_1$. Intuitively, the bind expression
$s_1\bind(\lambda x. \{ y\,|\, y=x+1\})$ denotes the set
$\bigcup\limits_{x\in s_1}\G\;x$.

The translation algorithm (Fig.~\ref{fig:logic}) takes free variables
into account. Given a set expression $s\in\SL$, whose (possible) free
variables are $\nubar$ in the order of their introduction (top-most
binder first), $\mssemof{\nubar}{s}$ returns the encoding of $s$ as
$(\phi,\G)$.  The meta-function $\G$ is a predicate indexed by the
(possible) free variables of $s$, and thus its arity is $|\nubar| + 1$.
Note that $\nubar$ is only a sequence of variables introduced by the
enclosing binders of $s$, and not all may actually occur free in $s$.
Nonetheless, its predicate $\G$ is always indexed by $|\nubar|$ free
variables for uniformity. The translation encodes database state as an
uninterpreted sort. Considering that the state is actually a set of
records, we define an uninterpreted relation ``$\in$'' to relate records
and states. Thus, a variable set expression $\stg$ denoting a database
state is encoded as the predicate $\lambda (\vbar,r).~r\in\stg$, where
$|\vbar|=|\nubar|$ (predicates are uncurried for simplicity; $\vbar$
is a comma-separated sequence; $r\not\in\SL$ is a special variable).
The constraints associated with the encoding of a state are trivial
(denoted $\top$). The set literal expression $\{x\,|\, \varphi\}$ is encoded
straightforwardly. The conditional set expression is encoded as an
if-then-else predicate in FOL, where the predicates on true and false
branches are computed from the set subexpressions $s_1$ and $s_2$,
respectively. The conjunction of constraints $\phi_1$ and $\phi_2$,
from $\mssemof{\nubar}{s_1}$ and $\mssemof{\nubar}{s_2}$ (resp.), is
propagated upwards as the constraint of the conditional expression.
A set union expression is encoded similarly.

The first-order encoding of a bind expression describes the semantics
of the set monad's bind operator in FOL. Let $s_1$ be a set, and let
$f$ be a function that maps each variable in $s_1$ to a new set. Then,
$s_2 = s_1 \bind f$ if and only if for all $y\in s_2$, there exists an
$x \in s_1$ such that $y = f(x)$, and forall $x\in s_1$, $f(x)\in
s_2$. The encoding essentially adds new constraints to this effect.
The translation first encodes $s_1$ and $s_2$ to obtain
$(\phi_1,\G_1)$ and $(\phi_2,\G_2)$, respectively. Since $s_2$ is
under a new binder that binds $x$, the free variable sequence of $s_2$
is $\nubar,x$. In the interest of hygiene, we substitute a fresh $a$
for $x$, making the sequence $\nubar,a$. The set $s$ is encoded as a
new uninterpreted predicate $g$ indexed by $s$'s free variables
($\nubar$). Since the set denoted by $g$ is the result of the bind
$s_1 \bind \lambda x.s_2$, first-order constraints defining the bind
operation (as described above) are generated. The constraints relate
the predicates $\G_1$ and $\G_2$, representing $s_1$ and $s_2$
(resp.), to the uninterpreted predicate $g$ that represents $s$. The
constraints are assigned names ($\pi_1$ and $\pi_2$) to give them an
easy handle.

The first-order encoding of the $\existsl(\stg,\phi,s)$ expression
essentially Skolemizes the existential. Skolemizing is the
process of substituting an existentially bound $x$ in
$\phi_x\in\Prop^1$ with $f(\nubar)$, where $f$ is a fresh
uninterpreted function (called the Skolem function), and $\nubar$ are
the free variables in $\phi_x$ bound by enclosing universal
quantifiers. Due to the
decidability restrictions (Sec.~\ref{sec:decidability}), the only
uninterpreted functions we admit in our logic are boolean (i.e.,
predicates/relations). Consequently, we cannot define the Skolem
function $f$ directly. Instead, we define it via an uninterpreted
relation, by explicitly asserting the function property:
\begin{smathpar}
  \begin{array}{c}
    (\forall \nubar.\forall a.\forall b.~f(\nubar,a) \wedge f(\nubar,b)
    \Rightarrow a = b)
    ~~\conj~~ (\forall \nubar.\exists a. f(\nubar,a))
  \end{array}
\end{smathpar}
We then replace the existentially bound $\stg$ with a new universally
bound $a$ in $\phi$ and $s$, such that $f(\nubar,a) $ holds, before
encoding the existentially bound $s$.

\noindent {\bf Example} Let us reconsider the  TPC-C \C{new\_order}
transaction from Sec.~\ref{sec:motivation}. Recall that the state
transformer ($\T$) for the \C{foreach} loop shown in
Fig.~\ref{fig:foreach_code} is (\C{k1}, \C{k2}, and \C{k3} are
constants):
\begin{smathpar}
\begin{array}{l}
  \lambda(\stl,\stg).~ \stl \cup \C{item\_reqs}\bind
      (\lambda\C{item\_req}.~ \F_U(\stg) \cup \F_I(\stg))
\end{array}
\end{smathpar}
where:
\begin{smathpar}
  \begin{array}{lcl}
    \F_U & = & \lambda(\stg).~ \stg \bind(\lambda s. 
                \itel{{\sf table}(s) = \C{Stock} \conj 
                  s.\C{s\_i\_id} = \C{item\_req.ol\_i\_id}\\
         &   & \hspace*{0.8in}}
                     {\{ \langle \C{s\_i\_id}=s.\C{s\_i\_id};\, 
                                 \C{s\_d\_id}=s.\C{s\_d\_id};\,
                                 \C{s\_qty} = \C{k1}\rangle \}\\
         &   & \hspace*{0.8in}} {\emptyset}) \\
    \F_I & = & \lambda(\stg).~ \{\langle\C{ol\_o\_id}=\C{k2};\,
                 \C{ol\_d\_id}=\C{k3};\,
                 \C{ol\_i\_id}=\C{item\_req.ol\_i\_id};\,\\
         &   & \hspace*{2.3in}\C{ol\_qty}=\C{item\_req.ol\_qty}\rangle\}\\
  \end{array}
\end{smathpar}
For any $\stg$, $\F_U(\stg)$ and $\F_I(\stg)$ are expressions in
$\SL$, so can be translated to FOL by the encoding algorithm in
Fig.~\ref{fig:logic}. Since the iteration variable \C{item\_req}
occurs free in these expressions, the appropriate application of the
encoding algorithm is $\mssemof{\C{item\_req}}{\F_U(\stg)}$ and
$\mssemof{\C{item\_req}}{\F_I(\stg)}$, which results in
$(\phi_U,\G_U)$ and $(\phi_I,\G_I)$, respectively, where $\phi_U$,
$\phi_I$, $\G_U$, $\G_I$ are as shown below:
\begin{smathpar}
  \begin{array}{lcl}
    \phi_U & = & \hspace*{0.1in}\forall \C{item\_req}.\forall s. \forall s'.~ 
        \pi_1(\C{item\_req}) \Leftrightarrow \\
    & & \hspace*{0.3in}(s \in \stg) \conj 
        (\itec{{\sf table}(s)=\C{Stock} \wedge 
              s.\C{s\_i\_id} = \C{item\_req.ol\_i\_id}\\
    & & \hspace*{0.85in}}
              {s' = \langle \C{s\_i\_id}=s.\C{s\_i\_id};\, 
                            \C{s\_d\_id}=s.\C{s\_d\_id};\,
                            \C{s\_qty} = \C{k1}\rangle\\
    & & \hspace*{0.85in}} {\bot}) \Rightarrow g_0(\C{item\_req},s')\\
    & & \hspace*{-0.05in}\conj \forall \C{item\_req}.\forall s'. \exists s.~ 
        \pi_2(\C{item\_req}) \Leftrightarrow \\
    & & \hspace*{0.15in} g_0(\C{item\_req},s') \Rightarrow s \in \stg \conj
        \itec{{\sf table}(s)=\C{Stock} \wedge 
              s.\C{s\_i\_id} = \C{item\_req.ol\_i\_id}\\
    & & \hspace*{1.5in}}
              {s' = \langle \C{s\_i\_id}=s.\C{s\_i\_id};\, 
                            \C{s\_d\_id}=s.\C{s\_d\_id};\,
                            \C{s\_qty} = \C{k1}\rangle\\
    & & \hspace*{1.5in}} {\bot} \\
    \G_U & = & \lambda (\C{item\_req},r).~ \pi_1(\C{item\_req}) \wedge 
        \pi_2(\C{item\_req}) \wedge g_0(\C{item\_req},r)\\
        \phi_I & = & \top \\ 
    \G_I & = & \lambda (\C{item\_req},r).~ 
        r = \langle \C{ol\_o\_id}=\C{k2};\,
                    \C{ol\_d\_id}=\C{k3};\, 
                    \C{ol\_i\_id}=\C{item\_req.ol\_i\_id};\,\\
    & & \hspace*{3in}\C{ol\_qty}=\C{item\_req.ol\_qty} \rangle\\
  \end{array}
\end{smathpar}
Since the transformer ($\T$) of the \C{foreach} loop is not nested
does not contain any free iteration variables, the appropriate
application of the encoding algorithm  is
$\mssemof{\emptyset}{\T(\stl,\stg)}$, which results in the
$(\phi_I \wedge \phi_U \wedge \phi_1 \wedge \phi_2,\G)$, where
$\phi_1$, $\phi_2$, and $\G$ are as defined below:
\begin{smathpar}
  \begin{array}{lcl}
    \phi_1 & = & \forall \C{item\_req}. \forall s.~ 
        \pi_3 \Leftrightarrow \C{item\_req} \in \C{item\_reqs} \wedge
        \G_U(\C{item\_req},s') \vee \G_I(\C{item\_req},s')
        \Rightarrow g_1(s)\\
    \phi_2 & = & \forall s.\exists \C{item\_req}. ~ 
        \pi_4 \Leftrightarrow g_1(s) \Rightarrow \C{item\_req} \in \C{item\_reqs} 
        \wedge \G_U(\C{item\_req},s') \vee \G_I(\C{item\_req},s')\\
    \G & = & \lambda(r).~\pi_3 \wedge \pi_4 \wedge g_1(r)\\
  \end{array}
\end{smathpar}

\subsection{Decidability}
\label{sec:decidability} 

Observe that the encoding shown in Fig.~\ref{fig:logic} maps to
a fragment of FOL that satisfies the following syntactic properties:
\begin{itemize}
  \item All function symbols, modulo those that are drawn from
    $\Prop^0$ and $\Prop^1$, are uninterpreted and boolean.
  \item All quantification is first-order; second-order objects, such
    as sets and functions, are never quantified.
  \item Quantifiers appear only at the prenex position, i.e., at the
    beginning of a quantified formula.
\end{itemize}
The simple syntactic structure of the fragment already makes is
amenable for automatic reasoning via an off-the-shelf SMT solver, such
as Z3. The decidability of this fragment, however, is more subtle and
discussed below.

Consider a set expression $s$ with no free variables (i.e., $\nubar =
\emptyset$, like $\T(\stl,\stg)$ from the above example). Let
$(\phi,G) = \mssemof{\emptyset}{s}$. Note that $\phi$ is a conjunction
of (a).  $\phi_i$'s, where each $\phi_i$ results from encoding a
subexpression $s_i$ of $s$, and (b). a $\phi_s$, resulting from
encoding $s$ itself (i.e., its top-level expression). From Fig.~\ref{fig:logic},
it is clear that $\phi_s$ is either $\top$ (for the first four cases),
or it is a prenex-quantified formula, where quantification is either
$\forall^2$, or $\exists$, or $\forall\exists$. Generalizing this
observation, for a set expression $s$ with $|\nubar|$ free variables,
$\phi_s$, if quantified, is a prenex-quantified formula, where
quantification assumes one among the forms of $\forall^{|\nubar|+2}$,
or $\forall^{|\nubar|}\exists$, or $\forall^{|\nubar|+1}\exists$.  In
other words, the number of $\forall$ quantifiers preceding an
$\exists$ quantifier is utmost one more than the number of free
variables ($\nubar$) in $s$. For the convenience of this discussion,
let us call $\forall^{|\nubar|+1}\exists$ as the prenex signature of
$\phi_s$. 

Next, in Fig.~\ref{fig:logic}, observe that the (ordered) set $\nubar$
is extended only in the encoding rule for $\bind$. Since an occurrence
of $\bind$ adds a quantifier to $|\nubar|$, if $s$ is a bind
expression nested inside a top-level bind expression (like
$\F_U(\stg)$ from the above example), then the prenex signature
of $\phi_s$ is $\forall^2\exists$.  Furthermore, if the subexpressions
of $s$ are neither {bind} nor $\existsl$ expressions, then none of the
$\phi_i$'s are quantified, and the prenex signature of $\phi = \bigwedge_i\phi_i \wedge \phi_s$
remains $\forall^2\exists$. A similar
observation holds when $s$ is an $\existsl$ expression nested inside a
top-level {bind} expression.  Since $\existsl$ is generated as a
result of stabilizing a SQL command transformer, which is always a
non-nested bind expression, the subexpression ($s'$) of $\existsl$ is
a non-nested bind expression. $s'$ is however nested inside a
top-level bind expression, hence its prenex signature is
$\forall^2\exists$.  Since $\existsl$ does not extend $\nubar$, the
prenex signature of $s$ remains $\forall^2\exists$. When $s$ is an
expression other than $\bind$ or $\existsl$, then $\phi_s$ is not a
quantified formula, and its prenex signature is trivially subsumed by
$\forall^2\exists$. Thus, for the subset of $\SL$, where bind
expressions are restricted to one level of nesting, the FOL formulas
generated by the encoding have the prenex signature as
$\forall^2\exists$.

The fragment of FOL that admits formulas with prenex signatures of the
form $\forall^2\exists^*$ is called the G\"odel-K\'almar-Sch\"utte ({\sf
GKS}) fragment~\cite{gks}, which is known to be decidable. The
language of encoding, however, is a combination of {\sf GKS} with (a).
$\Prop^0$, the theory from which quantifier-free propositions
($\varphi$) that encode object language expressions are drawn, and
(b). $\Prop^1$, the theory from which invariants ($I$) are drawn. Thus,
the encoding of the subset of $\SL$ described above is decidable if the
combination of ${\sf GKS} + \Prop^0 + \Prop^1$ is decidable. We write
$\SL[\Prop^0,\Prop^1]$ to highlight the parameterization of $\SL$ on
$\Prop^0$ and $\Prop^1$.  The discussion in the previous paragraph
points to the existence of non-trivial subsets in $\SL[\Prop^0,\Prop^1]$ that are
decidable:
\begin{theorem}
  There exist $\SL'[\Prop^0,\Prop^1] \subset \SL[\Prop^0,\Prop^1]$
  such that $\SL'$ is decidable if ${\sf GKS}+\Prop^0+\Prop^1$ is
  decidable.
\end{theorem}
One interesting example of such an $\SL'$ is the subset described above: $\SL$ with
bind expressions confined to one level of nesting. We denote this
subset as $\SL^1[\Prop^0,\Prop^1]$, for which we assert decidability:
\begin{corollary}
  $\SL^1[\Prop^0,\Prop^1]$ is decidable if ${\sf GKS}+\Prop^0+\Prop^1$ is
  decidable.
\end{corollary}
$\SL^1$ is a useful subset of $\SL$, for it corresponds to $\txnimp$
programs without nested \C{foreach} loops. Observe that the
\C{new\_order} transaction (Fig.~\ref{fig:new_order_code}) belongs to
this subset.  Indeed, $\SL^1$, while being a restricted version of
$\SL$, is nonetheless expressive enough to cover all the benchmarks we
considered in Sec.~\ref{sec:case-studies}.

A useful instantiation of $\SL^1$ is $\SL[{\sf BV},{\sf GKS}+{\sf
BV}]$, where ${\sf BV}$ is the theory of bit-vector arithmetic, which
is often used to encode the finite-bit integer arithmetic of real
programs. Finite-bit integer arithmetic has a finite axiomatization in
{\sf GKS}. For instance, 32-bit integers can be encoded as $2^{32}$
distinct constants of an uninterpreted sort $T$, while integer
operations like addition and multiplication can be encoded as
uninterpreted functions whose properties are enumerated for the entire
domain of $T$. Thus {\sf BV} is subsumed by {\sf GKS}. Since the
latter is decidable, the combination is decidable:
\begin{theorem}
  $\SL^1[{\sf BV},{\sf GKS}+{\sf BV}]$ is decidable.
\end{theorem}
This instantiation requires $I$ to be drawn from {\sf GKS}+{\sf BV},
which is expressive enough to describe common database integrity
constraints, such as referential integrity, non-negativeness of all
integer values in a column etc.  The isolation specifications
presented in \S\ref{sec:isolation} are already simple first-order
formulas that can be encoded in {\sf GKS}.  Furthermore, it is also
reasonable to expect the guarantee ($G$) of a transaction to be
expressible in the same logic as its inferred $\F$, since $\F$
(without the stability check) is essentially a complete
characterization of the transaction, while $G$ is only an abstraction.
Thus, with $\SL^1[{\sf BV},{\sf GKS}+{\sf BV}]$ as the language of
inference, the verification problem for weakly isolated transactions
is decidable.

\section{Implementation}
\label{sec:implementation}

\begin{figure}
\begin{ocaml}
type table_name =  District | Order | Order_line | Stock

type district = {d_id: int; d_next_o_id: int}
type order = {o_id: int; o_d_id: int; o_c_id: int; o_ol_cnt: int}
type order_line = {ol_o_id: int; ol_d_id: int; ol_i_id: int; ol_qty: int}
type stock = {s_i_id: int; s_d_id:int; s_qty: int}
\end{ocaml}
\caption{OCaml type definitions corresponding to the TPC-C schema from
Fig.~\ref{fig:schema}}
\label{fig:ocaml-schema}
\vspace*{-10pt}
\end{figure}

We have implemented our DSL to define transactions as monadic
computations in OCaml (modulo some syntactic sugar), and our automatic
reasoning framework as an extra frontend pass (called \thetool) in the
ocamlc 4.03 compiler\footnote{The source code is available at
available at \url{https://github.com/gowthamk/acidifier}}. The input
to \thetool is a program in our DSL that describes the schema of the
database as a collection of OCaml type definitions, and transactions
as OCaml functions, whose top-level expression is an application of
the \C{atomically\_do} combinator. For instance, TPC-C's schema from
Fig.~\ref{fig:schema} can be described via the OCaml type definitions
shown in Fig.~\ref{fig:ocaml-schema}.  \thetool also requires a
specification of the program in the form of a collection of guarantees
($G$), one per transaction, and an invariant $I$ that is a conjunction
of the integrity constraints on the database. An auxiliary DSL that
includes the first-order logic (FOL) combinators has been implemented
for this purpose. \thetool's verification pass follows OCaml's type
checking pass, hence the concrete artifact of verification is OCaml's
typed AST. The tool is already equipped with  an axiomatization of
PostgreSQL and MySQL's isolation levels expressed in our FOL DSL.
Other data stores can be similarly axiomatized. The concrete result of
verification is an assignment of an isolation level of the selected
data store to each transaction in the program.

At the top-level, \thetool runs a loop that picks an unverified
transaction and progressively strengthens its isolation level until it
passes verification. If the selected data store provides a
serializable isolation level, and if the program is sequentially
correct, then the verification is guaranteed to succeed. Within the
loop, \thetool first computes the various rely relations needed for
verification ($R$, $\R_l$, and $\R_c$). It then traverses the AST of a
transaction, applying the inference rules to construct a state
transformer, checks its stability, and weakens it ($\stabilize{\cdot}$)
if it is not stable. The result of traversing the transaction's AST is
therefore a state transformer ($\F$) that is stable w.r.t $\R_l$, which
is also stabilized against $\R_c$ (using $\stabilize{\cdot}$), and
then checked against the transaction's stated guarantee ($G$). If the
check passes, then the guarantee is verified to check if it preserves
the invariant $I$. The successful result from both checks results
in the transaction being certified correct under the current choice of
its isolation level. Successful verification of all transactions
concludes the top-level execution, returning the inferred isolation
levels as its output.  \thetool uses the Z3 SMT solver as its underlying reasoning engine. Each
implication check described above is first encoded in FOL, applying
the translation described in \S\ref{sec:inference} wherever
necessary.

\subsection{Pragmatics}

\textbf{Real-World Isolation Levels} The axiomatization of the
isolation levels presented in \S\ref{sec:isolation} leaves out
certain nuances of their implementations on real data stores, which
need to be taken into account for verification to be effective in
practice.  We take these into account while linking \thetool with
store-specific semantics (isolation specifications, etc.). As an
example, consider how PostgreSQL implements an \C{UPDATE} operation.
\C{UPDATE} first selects the records that meet the search criteria
from the snapshot against which it is executing (the snapshot is
established at the beginning of the transaction if the isolation level
is SI, or at the beginning of the \C{UPDATE} statement if the
isolation level is RC). The selected records are then visited in the
actual database (if they still exist), write locks are obtained, and
the update is performed, provided that each matched record still meets
\C{UPDATE}'s search criteria. If a record no longer meets the
search criteria (due to a concurrent update), it is excluded
from the update, and the write lock is immediately released.
Otherwise, the record remains locked until the transaction commits. 

Clearly, this sequence of events is not atomic, unlike the assumption
made by our formal model because  the implementation admits interference
between the updates of individual records that meet the search
criteria.  Nonetheless, through a series of relatively straightforward
deductions, we can show that PostgreSQL's \C{UPDATE} is in fact
equivalent (in behavior) to a sequential composition of two atomic
operations $c_1;c_2$, where $c_1$ is effectively a \C{SELECT}
operation with the same search criteria as \C{UPDATE}, and $c_2$ is
a slight variation of the original \C{UPDATE} that updates a
record only if a record with the same id is present in the set of records
returned by \C{SELECT}: 
\begin{smathpar}
\begin{array}{lcl}
\updatee{(\lambda x. ~e_1)}{(\lambda x.~e_2)}
&
\longrightarrow
&
\lete{y}{\selecte{(\lambda x.~e_1})}
     {\updatee{(\lambda x.~e_1 \wedge x.\idf\in\dom(y))}
              {(\lambda x.~e_2})}\\
\end{array}
\end{smathpar}
The intuition behind this translation is the observation that all
interferences possible during the execution of the \C{UPDATE} can be
accommodated between the time the records are selected from the
snapshot, and the time they are actually updated.  \thetool performs this
translation if the selected store is PostgreSQL, allowing it to reason
about \C{UPDATE} operations in a way that is faithful to its semantics
on PostgreSQL. \thetool also admits similar compensatory logic for
certain combinations of isolation levels and operations on MySQL.

\textbf{Set functions} SQL's \C{SELECT} query admits projections of
record fields, and also application of auxiliary functions such as
\C{MAX} and \C{MIN}, e.g., \C{SELECT MAX(ol\_o\_id) FROM
Order\_line WHERE $\ldots$}, etc. We admit such extensions as set functions
in our DSL (e.g., \C{project}, \C{max}, \C{min}), and axiomatize their
behavior. For instance:
\begin{smathpar}
\begin{array}{lcl}
  s_2 \;=\;\C{project}\,s_1\,(\lambda z.~e) & \Leftrightarrow &
  \forall y.~y\in s_2 \Leftrightarrow  \exists(x \in s_1).~y = [x/z]e\\
  x \;=\; \C{max}\,s & \Leftrightarrow & x \in s \conj \forall(y \in
  s).~y\le x\\
\end{array}
\end{smathpar}
There are however certain set functions whose behavior cannot be
completely axiomatized in FOL. These include \C{sum}, \C{count} etc.
For these, we admit imprecise axiomatizations. 

\textbf{Annotation Burden} \thetool significantly reduces the
annotation burden in verifying a weakly isolated transactions by
eliminating the need to annotate intermediate assertions and loop
invariants.  Guarantees ($G$) and global invariants ($I$), however,
still need to be provided. Alternatively, a weakly isolated
transaction $T$ can be verified against a generic serializability
condition,  eliminating the need for guarantee annotations. In this
mode, \thetool first infers the transformer $\F_{SER}$ of $T$ without
considering any interference, which then becomes its guarantee ($G$).
Doing likewise for every transaction results in a rely relation ($R$)
that includes $\F_{SER}$ of every transaction. Verification now
proceeds by taking interference into account, and verifying that each
transaction still yields the same $F$ as its $F_{SER}$. The result of
this verification is an assignment of (possibly weak) isolation levels
to transactions which nonetheless guarantees behavior equivalent to a
serializable execution.

\section{Evaluation}
\label{sec:case-studies}

In this section, we present our experience in running \thetool on two
different applications: \emph{Courseware}: a course registration
system described by~\cite{gotsmanpopl16}, and TPC-C.

\begin{figure}[t]
\begin{ocaml}
  type table_name = Student | Course | Enrollment
  type student = {s_id: id; s_name: string}
  type course = {c_id: id; c_name: string; c_capacity: int}
  type enrollment = {e_id: id; e_s_id: id; e_c_id: id}

  let enroll_txn sid cid = 
    let crse = SQL.select1 [Course] (fun c -> c.c_id = cid) in
    let s_c_enrs = SQL.select [Enrollment] (fun e -> e.e_s_id = sid && 
                                                     e.e_c_id = cid) in
    if crse.c_capacity > 0 && Set.is_empty s_c_enrs then
      (SQL.insert Enrollment {e_id=new_id (); e_s_id=sid; e_c_id=cid};
       SQL.update Course (fun c -> {c with c_capacity = c.c_capacity - 1})
                         (fun c -> c.c_id = cid)) else ()
       
  let deregister_txn sid = 
    let s_enrs = SQL.select [Enrollment] (fun e -> e.e_s_id = sid) in
    if Set.is_empty s_enrs then
      SQL.delete Student (fun s -> s.s_id = sid) else ()
\end{ocaml}
\caption{Courseware Application}
\label{fig:courseware_code}
\vspace*{-10pt}
\end{figure}

\textbf{Courseware} The Courseware application allows new courses to be
added (via an \C{add\_course} transaction), and new students to be
registered (via a \C{register} transaction) into a database. A registered
student can enroll (\C{enroll}) in an existing course,
provided that enrollment has not already exceeded the course
capacity (\C{c\_capacity}). A course with no enrollments can be
canceled (\C{cancel\_course}). Likewise, a student who is not enrolled
in any course can be deregistered (\C{deregister}). Besides
\C{Student} and \C{Course} tables, there is also an \C{Enrollment}
table to track the many-to-many enrollment relationship between
courses and students. The simplified code for the Courseware
application with only \C{enroll} and \C{deregister}
transactions is shown in Fig.~\ref{fig:courseware_code}. The
application is required to preserve the following invariants on the
database:

\begin{enumerate}
\item  $I_1$: An enrollment record should always refer to an existing student and an existing course.
\item  $I_2$: The capacity (\C{c\_capacity}) of a course should always be a
  non-negative quantity.
\end{enumerate}
\noindent Both $I_1$ and $I_2$ can be violated under weak isolation.
$I_1$ can be violated, for example, when \C{deregister} runs
concurrently with \C{enroll}, both at RC isolation. While the former
transaction removes the student record after checking that no
enrollments for that student exists, the latter transaction
concurrently adds an enrollment record after checking the student
exists.  Both can succeed concurrently, resulting in an invalid state.
Invariant $I_2$ can be violated by two \C{enroll}s, both reading
\C{c\_capacity}=1, and both (atomically) decrementing it, resulting in
\C{c\_capacity}=-1.  We ran \thetool on the Courseware application
(Fig.~\ref{fig:courseware_code}) after annotating transactions with
their respective guarantees, and asserting $I = I_1 \wedge I_2$ as the
correctness condition. The guarantees $G_e$ and $G_d$ for \C{enroll}
and \C{deregister} transactions, respectively, are shown below:
\begin{smathpar}
  \begin{array}{lcl}
    G_e(\stg,\stg') & \Leftrightarrow & \stg_s'=\stg_s
      \conj \exists\C{cid}.\exists\C{sid}.\\
    & & \hspace*{0.3in} \stg_c' = \stg_c \bind \lambda c.~
        \itel{c.\C{c\_id}=\C{cid}\\
    & & \hspace*{1.15in}}
          {\existsl(c',~c'.\C{id}=c.\C{id} \wedge
              c'.\C{c\_name}=c.\C{c\_name} \\
    & & \hspace*{2.6in}\wedge~ c'.\C{c\_capacity}\ge0, ~\{c'\})\\
    & & \hspace*{1.15in}}
          {\{c\}}\\
    & & \hspace*{0.15in}\conj \stg_e = \stg_e' \bind \lambda e.~ 
        \itel{e.\C{e\_c\_id}=\C{cid} \wedge e.\C{e\_s\_id}=\C{sid}}
          {\emptyset}{\{e\}}\\
    G_d(\stg,\stg') & \Leftrightarrow & \stg_c' = \stg_c \conj
      \stg_e' = \stg_e \conj \exists \C{sid}.~
      \itel{\forall(e \in \stg_e).~e.\C{e\_s\_id}\neq\C{sid}\\
    & & \hspace*{1.55in}}
          {\stg_s' = \stg_s \bind \lambda s.~
           \itel{s.\C{id}=\C{sid}}{\emptyset}{\{s\}}\\
    & & \hspace*{1.55in}}
          {\stg_s' = \stg_s}\\
  \end{array}
\end{smathpar}
For the sake of this presentation we split $\stg$ into three disjoint
sets of records, $\stg_s$, $\stg_c$, and $\stg_e$, standing for
\C{Student}, \C{Course}, and \C{Enrollment} tables, respectively.
Observing that the set language $\SL$ (Sec.~\ref{sec:inference}),
besides being useful for automatic verification, also facilitates
succinct expression of transaction semantics, we define $G_e$ and
$G_d$ in a combination of FOL and $\SL$. $G_e$ essentially says that
the \C{enroll} transaction leaves the \C{Student} table unchanged,
while it may update the \C{c\_capacity} field of a \C{Course} record
to a non-negative value (even when it doesn't update, it is the case
that $c'.\C{c\_capacity}\ge0$, because $c'=c$, and $c\in\stg_c$, and
we know that $I_2(\stg_c)$). $G_e$ also conveys that \C{enroll} might
insert a new \C{Enrollment} record by stating that $\stg_e$, the
\C{Enrollment} table in the pre-state, contains all records $e$ from
$\stg_e'$, the table in the post-state, except when $e.\C{e\_c\_id}$ and
$e.\C{e\_s\_id}$ match \C{cid} and \C{sid}, respectively. The
guarantee $G_d$ of \C{deregister} asserts that the
transaction doesn't write to \C{Course} and \C{Enrollment} tables. The
transaction might however delete a \C{Student} record bearing an
\C{id}=\C{sid} (formally, $\stg_s' = \stg_s \bind \lambda s.~
\itel{s.\C{id}=\C{sid}}{\emptyset}{\{s\}}$), for some \C{sid} for
which no corresponding \C{Enrollment} records are present in the
pre-state (in other words, $\forall(e \in
\stg_e).~e.\C{e\_s\_id}\neq\C{sid}$).

With help of the guarantees, such as those described above, \thetool
was able to automatically discover the aforementioned anomalous
executions, and was subsequently able to infer that the anomalies can
be preempted by promoting the isolation level of \C{enroll} and
\C{deregister} to SER (on both MySQL and PostgreSQL), leaving
the isolation levels of remaining transactions at RC. The total time
for inference and verification took less than a minute running on a
conventional laptop.

\begin{table}[]
\centering
\begin{tabular}{l|c|c|c|c|c|}
\cline{2-6}
                                 & \multicolumn{1}{l|}{\C{new\_order}} & \multicolumn{1}{l|}{\C{delivery}} & \multicolumn{1}{l|}{\C{payment}} & \multicolumn{1}{l|}{\C{order\_status}} & \multicolumn{1}{l|}{\C{stock\_level}} \\ \hline
\multicolumn{1}{|l|}{MySQL}      & SER                                   & SER                                 & RC                                 & RC                                       & RC                                      \\ \hline
\multicolumn{1}{|l|}{PostgreSQL} & SI                                    & SI                                  & RC                                 & RC                                       & RC                                      \\ \hline
\end{tabular}
\caption{The discovered isolation levels for TPC-C transactions}
\label{tab:tpcc}
\vspace*{-10pt}
\end{table}

\textbf{TPC-C} The simplified schema of the TPC-C benchmark has been
described in Sec.~\ref{sec:motivation}. In addition to the tables
shown in Fig.~\ref{fig:schema}, the TPC-C schema also has
\C{Warehouse} and \C{New\_order} tables that are relevant for
verification.  To verify TPC-C, we examined four of the twelve
consistency conditions specified by the standard, which we name $I_1$
to $I_4$:

\begin{enumerate}
\item Consistency condition $I_1$  requires that the sales bottom line
of each warehouse equals the sum of the sales bottom lines of all
districts served by the warehouse.

\item Conditions $I_2$ and $I_3$ effectively enforce uniqueness of ids assigned
  to \C{Order} and \C{New\_order} records, respectively, under a district.

\item  Condition $I_4$ requires that the number of order lines under a district
  must match the sum of order line counts of all orders under the district.
\end{enumerate}

Similar to the example discussed in Sec.~\ref{sec:motivation}, there
are a number of ways TPC-C's transactions violate the aforementioned
invariants under weak isolation. \thetool was able to discover all such
violations when verifying the benchmark against $I =
\bigwedge_{i}I_i$, with guarantees of all three transactions
provided. The isolation levels were subsequently strengthened  as
shown in Table.~\ref{tab:tpcc}.  As before, inference and verification
took less than a minute.


To sanity-check the results of \thetool, we conducted experiments with a
high-contention OLTP workload on TPC-C aiming to explore the space of
correct isolation levels for different transactions. The workload
involves a mix of all five TPC-C transactions executing against a
TPC-C database with 10 warehouses. Each warehouse has 10 districts,
and each district serves 3000 customers. There are a total of 5
transactions in TPC-C, and given that MySQL and PostgreSQL support 3
isolation levels each, there are a total of $3^5 = 243$ different
configurations of isolation levels for TPC-C transactions on MySQL and
PostgreSQL. We executed the benchmark with all 243 configurations, and
found 171 of them violated at least one of the four invariants we
considered.  As expected, the isolation levels that \thetool infers for the
TPC-C transactions do not result in invariant violations, either on
MySQL or on PostgreSQL, and were determined to be the weakest safe
assignments possible.

\section{Related Work}
\label{sec:relatedwork}

\paragraph{Specifying weak isolation.}
Adya~\cite{adyaphd} specifies several weak isolation levels in terms
of \emph{dependency graphs} between transactions, and the kinds of
dependencies that are forbidden in each case. The operational nature
of Adya's specifications make them suitable for runtime monitoring and
anomaly detection~\cite{kemmevldb,feketesigmod08,pssi2011}, whereas
the declarative nature of our specifications make them suitable for
formal reasoning about program behavior. ~\cite{pldi15} specify
isolation levels declaratively as trace well-formedness conditions,
but their specifications implicitly assume a complete trace with only
committed transactions, making it difficult to reason about a program
as it builds the trace. ~\cite{gotsmanconcur15} specify isolation
levels with atomic visibility, but their specifications are also for
complete traces.  Both the aforementioned specification frameworks use
the vocabulary introduced in~\cite{burckhardt14}. However, none of
them are equipped with a reasoning framework that can use such
specifications to verify programs under weak isolation.

Recent work described in~\cite{CPA+17} also explores the
use of a state-based interpretation of isolation as we do, and 
like our approach, develops specifications of weak isolation that are not
tied to implementation-specific artifacts.  However, they do not
consider verification (manual or automated) of client programs, and it
is not immediately apparent if their specification formalism is
amenable for use within a verification toolchain. ~\cite{WB17}
present a dynamic analysis for weak isolation that attempts to
discover weak isolation anomalies from SQL log files.  Their solution,
while capable of identifying database attacks due to the use of
incorrect isolation levels, does not consider how
to verify application correctness, infer proper isolation levels, or
formally reason about the relationship between weak-isolation levels
and application invariants.
\paragraph{Reasoning under weak isolation.} ~\cite{feketessi} propose
a theory to characterize non-serializable executions that arise under
{\sc si}. They also propose an algorithm that
allocates either {\sc si} or {\sc ser} isolation levels to
transactions while guaranteeing
serializability. ~\cite{gotsmanpodc16} improve
on Adya's {\sc si} specification and use it to derive a static
analysis that determines the safety of substituting {\sc si} with a
weaker variant called \iso{Parallel Snapshot Isolation}~\cite{psi}.
These efforts focus on establishing the equivalence of executions
between a pair of isolation levels, without taking application
invariants into account. ~\cite{bern2000}
propose informal semantic conditions to ensure the satisfaction of
application invariants under weaker isolation levels.  All these
techniques are tailor-made for a finite set of well-understood
isolation levels (rooted in~\cite{berenson}).
\vspace*{-6pt}
\paragraph{Reasoning under weak consistency.} There have been several
recent proposals to reason about programs executing under weak
consistency~\cite{bailisvldb, alvarocalm, gotsmanpopl16,redblueatc,
redblueosdi, ecinec}. All of them assume a system model that offers a
choice between a \emph{coordination-free} weak consistency level
(\emph{e.g.}, eventual consistency~\cite{redblueosdi, redblueatc,
ecinec, alvarocalm, bailisvldb}) or causal
consistency~\cite{lbc16,gotsmanpopl16}). All these efforts involve
proving that atomic and fully isolated operations preserve application
invariants when executed under these consistency levels.  In contrast,
our focus in on reasoning in the presence of weakly-isolated
transactions under a strongly consistent store.  ~\cite{gotsmanpopl16}
adapt \iso{Parallel Snapshot Isolation} to a transaction-less setting
by interpreting it as a consistency level that serializes writes to
objects; a dedicated proof rule is developed to help prove prove
program invariants hold under this model. By parameterizing our proof
system over a gamut of weak isolation specifications, we avoid the
need to define a separate proof rule for each new isolation level we
may encounter.
\paragraph{Inference.}  ~\cite{Vaf10,Vaf10a}
describe \emph{action inference}, an inference procedure for computing
rely and guarantee relations in the context of RGSep~\cite{VP07}, an
integration of rely-guarantee and separation logic~\cite{Rey02} that
allows one to precisely reason about local and shared state of a
concurrent program. The ideas underlying action inference have been
used to prove memory safety, linearizability, shape invariant
inference, etc.  of fine-grained concurrent data structures.  While
our motivation is similar (automated inference of intermediate
assertions and local invariants), the context of study (transactions
vs. shared-memory concurrency), the objects being analyzed (relational
database tables vs. concurrent data structures), the properties being
verified (integrity constraints over relational tables vs. memory
safety, or linearizability of concurrent data structure operations)
and the analysis technique used to drive inference (state transformers
vs. abstract interpretation) are quite different.
\vspace*{-7pt}
\section{Conclusions}
\label{sec:conclusions}

To improve performance, modern database systems employ techniques that
weaken the strong isolation guarantees provided by serializable
transactions in favor of alternatives that allow a transaction to
witness the effects of other concurrently executing transactions that
happen commit during its execution.  Typically, it is the
responsibility of the database programmer to determine if an available
weak isolation level would violate a transaction's consistency
constraints.  Although this has proven to be a difficult and
error-prone process, there has heretofore been no attempt to formalize
notions of weak isolation with respect to application semantics, or
consider how we might verify the correctness of database programs that
use weakly-isolated transactions.  In this paper, we provide such a
formalization.  We develop a rely-guarantee proof framework cognizant
of weak isolation semantics, and build on this foundation to devise an
inference procedure that facilitates automated verification of
weakly-isolated transactions, and have applied our ideas on
widely-used database systems to justify their utility.  Our solution
enables database applications to leverage the performance advantages
offered by weak isolation, without compromising correctness.

\section*{Acknowledgements}

We thank KC Sivaramakrishnan for numerous helpful discussions about
weak isolation, and for thorough analysis of the material presented in
this paper. We are grateful to the anonymous reviewers, and our
shepherd, Peter M{\"u}ller, for their careful reading and insightful
comments.  This material is based upon work supported by the National
Science Foundation under Grant No. CCF-SHF 1717741 and the Air Force
Research Lab under Grant No.  FA8750-17-1-0006.


\small
\bibliography{all}

\newpage
\appendix

\section{Full Operational Semantics}

%
\textbf{Syntax}\\
\begin{smathpar}
\renewcommand{\arraystretch}{1.2}
\begin{array}{lclcl}
\multicolumn{5}{c} {
  {x,y} \in \mathtt{Variables}\qquad
  {f} \in \mathtt{Field\;Names} \qquad
  {i,j} \in \mathbb{N} \qquad
  {\odot} \in \{+,-,\le,\ge,=\}\qquad
  {k} \in \mathbb{Z}\cup\mathbb{B} \qquad
  {\rec} \in \{\bar{f}=\bar{k}\}\
}\\
{\stl,\stg,s} & \in & \mathtt{State} & \coloneqq &  \Pow{\{\bar{f}=\bar{k}\}} \\
{\I_e, \I_c }  & \in & \mathtt{Isolation Spec} & \coloneqq & (\stl,\stg,\stg') \rightarrow \Prop\\
v & \in & \mathtt{Values} & \coloneqq & k \ALT \rec \ALT s\\
e & \in & \mathtt{Expressions} & \coloneqq & k \ALT x \ALT x.f 
    \ALT \{\bar{f}=\bar{e}\} \ALT e_1 \odot e_2\\ 
c & \in & \mathtt{Commands} & \coloneqq & \cskip \ALT \lete{x}{e}{c}
    \ALT \ite{e}{c_1}{c_2}\ALT c_1;c_2 \ALT \inserte{x}  \\
&&&&\ALT \deletee{\lambda x.e}
    \ALT \lete{x}{\selecte{\lambda x.e}}{c}
    \ALT \updatee{\lambda x.e_1}{\lambda x.e_2}\\
&&&&\ALT \foreache{x}{\lambda y.\lambda z. c} 
    \ALT \foreachr{s_1}{s_2}{\lambda x.\lambda y. e}\\
&&&&\ALT \ctxn{i}{\I_e,\I_c}{ c } \ALT \ctxn{i}{\I_e,\I_c,\stl,\stg}{c} \ALT c1 || c2\\
\ectx & \in & \mathtt{Eval\;Ctx} & ::= & \bullet \ALT  
  \bullet || c_2 \ALT c_1 || \bullet \ALT \bullet;\,c_2 
  \ALT \ctxn{i}{\I_e,\I_c,\stl,\stg}{\bullet} \\
\end{array}
\end{smathpar}
\bigskip

\renewcommand{\arraystretch}{1.2}

\textbf{Local Reduction} \quad 
\fbox {\(\stg \vdash (c,\stl) \stepsto (c',\stl')\)}\\
\begin{minipage}{2.8in}
\rulelabel{E-Insert}
\begin{smathpar}
\begin{array}{c}
\RULE
{
  j \not\in \dom(\stl \cup \stg)\\
  r' = \langle r \;\C{with}\; \idf=j;\,\txnf=i;\,\delf=\C{false} \rangle
}
{
  \stg \vdash (\tbox{\inserte{r}}_i,\stl) \stepsto
  (\tbox{\cskip}_i,\stl \cup \{r'\})
}
\end{array}
\end{smathpar}
\end{minipage}
\begin{minipage}{2.8in}
\rulelabel{E-Delete}
\begin{smathpar}
\begin{array}{c}
\RULE
{
  s = \{r' \,|\, \exists(r\in\Delta).~ \eval([r/x]e)=\C{true} \\
        \hspace*{0.7in}\conj r'=\langle r \;\C{with}\;\delf=\C{true};\,\txnf=i \rangle\}\\
 \dom(s) \cap \dom(\delta) = \emptyset
}
{
  \stg \vdash (\tbox{\deletee{\lambda x.e}}_i,\stl) \stepsto (\tbox{\cskip}_i,\stl \cup s)
}
\end{array}
\end{smathpar}
\end{minipage}
\bigskip

%
\rulelabel{E-Select}
\begin{smathpar}
\begin{array}{c}
\RULE
{
  s = \{r\in\Delta \,|\, \eval([r/x]e)=\C{true}\}\spc
  c' = [s/y]c
}
{
  \stg \vdash (\tbox{\lete{y}{\selecte{\lambda x.e}}{c}}_i, \stl) \stepsto 
              (\tbox{c'}_i,\stl)
}
\end{array}
\end{smathpar}
%
%
\rulelabel{E-Update}
\begin{smathpar}
\begin{array}{c}
\RULE
{
  s = \{r' \,|\, \exists(r\in\Delta).~ \eval([r/x]e_2)=\C{true} \conj r'=\langle [r/x]e_1 \;\C{with}\;\\ 
\idf=r.\idf;\,\txnf=i;\,\delf=r.\delf \rangle\} \spc \dom(\stl) \cap \dom(s) = \emptyset
}
{
  \stg \vdash (\tbox{\updatee{\lambda x.e_1}{\lambda x.e_2}}_i,\stl) \stepsto 
              (\cskip,\stl \cup s)
}
\end{array}
\end{smathpar}
%

\begin{minipage}{2.8in}
\rulelabel{E-Seq1}
\begin{smathpar}
\begin{array}{c}
\RULE
{
 \stg \vdash (\tbox{c1}_i, \stl) \stepsto (\tbox{c1'}_i, \stl') \spc c1 \neq \cskip
}
{
  \stg \vdash (\tbox{c1;c2}_i, \stl) \stepsto 
              (\tbox{c1';c2}_i,\stl')
}
\end{array}
\end{smathpar}
\end{minipage}
\begin{minipage}{2.8in}
\rulelabel{E-Seq2}
\begin{smathpar}
\begin{array}{c}
\RULE
{
 \stg \vdash (\tbox{c1}_i, \stl) \stepsto (\tbox{\cskip}_i, \stl')
}
{
  \stg \vdash (\tbox{c1;c2}_i, \stl) \stepsto 
              (\tbox{c2}_i,\stl')
}
\end{array}
\end{smathpar}
\end{minipage}

\begin{minipage}{2.8in}
\rulelabel{E-IfTrue}
\begin{smathpar}
\begin{array}{c}
\RULE
{
 \eval(e) = \C{true}
}
{
  \stg \vdash (\tbox{\ite{e}{c_1}{c_2}}_i, \stl) \stepsto 
              (\tbox{c1}_i,\stl)
}
\end{array}
\end{smathpar}
\end{minipage}
\begin{minipage}{2.8in}
\rulelabel{E-IfFalse}
\begin{smathpar}
\begin{array}{c}
\RULE
{
 \eval(e) = \C{false}
}
{
  \stg \vdash (\tbox{\ite{e}{c_1}{c_2}}_i, \stl) \stepsto 
              (\tbox{c2}_i,\stl)
}
\end{array}
\end{smathpar}
\end{minipage}

\begin{smathpar}
\begin{array}{ll}
  \rulelabel{E-Foreach1} & \stg \vdash (\tbox{\foreache{s}{\lambda y.\lambda
  z.c}}_i,\stl) \stepsto (\tbox{\foreachr{\emptyset}{s}{\lambda y.\lambda z. c}}_i, \stl)\\
  \rulelabel{E-Foreach2} & \stg \vdash (\tbox{\foreachr{s_1}{\{r\} \uplus s_2}{\lambda y.\lambda
  z.c}}_i,\stl) \stepsto (\tbox{[r/z][s_1/y]c;\,\foreachr{s_1 \cup \{r\}}{s_2}{\lambda y.\lambda z. c}}_i, \stl)\\
  \rulelabel{E-Foreach3} & \stg \vdash (\tbox{\foreachr{s}{\emptyset}{\lambda y.\lambda
  z.c}}_i,\stl) \stepsto (\tbox{\cskip}_i,\stl)\\
\end{array}
\end{smathpar}
%

%
\textbf{Top-Level Reduction} \quad 
\fbox {\((c,\stg) \stepsto (c',\stg')\)}\\
%
  \rulelabel{E-Txn-Start}
  \begin{smathpar}
  \begin{array}{c}
    \RULE{}
         {(\ctxn{i}{\I_e,\I_c}{c},\stg) \stepsto (\ctxn{i}{\I_e,\I_c,\emptyset,\stg}{c},\stg)}
  \end{array}
  \end{smathpar}
\hfill
\rulelabel{E-Txn}
\begin{smathpar}
\begin{array}{c}
\RULE
{
  \I_e\,\,(\stl,\stg,\stg')\spc
  \stg \vdash (\tbox{c}_i,\stl) \stepsto (\tbox{c'}_i,\stl')
}
{
  (\ctxn{i}{\I_e,\I_c,\stl,\stg}{c},\stg') \stepsto
  (\ctxn{i}{\I_e,\I_c,\stl',\stg'}{c'},\stg')
}
\end{array}
\end{smathpar}

\begin{center}
\rulelabel{E-Commit}
\begin{smathpar}
\begin{array}{c}
\RULE
{
  \I_c\,\,(\stl,\stg,\stg')
}
{
  (\ctxn{i}{\I_e,\I_c,\stl,\stg}{\cskip},\stg') \stepsto (\cskip,\stl \rhd \stg')
}
\end{array}
\end{smathpar}
\end{center}
\hfill
%

\section{Rely-Guarantee Reasoning}

%
\textbf{Txn-Local Reasoning} \quad 
  \fbox {\( \R \vdash \hoare{P}{c}{Q} \)} \\
%
\rulelabel{RG-Insert}
\begin{smathpar}
\begin{array}{c}
\RULE
{
  \stable(\R,P)\\
  \forall\stl,\stl',\stg,i.~P(\stl,\stg) \conj j \not\in
  \dom(\stl\cup\stg) \\
  \conj \stl'=\stl \cup 
  \{\langle x \with \idf=j;\,\txnf=i;\,\delf=\C{false}\rangle\} \Rightarrow Q(\stl',\stg)
}
{
  \R \vdash \hoare{P}{\inserte{x}}{Q}
}
\end{array}
\end{smathpar}
\rulelabel{RG-Delete}
\begin{smathpar}
\begin{array}{c}
\RULE
{
  \stable(\R,P)\\
  \forall\stl,\stl',\stg.~P(\stl,\stg) \conj 
  \stl' = \stl \cup \{r' \,|\, \exists(r\in\Delta).~ [r/x]e \\
        \conj r'=\langle r \with \txnf=i; \delf=\C{true}\rangle\}
  \Rightarrow 
  Q(\stl',\stg)
}
{
  \R \vdash \hoare{P}{\deletee{\lambda x.e}}{Q}
}
\end{array}
\end{smathpar}
%

%
\rulelabel{RG-Update}
\begin{smathpar}
\begin{array}{c}
\RULE
{
  \stable(\R,P)\\
  \forall\stl,\stl',\stg.~P(\stl,\stg) \conj 
  \stl' = \stl \cup \{r' \,|\, \exists(r\in\Delta).[r/x]e_2 \conj\\
   r'=\langle[r/x]e_1 \with \idf=r.\idf;\,\txnf=i;\,\delf=\C{false}\rangle\} \Rightarrow   Q(\stl',\stg)
}
{
  \R \vdash \hoare{P}{\updatee{\lambda x.e_1}{\lambda x.e_2}}{Q}
}
\end{array}
\end{smathpar}
\rulelabel{RG-Select}
\begin{smathpar}
\begin{array}{c}
\RULE
{
   \R \vdash \hoare{P'}{c}{Q}\spc
  \stable(\R,P)\\
  P'(\stl,\stg) \Leftrightarrow P(\stl,\stg) \\
  \hspace*{0.5in}\wedge
  x = \{r' \,|\, \exists(r\in\Delta).~ [r/y]e_2\} \\
}
{
  \R \vdash \hoare{P}{\lete{y}{\selecte{\lambda x.e}}{c}}{Q}
}
\end{array}
\end{smathpar}
%

%
\rulelabel{RG-Foreach}
\begin{smathpar}
\begin{array}{c}
\RULE
{
  \stable(\R,Q)\spc
  \stable(\R,\psi)\\
  P \Rightarrow [y/\phi]\psi\spc
  \R \vdash \hoare{\psi \wedge z\in x}{c}{Q_c}\\
  Q_c  \Rightarrow [y \cup \{z\} / y]\psi \spc
  [x / y]\psi \Rightarrow Q
}
{
  \R \vdash \hoare{P}{\foreache{x}{\lambda y.\lambda z.c}}{Q}
}
\end{array}
\end{smathpar}
\rulelabel{RG-Seq}
\begin{smathpar}
\begin{array}{c}
\RULE
{
\R \vdash \hoare{P} {c1} {Q^{'}} \spc \R \vdash \hoare{Q^{'}}{c2}{Q}\\
\stable(\R, Q^{'})
}
{
  \R \vdash \hoare{P}{c1;c2}{Q}
}
\end{array}
\end{smathpar}

\rulelabel{RG-If}
\begin{smathpar}
\begin{array}{c}
\RULE
{
\R \vdash \hoare{P \wedge e} {c1} {Q} \spc \R \vdash \hoare{P \wedge \neg e}{c2}{Q}\\
\stable(\R, P)
}
{
  \R \vdash \hoare{P}{\ite{e}{c_1}{c_2}}{Q}
}
\end{array}
\end{smathpar}
%
\rulelabel{RG-Conseq}
\begin{smathpar}
\begin{array}{c}
\RULE
{
  \R \vdash \hoare{P}{c}{Q}\\
  P' \Rightarrow P \spc
  Q \Rightarrow Q' \spc
  \stable(\R,P')\spc
  \stable(\R,Q')\spc
}
{
  \R \vdash \hoare{P'}{c}{Q'}
}
\end{array}
\end{smathpar}
%
\bigskip

\textbf{Top-Level Reasoning} \quad \fbox {\(
\rg{I,R}{c}{G,I} \)}\\
%
%
\rulelabel{RG-Txn}
\begin{smathpar}
\begin{array}{c}
\RULE
{
  \stable(R,\I)\spc
  \stable(R,I)\spc
  P(\stl,\stg) \Leftrightarrow \stl=\emptyset \wedge I(\stg)\\
  \R_e = R \backslash \I_e \spc \R_c = R \backslash \I_c \spc 
   \R_e \vdash \rg{P}{c}{Q} \spc \stable(\R_c,Q) \\ 
  \forall \stl,\stg.~ Q(\stl,\stg) \Rightarrow 
    G(\stg, \stl \rhd \stg)\spc
  \forall \stg,\stg'.~I(\stg) \wedge G(\stg,\stg') \Rightarrow I(\stg')\\
}
{
  \rg{I,R}{\ctxn{i}{\I}{c}}{G,I}
}
\end{array}
\end{smathpar}
%

%
\rulelabel{RG-Par}
\begin{smathpar}
\begin{array}{c}
\RULE
{
 \rg{I,R \cup G_2}{t_1}{G_1,I} \\
 \rg{I,R \cup G_1}{t_2}{G_2,I}
}
{\rg{I,R}{t_1||t_2}{G_1\cup G_2,I}}
\end{array}
\end{smathpar}
%
%
\rulelabel{RG-Conseq2}
\begin{smathpar}
\begin{array}{c}
\RULE
{
  \rg{I,R}{\ctxn{i}{\I}{c}}{G,I}\\
  \I' \Rightarrow \I \spc 
  R' \subseteq R \spc 
  \stable(R',\I')\\
  G \subseteq G' \spc
  \forall \stg,\stg'.~I(\stg) \wedge G'(\stg,\stg') \Rightarrow I(\stg')\\
}
{
  \rg{I,R'}{\ctxn{i}{\I'}{c}}{G',I}
}
\end{array}
\end{smathpar}
%


\section{Soundness of RG-Reasoning}
\begin{definition}[ Step-indexed reflexive transitive closure]
For all $A:\text{Type}$, $R: A \rightarrow A \rightarrow \mathbb{P}$, and $n :
\mathbb{N}$, the step-indexed reflexive transitive closure $R^n$ of $R$ is
the smallest relation satisfying the following
properties:
\begin{itemize}
\item $\forall (x:A).\, R^0 (x,x)$
\item $\forall (x,y,z : A).\, R(x,y) \conj R^{n-1}(y,z) \Rightarrow
R^{n}(x,z)$
\end{itemize}
\end{definition}

\begin{definition} [Interleaved step relation]
The interleaved step relation (denoted as $\rightarrow_R$) interleaves transaction local reduction with interference from concurrent transactions captured as the Rely relation (R). It is defined as follows:
\begin{mathpar}
\begin{array}{lcl}
(t, \stg) \rightarrow_{R}  (t^{'}, \stg^{'}) & \defeq & (t = t^{'} \wedge R(\stg, \stg^{'}))  \vee ((t, \stg) \rightarrow (t^{'}, \stg^{'}))
\end{array}
\end{mathpar}

The interleaved multistep relation (denoted as $\rightarrow_R^n$) is the step-indexed reflexive transitive closure of $\rightarrow_R$.
\end{definition}

Given a transaction $t = \texttt{txn}\langle \mathbb{I}, \stl, \stg \rangle \{c\}$, we use the notation $t.\stl, t.\stg, t.\mathbb{I}$ and $t.c$ to denote the various components of $t$. Below, we provide a more precise definition of the transaction-local RG judgement:

\begin{mathpar}
\begin{array}{lcl}
\R \vdash \hoare{P}{c}{Q} & \defeq & \forall t, \stg,\stg'.
  P(t.\stl,\stg) \conj t.c = c \wedge (t,\stg) \rightarrow_{\R}^{n} (t_2, \stg^{'}) \rightarrow (t^{'}, \stg^{'}) \wedge t^{'}.c = \texttt{SKIP} \wedge Q(t^{'}.\stl, \stg^{'})
\end{array}
\end{mathpar}

Note that even though $\R$ is a ternary relation, its step-indexed reflexive-transitive closure can be defined in a similar fashion as R. In the above definition, we have explicitly stated that the last step in the reduction sequence is taken by the transaction (and not by the environment),  finishing in the state satisfying the assertion $Q$. The nature of interference before and after the last step of the transaction are different (after the last step and before the commit step, the interference is controlled by $\mathbb{I}_c$, while before the last step, the interference is controlled by $\mathbb{I}_e$).

\begin{lemma}
If $\stable(R, Q)$, then $\forall \stl, \stg, \stg', k. Q(\stl, \stg) \wedge R^{k}(\stg, \stg') \Rightarrow Q(\stl, \stg')$
\end{lemma}

\begin{proof}
We use induction on $k$. 

\textbf{Base Case}: For $k = 0$, $\stg = \stg'$ and hence $Q(\stl, \stg')$. 

\textbf{Inductive Case}: For the inductive case, assume that for $k'$, $\forall \stl, \stg, \stg'. Q(\stl, \stg) \wedge R^{k'}(\stg, \stg') \Rightarrow Q(\stl, \stg')$.
Given $\stl, \stg, \stg_1$ such that $Q(\stl, \stg)$, $R^{k'+1}(\stl, \stg_1)$, we have to show $Q(\stl, \stg_1)$. There exists $\stg'$ such that $R^{k'}(\stg, \stg')$ and $R(\stg', \stg_1)$. By the inductive hypothesis, $Q(\stl, \stg')$. $\stable(R,Q)$ is defined as follows:
$$
\stable(R,Q) = \forall \stl, \stg, \stg'. Q(\stl, \stg) \wedge R(\stg, \stg') \Rightarrow Q(\stl, \stg')
$$
Instantiating the above statement with $\stl, \stg', \stg_1$, we get $Q(\stl, \stg_1)$
\end{proof}

\begin{theorem}
RG-Txn is sound.
\end{theorem}
\begin{proof}
\begin{mathpar}
\begin{array}{lc}
  \stable(R, I) & HI\\
  \stable(R,\I) & H\I \\
  P(\stl,\stg) \Leftrightarrow \stl=\emptyset \wedge I(\stg) & HP\\ 
  \R_e(\stl,\stg,\stg') \Leftrightarrow \exists \stg_1. R(\stg, \stg') \wedge \I_e(\stl, \stg_1, \stg) \wedge  \I_e(\stl, \stg_1, \stg') & H\R_l \\
  \R_e\vdash \rg{P}{c}{Q} & Hc\\
  \R_c(\stl,\stg,\stg') \Leftrightarrow \exists \stg_1. R(\stg, \stg') \wedge \I_c(\stl, \stg_1, \stg) \wedge  \I_c(\stl, \stg_1, \stg') & H\R_c \\
  \stable(\R_c,Q) & HQ \\
  \forall \stl,\stg.~ Q(\stl,\stg) \Rightarrow 
    G(\stg, \stl \rhd \stg)\spc  & HQG \\
  \forall \stg,\stg'.~I(\stg) \wedge G(\stg,\stg') \Rightarrow I(\stg') & HG\\
\end{array}
\end{mathpar}

Let $t_s = \ctxn{i}{\I}{c}$. Consider $\stg$ such that $I(\stg)$, and let $(t_s, \stg) \rightarrow_{R}^{n} (\texttt{SKIP}, \stg^{'})$. We have to show (1) $I(\stg^{'})$ and (2) $\texttt{step-guaranteed}(R, G, t_s, \stg)$. We break down the sequence of reductions into four parts :

\begin{itemize}
\item $\pi_1 = (t_s, \stg) \rightarrow_{R}^{n_1} (t_s, \stg_1) \rightarrow (t, \stg_1)$, where initially only the environment takes steps and the last step in the sequence is the start of the transaction using the rule E-Txn-Start. 
\item $\pi_2 = (t, \stg_1) \rightarrow_{R}^{n_2} (t^{'}, \stg_2)$, which begins from $t$ taking its first step at state $\stg_1$ and ends at the first configuration where $t^{'}.c = \texttt{SKIP}$.
\item $\pi_3 = (t^{'}, \stg_2) \rightarrow_{R}^{n_3} (\texttt{SKIP}, \stg_3)$ which ends at the step where $t$ commits.
\item $\pi_4 = (\texttt{SKIP}, \stg_3) \rightarrow_{R}^{n_4} (\texttt{SKIP}, \stg^{'})$ where only the environment takes a step.
\end{itemize}

In the sequence $\pi_1$, $R^{n_1}(\stg, \stg_1)$. By $I(\stg)$, $HI$ and Lemma 3.3, $I(\stg_1)$. By the rule E-Txn-Start, $t.\stl = \phi, t.\stg = \stg_1$ and $t.c = c$. Hence $P(t.\stl, \stg_1)$.

Expanding the definition of the assertion $Hc$ and instantiating it with $\stg = \stg_1$ and $\stg^{'} = \stg_2$, we would get $Q(t^{'}.\stl, \stg_2)$. However, the environment steps in sequence $\pi_2$ are in $R$, while the environment steps in assertion $Hc$ are in $\R_e$. By definition of $\R_e$, every step of $\R_e$ corresponds to a unique step in $R$. We will now show that only those environment steps in R which correspond to steps in $\R_e$ can happen in the sequence $\pi_2$. We will show this in two steps. In the first step, we will prove that for all configurations $(t_p, \stg_p)$ in the sequence $\pi_2$ except possibly the last configuration, $\I_e(t_p.\stl, t_p.\stg, \stg_p)$. 

We will prove this by contradiction. Assume that there is a configuration $(t_1, \stg_b)$ such that $\neg \I_e( t_1.\stl, t_1.\stg, \stg_b)$. Let $(t_1, \stg_{b}^{'}) \rightarrow (t_{1}^{'}, \stg_{b}^{'})$ be the next step in $\pi_2$ taken by the transaction. We know that this step always exists because the last step in $\pi_2$ is taken by the transaction.  Then $\mathbb{I}_e(t_1.\stl, t_1.\stg, \stg_{b}^{'})$. All steps between $(t_1, \stg_b)$ and $(t_1, \stg_{b}^{'})$ are taken by the environment, i.e. $R^{k}(\stg_b, \stg_{b}^{'})$ for some $k$. However,  $\neg \I_e(t_1.\stl, t_1.\stg, \stg_b)$, the assertion $H\I$ and a simple induction on $k$ would imply that $\neg \I_e(t_1.\stl, t_1.\stg, \stg_{b}^{'})$. This is a contradiction. Hence, for all configurations $(t_p, \stg_p)$ in the sequence $\pi_2$ except possibly the last configuration, $\I_e(t_p.\stl, t_p.\stg, \stg_p)$.

Now, we will show that every environment step in $\pi_2$ corresponds to a step in $\R_e$. Assume that $(t_1, \stg_a) \rightarrow_R (t_1, \stg_b)$ is an environment step such that $R(\stg_a, \stg_b)$. Then, we know that $\I_e(t_1.\stl, t_1.\stg, \stg_b)$ and $\I_e(t_1.\stl, t_1.\stg, \stg_a)$. Hence, $t_1.\stg$ provides the existence of $\stg_1$ in the definition of $\R_e$. Thus, $\R_e(t_1.\stl, \stg_a, \stg_b)$. This implies that we can use $Hc$ and make the assertion $Q(t^{'}.\stl, \stg_2)$.

Note that $t^{'}.\stg = \stg_2$. Also, since all the changes in the global database state have so far been made by the environment, $I(\stg_2)$.

$\pi_3 = (t^{'}, \stg_2) \rightarrow_{R}^{n_3 - 1} (t^{'}, \stg_{2}^{'}) \rightarrow (\texttt{SKIP}, \stg_3)$, where the first $n_3-1$ steps are only performed by the environment. Since the transaction commits at state $\stg_{2}^{'}$, by the E-Commit rule, $\I_c(t^{'}.\stl, \stg_2, \stg_{2}^{'})$. We will now show that all environment steps in the above sequence must be correspond to steps in $\R_c$. Again, we will show this in two steps. Let $m = n_3 - 1$ and $(t^{'}, \stg_2) \rightarrow_R (t^{'}, \stg_{21}) \rightarrow_R (t^{'}, \stg_{22}) \ldots \rightarrow_R (t^{'}, \stg_{2m}) \rightarrow (\texttt{SKIP}, \stg_3)$. We will show that $\I_c(t^{'}.\stl, \stg_2, \stg_{2k})$ for all $k, 1 \leq k \leq m$. 

We will prove this by contradiction. Suppose for some $i$, $\neg \I_c(t^{'}.\stl, \stg_2, \stg_{2i})$. Consider $j$ such that $R^{j}(\stg_{2i}, \stg^{'})$. Then, by $H\I$ and a simple induction on $j$, we can show that $\neg \I_c(t^{'}.\stl, \stg_2, \stg_{2}^{'})$. However, this is a contradiction. Hence, $\forall k$, $\I_c(t^{'}.\stl, \stg_2, \stg_{2k})$.

Now, we will show that every environment step in $\pi_3$ corresponds to a step in $\R_c$. Consider the step $(t^{'}, \stg_{2k}) \rightarrow_R (t^{'}, \stg_{2(k+1)})$. We have $\I_c(t^{'}.\stl, \stg_2, \stg_{2(k+1)})$ and $\I_c(t^{'}.\stl, \stg_2, \stg_{2k})$. Hence, $\stg_2$ provides the existence of $\stg_1$ in the definition of $\R_c$. Thus, $\R_c(t^{'}.\stl, \stg_{2k}, \stg_{2(k+1)})$.

By $HQ$, $Q(t^{'}.\stl, \stg_2)$ and by $\stable(\R_c, Q)$ we have $Q(t^{'}.\stl, \stg_{2}^{'})$. Since all state changes upto $\stg_{2}^{'}$ have been made by the environment, $I(\stg_{2}^{'})$. By $HQG$, $G(\stg_{2}^{'}, t^{'}.\stl \rhd \stg_{2}^{'})$. By the E-Commit rule, $\stg_3 = (t^{'}.\stl \rhd \stg_{2}^{'})$. Hence, $G(\stg_{2}^{'}, \stg_3)$. All the steps of the transaction except the commit step do not change the global database state and the commit step satisfies $G$. This proves the \texttt{step-guaranteed} assertion. Finally, by $HG$, $I(\stg_3)$.


All the steps in $(\texttt{SKIP}, \stg_3) \rightarrow_{R}^{n_4} (\texttt{SKIP}, \stg^{'})$ are performed by the environment. Since $I(\stg_3)$, by $HI$ and Lemma 3.3, $I(\stg^{'})$.
\end{proof}


\begin{theorem}
RG-Select is sound
\end{theorem}

\begin{proof}
Given the premises of RG-Select, $t, \stg$ such that $t.c \sim \lete{x}{\selecte{\lambda y.e}}{c}$, $(t, \stg) \rightarrow_{\R}^{m} (t_2, \stg^{'}) \rightarrow (t^{'}, \stg^{'})$, $P(t.\stl, \stg)$ and $t^{'}.c = \texttt{SKIP}$, we have to show that $Q(t^{'}.\stl, \stg^{'})$. The reduction sequence can be broken down into following parts:

\begin{itemize}
\item $\pi_1 = (t, \stg) \rightarrow_{\R}^{n_1} (t, \stg_1) \rightarrow (t_1, \stg_1)$ where initially only the environment takes steps, and ends with the application of the E-Select rule.
\item $\pi_2 = (t_1, \stg_1) \rightarrow_{\R}^{n_2} (t_2, \stg^{'}) \rightarrow (t^{'}, \stg^{'}) $ which corresponds to the execution of $\C{c}$
\end{itemize}

In $\pi_1$, $\R^{n_1}(t.\stl, \stg, \stg_1)$. By $P(t.\stl, \stg)$ and $\stable(\R, P)$, we get $P(t.\stl, \stg_1)$. By applying the E-Select rule, $t_1.\stl = t.\stl$, $t_1.c = [s/x]c$, where $s = \{r\in\Delta_1 \,|\, \eval([r/y]e)=\C{true}\}$. By definition of $P'$, $P'(t_1.\stl, \stg_1)$. The following property holds trivially:

$$
\R \vdash \hoare{P \wedge x = s} {c} {Q} \Leftrightarrow \R \vdash \hoare{P} {[s/x]c} {Q}
$$

Since $\R \vdash \hoare{P'}{c}{Q}$, by the above property, $\R \vdash \hoare{P}{[s/x]c}{Q}$. Since $P(t_1.\stl, \stg_1)$, by definition of $\R \vdash \hoare{P}{[s/x]c}{Q}$, we get $Q(t^{'}.\stl, \stg')$.


\end{proof}

\begin{theorem}
RG-Update is sound
\end{theorem}

\begin{proof}
Given the premises of RG-Update, $t, \stg$ such that $t.c = \updatee{\lambda x.e_1}{\lambda x.e_2}$, $(t, \stg) \rightarrow_{\R}^{m} (t_2, \stg^{'}) \rightarrow (t^{'}, \stg^{'})$, $P(t.\stl, \stg)$ and $t^{'}.c = \texttt{SKIP}$, we have to show that $Q(t^{'}.\stl, \stg^{'})$. 

Since only a single step needs to be taken by the transaction (by applying the E-Update rule), $t_2.c = t.c$, $t_2.\stl = t.\stl$ and $\R^m(t.\stl, \stg, \stg')$. By $\stable(\R, P)$, $P(t_2.\stl, \stg')$. According to E-Update,  $t^{'}.\stl = t_2.\stl \cup  \{r' \,|\, \exists(r\in \stg').~ \eval([r/x]e_2)=\C{true} \conj r'=\langle[r/x]e_1 \with \idf=r.\idf;\,\txnf=i;\,\delf=\C{false}\rangle\ \}$. From the premise of RG-Update, we know that 

\begin{eqnarray*}
\forall\stl,\stl',\stg.~P(\stl,\stg) \conj 
  \stl' = \stl \cup \{r' \,|\, \exists(r \in \stg).~ [r/x]e_2 =\C{true} \\
  \conj r'=\langle[r/x]e_1 \with \idf=r.\idf;\,\txnf=i;\,\delf=\C{false}\rangle\} \Rightarrow 
  Q(\stl',\stg)
\end{eqnarray*}

Instantiating the above statement with $\stl = t_2.\stl$ and $\stg = \stg'$, we get $Q(\stl',\stg')$. However, $\stl^{'} = t^{'}.\stl$. Hence, $Q(t^{'}.\stl,\stg')$.
\end{proof}

\begin{theorem}
RG-Insert is sound
\end{theorem}

\begin{proof}
Given the premises of RG-Insert, $t, \stg$ such that $t.c = \inserte{x}$, $(t, \stg) \rightarrow_{\R}^{m} (t_2, \stg^{'}) \rightarrow (t^{'}, \stg^{'})$, $P(t.\stl, \stg)$ and $t^{'}.c = \texttt{SKIP}$, we have to show that $Q(t^{'}.\stl, \stg^{'})$. 

Since only a single step needs to be taken by the transaction (by applying the E-Insert rule), $t_2.c = t.c$, $t_2.\stl = t.\stl$ and $\R^m(t.\stl, \stg, \stg')$. By $\stable(\R, P)$, $P(t_2.\stl, \stg')$. Since the transaction takes is able to take the step according to the rule E-Insert, $P$ must assert that $x$ is bound to a record $r$, and $x$ must have been substituted with $r$, thus $t^{'}.\stl = t_2.\stl \cup \{ \langle r \;\C{with}\; \idf=j;\,\txnf=i;\,\delf=\C{false} \rangle\}$ and $j \not\in \dom(t_2.\stl \cup \stg')$. From the premise of RG-Insert, we know that 

$$
\forall\stl,\stl',\stg,i.~P(\stl,\stg) \conj j \not\in
  \dom(\stl\cup\stg) 
  \conj \stl'=\stl \cup 
  \{\langle x \;\C{with}\; \idf=j;\,\txnf=i;\,\delf=\C{false} \rangle\} \Rightarrow 
  Q(\stl',\stg)
$$

Instantiating the above statement with $\stl = t_2.\stl$ and $\stg = \stg'$, we get $Q(\stl',\stg')$. Since $x=r$, $\stl^{'} = t^{'}.\stl$. Hence, $Q(t^{'}.\stl,\stg)$.
\end{proof}

\begin{theorem}
RG-Delete is sound
\end{theorem}

\begin{proof}
Given the premise of RG-Delete $t, \stg$ such that $t.c = \deletee{\lambda x.e}$, $(t, \stg) \rightarrow_{\R}^{m} (t_2, \stg^{'}) \rightarrow (t^{'}, \sigma^{'})$, $P(t.\stl, \stg)$ and $t^{'}.c = \texttt{SKIP}$, we have to show that $Q(t^{'}.\stl, \stg^{'})$. 

Since only a single step needs to be taken by the transaction (by applying the E-Delete rule), $t_2.c = t.c$, $t_2.\stl = t.\stl$ and $\R^m(t.\stl, \stg, \stg')$. By $\stable(\R, P)$, $P(t_2.\stl, \stg')$. According to E-Delete,  $t^{'}.\stl = t_2.\stl \cup \{r' \,|\, \exists(r\in\stg').~ \eval([r/x]e)=\C{true} \conj r'=\langle r \;\C{with}\; \idf=r.\idf;\,\txnf=i;\,\delf=\C{true}\}\}$. From the premise of RG-Delete, we know that 

$$
\forall \stl,\stl',\stg.~P(\stl,\stg) \conj 
  \stl' = \stl \cup \{r' \,|\, \exists(r\in\Delta).~ [r/x]e=\C{true}
        \conj r'=\langle r \;\C{with}\; \idf=r.\idf;\,\txnf=i;\,
        \delf=\C{true} \rangle\}  \Rightarrow 
  Q(\stl',\stg)
$$

Instantiating the above statement with $\stl = t_2.\stl$ and $\stg = \stg'$, we get $Q(\stl',\stg')$. However, $\stl^{'} = t^{'}.\stl$. Hence, $Q(t^{'}.\stl,\stg')$.
\end{proof}


\begin{theorem}
RG-Foreach is sound
\end{theorem}

\begin{proof}
\begin{mathpar}
\begin{array}{lc}
\stable(\R,Q) & HQ\\
  \stable(\R,\psi) & HI\\
  \stable(\R,P) & HP\\
  P  \Rightarrow [\phi / y]\psi & H1\\
  \R \vdash \hoare{\psi \wedge z\in x}{c}{Q_c} & Hc\\
  Q_c \Rightarrow [y \cup \{z\} / y]\psi & H2\\
 \end{array}
  \end{mathpar}
Given $t, \stg$ such that $t.c = \foreache{x}{\lambda y.\lambda z.c}$, $(t, \stg) \rightarrow_{\R}^{n} (t_2, \stg^{'}) \rightarrow (t^{'}, \stg^{'})$, $P(t.\stl, \stg)$ and $t^{'}.c = \texttt{SKIP}$, we have to show that $Q(t^{'}.\stl, \stg^{'})$. 

The operational semantics of \texttt{foreach} (E-Foreach1, E-Foreach2, E-Foreach3) essentially execute the command $\texttt{c}$ for a number of iterations, where in each iteration, $\texttt{z}$ is bound to a record $r \in \texttt{x}$, while $\C{y}$ is bound to a set containing records which were bound to $\C{z}$ in previous iterations. $\C{z}$ is bound to a different record in each iteration, and the loop stops when all records in $\C{x}$ are iterated over. 

Assuming that $|x| = s$, the reduction sequence for \texttt{foreach} will have the following structure : 

$$
(t, \stg) \rightarrow_{\R}^{m} (t_{1}, \stg_{1}) \rightarrow_{\R}^{n_1} (t_{2}^{'}, \stg_{2}^{'}) \rightarrow_{\R}^{n_{1}^{'}} (t_{2}, \stg_{2}) \rightarrow_{\R}^{n_2} (t_{3}^{'},\stg_{3}^{'}) \rightarrow_{\R}^{n_{2}^{'}} (t_{3}, \stg_{3}) \ldots (t_{s}, \stg_{s}) \rightarrow_{\R}^{n_s} (t_{s+1}^{'}, \stg_{s+1}^{'}) \rightarrow_{\R}^{l} (t^{'}, \stg^{'})
$$
 
The reduction sequence $\pi_i = (t_{i}, \stg_{i}) \rightarrow_{\R}^{n_i} (t_{i+1}^{'}, \stg_{i+1}^{'})$ corresponds to the execution of the command $\C{c}$ in the $i$th iteration, such that the first and last steps in $\pi_i$ are not environment steps. The sequence $\pi_0 = (t, \stg) \rightarrow_{\R}^{m} (t_{1}, \stg_{1})$ corresponds to the steps E-Foreach1 and E-Foreach2 along with environment steps. Similarly, the sequence $\pi_{i}^{'} = (t_{i+1}^{'}, \stg_{i+1}^{'}) \rightarrow_{\R}^{n_{1}^{'}} (t_{i+1}, \stg_{2})$ corresponds to the execution of the E-Foreach2 step required to prepare the $(i+1)$th iteration along with environment steps.  

Let $x = \{r_1, \ldots, r_s\}$, and assume that the records are picked in the increasing order. Then at the start of the $i$th iteration, $\C{z}$ is bound to $r_i$, while $\C{y}$ is bound to $\{r_1, \ldots, r_{i-1}\}$. We will show that $[\{r_1,\ldots, r_i\} / y]\psi$ holds at the end of iteration $i$, for all $1 \leq i \leq s$. More precisely, we will show $[\{r_1,\ldots, r_i\} / y]\psi(t_{i+1}^{'}.\stl, \stg_{i+1}^{'})$. We will use induction on $i$.

\textbf{Base Case}: The steps E-Foreach1 and E-Foreach2 do not change $\stl$. Also, $P(t.\stl, \stg)$ and $\stable(\R, P)$. Hence, at the end of the sequence $\pi_0$, $P(t_{1}.\stl, \stg_{1})$. By $H1$, this implies $[\phi / y] \psi(t_1.\stl, \stg_1)$. The sequence $\pi_1 = (t_{1}, \stg_{1}) \rightarrow_{\R}^{n_1} (t_{2}^{'}, \stg_{2}^{'})$ corresponds the execution of $\C{c}$ in the first iteration with $\C{z}$ bound to $r_1$ and $\C{y}$ bound to $\phi$. Clearly, $\psi(t_1.\stl, \stg_1) \wedge z \in x$ holds. Hence, by $Hc$, $Q_c(t_{2}^{'}.\stl, \stg_{2}^{'})$. By H2, this implies $[\{r_1\} / y] \psi(t_2^{'}.\stl, \stg_{2}^{'})$.

\textbf{Inductive Case}: Assume that $[\{r_1,\ldots, r_{k-1}\} / y]\psi(t_{k}^{'}.\stl, \stg_{k}^{'})$. The next sequence of reductions $(t_{k}^{'}, \stg_{k}^{'}) \rightarrow_{\R}^{n_{k}^{'}} (t_k, \stg_k)$ only corresponds to the execution of the E-Foreach2 step for the $k$th iteration and environment steps. E-Foreach2 does not change $\stl$, and since $\stable(\R, \psi)$, we get $[\{r_1,\ldots, r_k\} / y]\psi(t_{k}.\stl, \stg_{k})$. At the start of the next iteration, $\C{z}$ is bound to $r_{k}$, and $\C{y}$ is bound to $\{r_1, \ldots, r_{k-1}\}$. Hence, $\psi(t_k.\stl, \stg_k) \wedge z \in x$. By Hc, this implies $Q_c(t_{k+1}^{'}.\stl, \stg_{k+1}^{'})$. By $H2$, this implies $[y \cup z / y] \psi(t_{k+1}^{'}.\stl, \stg_{k+1}^{'}) = [\{r_1, \ldots, r_k\} / y]\psi(t_{k+1}^{'}.\stl, \stg_{k+1}^{'})$. This proves the inductive step.

Hence, at the end of the $s$th iteration, $[x / y]\psi(t_{s+1}^{'}.\stl, \stg_{s+1}^{'})$. This implies $Q(t_{s+1}^{'}.\stl, \stg_{s+1}^{'})$. Finally, the last part of the reduction, $(t_{s+1}^{'}, \stg_{s+1}^{'}) \rightarrow_{\R}^{l} (t^{'}, \stg^{'})$ corresponds environment steps and E-Foreach3 (as the last step). Since $\stable(\R, Q)$ and E-Foreach3 does not change $\stl$, we have $Q(t^{'}.\stl, \stg)$. 

\end{proof}

\begin{theorem}
RG-Seq is sound
\end{theorem}

\begin{proof}
\begin{mathpar}
\begin{array}{lc}
\hoare{P} {c1} {Q^{'}} & H1\\
 \hoare{Q^{'}}{c2}{Q} & H2\\
\stable(\R, Q^{'}) & H3
\end{array}
\end{mathpar}
Given $t, \stg$ such that $t.c = c1;c2$, $(t, \stg) \rightarrow_{\R}^{m} (t_2, \stg^{'}) \rightarrow (t^{'}, \stg^{'})$, $P(t.\stl, \stg)$ and $t^{'}.c = \texttt{SKIP}$, we have to show that $Q(t^{'}.\stl, \stg^{'})$. We can divide the reduction sequence into three parts :

\begin{itemize}
\item $(t,\stg) \rightarrow_{\R}^{m_1} (t_{m}^{'}, \stg_1) \rightarrow (t_{m}, \stg_1)$, where $t_m.c = c2$. We denote this sequence as $\pi_1$.
\item $(t_m, \stg_1) \rightarrow_{\R}^{m_2} (t_m, \stg_{1}^{'}) $ where all steps are taken by the environment. This sequence is denoted as $\pi_2$.
\item $(t_m, \stg_{1}^{'}) \rightarrow_{\R}^{m_3} (t_{2}, \stg^{'}) \rightarrow (t^{'}, \stg^{'})$. This sequence is denoted as $\pi_3$.
\end{itemize}

By the premise of the E-Seq1 and E-Seq2 rules, all the reductions in the sequence $\pi_1$ are also applicable to $c1$. Hence, consider transaction $s$ such that $s.c = c1$, $s.\stl=t.\stl$. Then, there exists the sequence $(s, \stg)  \rightarrow_{\R}^{m_1} (s_2, \stg_1) \rightarrow (s^{'}, \stg_1)$ with $s^{'}.c = \texttt{SKIP}$, $s^{'}.\stl = t_{m}.\stl$. Since $P(s.\stl, \stg)$, by H1, $Q^{'}(s^{'}.\stl, \stg_1)$. This implies $Q^{'}(t_m.\stl, \stg_1)$.

In the sequence $\pi_2$, all steps are taken by the environment. By H3, $Q^{'}(t_m.\stl, \stg_{1}^{'})$.

Since $t_m.c = c2$, by H3, $Q(t^{'}.\stl, \stg^{'})$.



\end{proof}

\begin{theorem}
RG-If is sound
\end{theorem}
\begin{proof}
\begin{mathpar}
\begin{array}{lc}
\hoare{P \wedge e} {c1} {Q} & H1 \\
\hoare{P \wedge \neg e}{c2}{Q} & H2\\
\stable(\R, P) & H3
\end{array}
\end{mathpar}

Given $t, \stg$ such that $t.c = \ite{e}{c_1}{c_2}$, $(t, \stg) \rightarrow_{\R}^{m} (t_2, \stg^{'}) \rightarrow (t^{'}, \stg^{'})$, $P(t.\stl, \stg)$ and $t^{'}.c = \texttt{SKIP}$, we have to show that $Q(t^{'}.\stl, \stg^{'})$. Assume that $\texttt{eval}(e) = \texttt{true}$. We divide the sequence of steps into two parts:

\begin{itemize}
\item $\pi_1 = (t, \stg) \rightarrow_{\R}^{n_1} (t, \stg_1) \rightarrow (t_1, \stg_1)$ where initially only the environment takes steps, and the last step is taken by the transaction using E-IfTrue.
\item $\pi_2 = (t_1, \stg_1) \rightarrow_{\R}^{n_2} (t_2, \stg^{'}) \rightarrow (t^{'}, \stg^{'})$.
\end{itemize}

Since $P(t.\stl, \stg)$ and $R^{n_1}(t.\stl, \stg, \stg_1)$, by H3, we have $P(t.\stl, \stg_1)$. By applying the rule E-IfTrue, we have $t_1.\stl = t.\stl$, $t_1.c = c1$. Hence, $P(t_1.\stl, \stg_1)$. By the definition of $H1$, $Q(t^{'}, \stg')$. A similar proof follows for the case $\texttt{eval}(e) = \texttt{false}$

\end{proof}

%

\begin{lemma}
If $\stable(\R, Q)$ and $\R' \subseteq \R$, then $\stable(\R', Q)$
\end{lemma}
\begin{proof}
Given $\stl, \stg, \stg'$ such that $Q(\stl, \stg)$ and $\R'(\stl, \stg, \stg')$, we have to show that $Q(\stl, \stg')$. Since $\R' \subseteq \R$, $\R(\stg, \stg')$. Hence, by $\stable(\R,Q)$, $Q(\stl, \stg')$.
\end{proof}

\begin{lemma}
If $\rg{I,R}{\ctxn{i}{\I}{c}}{G,I}$ and $R^{'} \subseteq R$, then $\rg{I,R'}{\ctxn{i}{\I}{c}}{G,I}$
\end{lemma}
\begin{proof}
Let $t = \ctxn{i}{\I}{c}$. Then, given $\stg$ such that $I(\stg)$ and $(t, \stg) \rightarrow_{R'}^{n} (\C{SKIP}, \stg^{'})$, we have to show (1) $I(\stg^{'})$ and (2) $\texttt{step-guaranteed}(R', G \cup ID, t, \stg)$. Since $R' \subseteq R$, every environment step in the above reduction sequence is in $R$. Thus,  $(t, \stg) \rightarrow_{R'}^{n} (\C{SKIP}, \stg^{'})$, which by definition of $\rg{I,R}{\ctxn{i}{\I}{c}}{G \cup ID,I}$ implies $I(\stg')$. The same argument holds for $\texttt{step-guaranteed}(R', G \cup ID, t, \stg)$.
\end{proof}

\begin{theorem}
RG-Par is sound
\end{theorem}

\begin{proof}
\begin{mathpar}
\begin{array}{lc}
 \rg{I,R \cup G_2}{t_1}{G_1,I} & H1\\
 \rg{I,R \cup G_1}{t_2}{G_2,I} & H2
\end{array}
\end{mathpar}

Consider $\stg$ such that $I(\stg)$, and let $(t_1 || t_2, \stg) \rightarrow_{R}^{n} (\texttt{SKIP}, \stg^{'})$. We have to show (1) $I(\stg^{'})$ and (2) $\texttt{step-guaranteed}( R, G_1 \cup G_2, t_1 || t_2, \stg)$. 

Suppose that $t_1$ commits before $t_2$ in the execution sequence. Consider the sequence upto (and including) the commit step of $t_1$, i.e. $(t_1 || t_2, \stg) \rightarrow_{R}^{n_1} (t_{1}^{'} || t_{2}^{'}, \stg_1) \rightarrow (t_{2}^{'}, \stg_{1}^{'})$. In this sequence, all steps apart from the steps taken by $t_1$ belong to $R \cup ID$, since any step taken by $t_2$ cannot change the global database state. Hence, there exists the sequence $(t_1, \stg) \rightarrow_{R}^{n_{1}^{'}} (t_{1}^{'}, \stg_1) \rightarrow (\texttt{SKIP}, \stg_{1}^{'})$ (the steps taken by $t_2$ can be removed). Since $R \subseteq R \cup G_2$, by H1 and Lemma 3.13, $I(\stg_{1}^{'})$ and $G_1(\stg_1, \stg_{1}^{'})$. Now, consider the entire sequence from the perspective of $t_2$. All steps taken by $t_1$ except the commit step do not change the global database state, and the change during the commit step belongs to $G_1$. Hence, all steps in the sequence apart from the steps taken by $t_2$ belong to $R \cup G_1 \cup ID$. Hence, there exists a sequence $(t_2, \stg) \rightarrow_{R \cup G_1}^{n^{'}} (\texttt{SKIP}, \stg^{'})$. By H2, $I(\stg^{'})$. 

Finally, the commit step of $t_1$ belongs to $G_1$, while the commit step of $t_2$ belongs to $G_2$, and every other step of either transaction does not change the global database state. Hence,  $\texttt{step-guaranteed}(R, G_1 \cup G_2, t1 || t2, \stg)$. The proof for the case where $t_2$ commits before $t_1$ would be similar.
\end{proof}

\begin{theorem}
RG-Conseq is sound
\end{theorem}

\begin{proof}
\begin{mathpar}
\begin{array}{lc}
 \R \vdash \hoare{P}{t}{Q} & H1\\
  P' \Rightarrow P & H2\\
  Q \Rightarrow Q' & H3\\
\end{array}
\end{mathpar}
Given $t, \stg$ such that $(t, \stg) \rightarrow_{\R}^{m} (t_2, \stg^{'}) \rightarrow (t^{'}, \stg^{'})$, $P'(t.\stl, \stg)$ and $t^{'}.c = \texttt{SKIP}$, we have to show that $Q'(t^{'}.\stl, \stg^{'})$. By H2, $P(t.\stl, \stg)$. Then, expanding the definition in H1, we get $Q(t^{'}.\stl, \stg')$. By H3, $Q'(t^{'}.\stl, \stg^{'})$.
\end{proof}

\begin{theorem}
RG-Conseq2 is sound
\end{theorem}

\begin{proof}
\begin{mathpar}
\begin{array}{lc}
\rg{I,R}{\ctxn{i}{\I}{c}}{G,I} & H1\\
  \I' \Rightarrow \I \spc & H2\\
  R' \subseteq R \spc & H2\\
  \stable(R',\I') & H3\\
  G \subseteq G' & H4\\
  \forall \stg,\stg'.~I(\stg) \wedge G'(\stg,\stg') \Rightarrow I(\stg')& H5
  \end{array}
  \end{mathpar}
Let $t = \ctxn{i}{\I'}{c}$. Given $\stg$ such that $I(\stg)$ and reduction sequence $\pi = (t, \stg) \rightarrow_{R'}^{n} (\C{SKIP}, \stg')$, we have to show that $I(\stg')$ and $\texttt{step-guaranteed}(R', G', t, \stg)$. First, we will show that the above reduction sequence is valid even if the isolation level of $t$ is changed to $\mathbb{I}$. Assume that the transaction performs $m$ steps in $\pi$. We will use induction on $m$ to show that every step of the transaction is valid for isolation level $\mathbb{I}$. 

For the base case, the first step is always valid irrespective of any isolation level. For the inductive case, assume that all steps upto the $k$th step of the transaction in $t$ are valid with isolation level $\mathbb{I}$. Let the $(k+1)$th step of the transaction be $(t_1, \stg_1) \rightarrow (t_2, \stg_1)$. Then $\mathbb{I}'(t_1.\stl, t_1.\stg, \stg_1)$. By H2, $\mathbb{I}(t_1.\stl, t_1.\stg, \stg_1)$. Hence, the $k+1$th step is also valid for isolation level $\mathbb{I}$. This shows that the entire reduction sequence is valid even if the isolation level of $t$ is changed to $\mathbb{I}$. Let $t' = \ctxn{i}{\I'}{c}$. Since $R' \subseteq R$, it follows that the reduction sequence $\pi' = (t', \stg) \rightarrow_{R}^{n} (\C{SKIP}, \stg')$ comprising of the same steps as $\pi$ is valid. By H1, $I(\stg')$. Finally, by $\texttt{step-guaranteed}(R, G, t', \stg)$, all global database state changes caused by $t'$ in $\pi'$ are in $G$. But these are the same global database stage changes in $\pi$. Since $G \subseteq G'$, these state changes are also in $G'$.
\end{proof}


\begin{theorem}
If $\I_e = \I_{ss}$, then $\forall \stl, \stg, \stg'. \R_e(\stl, \stg, \stg') \Rightarrow \stg = \stg'$ 
\end{theorem}
\begin{proof}
First, we show that 
$$
  \forall \stl,\stg,\stg',\stg''.~
  \neg\I_e(\stl,\stg,\stg') \conj R(\stg',\stg'') \Rightarrow
  \neg\I_e(\stl,\stg,\stg'')
$$
Given $\stl, \stg, \stg'$ such that $\neg \I_e(\stl, \stg, \stg')$, $\stg \neq \stg'$. Since $R(\stg', \stg'')$ corresponds to the commit of a transaction, either the transaction is read-only, in which case $\stg' = \stg''$ and hence $\stg \neq \stg''$ which implies $\neg \I_e(\stl, \stg, \stg'')$, or the transaction modifies/inserts a record, in which case it will also add its own unique transaction id to the record, so that $\stg \neq \stg''$, which again implies the result.

Now, consider $\stl, \stg, \stg'$ such that $\R_e(\stl, \stg, \stg')$.By definition of $\R_e$, there exists $\stg_1$ such that $\I_e(\stl, \stg_1, \stg')$. Hence, $\stg_1 = \stg'$. Now, $\I_e(\stl, \stg_1, \stg)$, because otherwise, if $\neg \I_e(\stl, \stg_1, \stg)$, then by the earlier result, $\neg \I_e((\stl, \stg_1, \stg')$ which is a contradiction. Hence, $\stg_1 = \stg$. This implies that $\stg = \stg'$.
\end{proof}


\newpage
\begin{theorem}
  \label{thm:inference-sound-strong}
  For all $i$,$R$,$I$,$c$,$\Fx$,$s$, $\F$, if $\stable(\R,I)$, $\stable(\R, \Fx)$ and $\Fx
  \vdash c \elabsto \F$, then:\\\vspace*{-0.2cm}
  \begin{smathpar}
  \begin{array}{c}
    \R \vdash \hoare{\lambda(\stl,\stg).~\stl= s \cup \Fx(\stg) \conj
    I(\stg)}{c}{\lambda(\stl,\stg).~\stl = s \cup \Fx(\stg) \cup \F(\stg)}
  \end{array}
  \end{smathpar}
\end{theorem}
\begin{proof}
Hypothesis:
\begin{smathpar}
\begin{array}{lr}
  \stable(\R,I) & H1\\
  \Fx \vdash c \elabsto \F & H2\\
\end{array}
\end{smathpar}
Proof by induction on $H2$. 

We prove the statement separately for every type of $c$. The base cases correspond to the SQL statements INSERT, UPDATE and DELETE. We note that

$$
\stable(\R, \Fx) \Leftrightarrow \forall \stg, \stg'. \R(\Fx(\stg), \stg, \stg') \Rightarrow \Fx(\stg) = \Fx(\stg')
$$

\textbf{Case : INSERT}. We have to show that $\forall \R,I, \Fx,\F,s$ if $\stable(\R, I)$, $\stable(\R, \Fx)$ and \\
$\Fx \vdash INSERT\ x \elabsto \stabilize{\Fx[\F]}$, then
$$
\R \vdash \hoare{\lambda(\stl,\stg).~\stl=s \cup \Fx(\stg) \conj
  I(\stg)}{INSERT\ x}{\lambda(\stl,\stg).\stl = s \cup \Fx(\stg) \cup \stabilize{\Fx[\F]}(\stg)}
$$
We will prove the premises of the \rulelabel{RG-Insert} rule. Here, $P \Leftrightarrow \lambda(\stl, \stg).~ \stl=s \cup\Fx(\stg) \conj I(\stg)$. By $\stable(\R, I)$ and $\stable(\R, \Fx)$, we have $\stable(\R, P)$. Note that by the definition of $\F$, we have $\stable(\R, \Fx[F])$ and hence, $\stabilize{\Fx[\F]} = \F$. 
$Q \Leftrightarrow \lambda(\stl, \stg). \stl = s \cup \Fx(\stg) \cup \F(\stg)$. Given $\stl, \stg, i$ such that $P(\stl, \stg)$ and $\stl' = \stl \cup \{x \with \delf=\mathit{false};\, \txnf = i\}$, it follows from definition of $\F$ that $Q(\stl', \stg)$. Thus, all premises of \rulelabel{RG-Insert} are satisfied.

\textbf{Case : UPDATE}. We have to show that $\forall \R,I, \Fx, s, \F$ if $\stable(\R, I)$, $\stable(\R, \Fx)$ and\\
 $UPDATE\ \lambda x.e_1\ \lambda x.e_2 \elabsto \stabilize{\Fx[\F]}$, then
$$
\R \vdash \hoare{\lambda(\stl,\stg).~\stl=s \cup \Fx(\stg) \conj
  I(\stg)}{UPDATE\ \lambda x.e_1\ \lambda x.e_2}{\lambda(\stl,\stg).\stl = s \cup \Fx(\stg) \cup \stabilize{\Fx[\F]}(\stg)}
$$
We will prove the premises of the \rulelabel{RG-Update} rule. Here, $P \Leftrightarrow \lambda(\stl, \stg).~ \stl=s \cup \Fx(\stg) \conj I(\stg)$ and $Q \Leftrightarrow \lambda(\stl, \stg).~ \stl = s \cup \Fx(\stg) \cup \stabilize{\Fx[\F]}(\stl, \stg)$. By $\stable(\R, I)$ and $\stable(\R, \Fx)$, we have $\stable(\R, P)$. We can have either $\stable(\R, \Fx[F])$ or $\neg \stable(\R, \Fx[F])$. In either case, we will show that all premises of \rulelabel{RG-Update} are satisfied.

Suppose $\stable(\R, \Fx[F])$. Then $\stabilize{\Fx[\F]} = \F$. Then, given $\stl, \stg$ such that $P(\stl, \stg)$ and $\stl' = \stl \cup \{r' \,|\, \exists(r\in\Delta).~ [r/x]e_2 =\C{true} \conj r'=[r/x]e_1 \with \idf=r.\idf;\,\delf = y.\delf;\,\txnf=i\}$, it follows from definition of $\F$ that $Q(\stl', \stg)$. Suppose $\neg \stable(\R, \Fx[F])$. Then, $\stabilize{\Fx[\F]} = \lambda \stg.\existsl(\stg',I(\stg'),\F(\stg'))$. Also, since $P(\stl, \stg)$, we have $I(\stg)$. Hence, $Q(\stl', \stg)$, since $\stg$ provides the existential $\stg'$, and $\Fx(\stg) \cup \F(\stg)$ is $\stl'$.  

\textbf{Case: DELETE}. We have to show that $\forall \R,I, \Fx, s, \F$ if $\stable(\R, I)$, $\stable(\R, \Fx)$ and \\
$DELETE\ \lambda x.e \elabsto \stabilize{\Fx[\F]}$, then
$$
\R \vdash \hoare{\lambda(\stl,\stg).~\stl=s \cup \Fx(\stg) \conj
  I(\stg)}{DELETE\ \lambda x.e}{\lambda(\stl,\stg).\stl = s \cup \Fx(\stg) \cup \stabilize{\Fx[\F]}(\stg)}
$$
We will prove the premises of the \rulelabel{RG-Delete} rule. Here, $P \Leftrightarrow \lambda(\stl, \stg).~ \stl=s \cup \Fx(\stg) \conj I(\stg)$ and $Q \Leftrightarrow \lambda(\stl, \stg).~ \stl = s \cup \Fx(\stg) \cup \stabilize{\Fx[\F]}(\stg)$. By $\stable(\R, I)$ and $\stable(\R, \Fx)$, we have $\stable(\R, P)$. We can have either $\stable(\R, \Fx[\F])$ or $\neg \stable(\R, \Fx[\F])$. In either case, we will show that all premises of \rulelabel{RG-Delete} are satisfied.

Suppose $\stable(\R, \Fx[\F])$. Then $\stabilize{\Fx[\F]} = \F$. Then, given $\stl, \stg$ such that $P(\stl, \stg)$ and $\stl' = \stl \cup \{r' \,|\, \exists(r\in\Delta).~ [r/x]e=\C{true} \conj r'=\{\bar{f}=r.\bar{f}; \idf=r.\idf;\delf=\C{true}\}\}$, it follows from definition of $\F$ that $Q(\stl', \stg)$. Suppose $\neg \stable(\R, \Fx[\F])$. Then, $\stabilize{\Fx[\F]} = \lambda \stg.\existsl(\stg',I(\stg'),\F(\stg'))$.Also, since $P(\stl, \stg)$, we have $I(\stg)$. Hence, $Q(\stl', \stg)$, since $\stg$ provides the existential $\stg'$, and $\Fx(\stg) \cup F(\stg)$ is $\stl'$.

\textbf{Case: SELECT}. Given $\R, I, s \Fx$ such that $\stable(\R, I)$ and $\stable(\R, \Fx)$, $\Fx \vdash c \elabsto \F$, $G = \lambda r.~ \itel{[r/x]e}{\{r\}}{\emptyset}$, and $F' =\stabilize{\Fx[\lambda \stg.~ (\stg \bind G)]}$, we have to show that  
 $$
\R \vdash \hoare{\lambda(\stl,\stg).~\stl=s \cup \Fx(\stg) \conj
  I(\stg)}{\lete{y}{\selecte{\lambda x.e}}{c}}{\lambda(\stl,\stg).\stl =s \cup  \Fx(\stg) \cup [F'(\stg)/y]F(\stg)}
$$

We will prove all the premises of \rulelabel{RG-Select}. Here, $P \Leftrightarrow \lambda(\stl,\stg).~\stl=s \cup \Fx(\stg) \conj I(\stg)$, while $Q \Leftrightarrow \lambda(\stl,\stg).\stl = s \cup \Fx(\stg) \cup [F'(\stg)/y]F(\stg)$. By $\stable(\R, I)$ and $\stable(\R, \Fx)$, we have $\stable(\R, P)$. By inductive hypothesis and $\Fx \vdash c \elabsto \F$ we have
$$
\R \vdash \hoare{\lambda(\stl,\stg).~\stl=s \cup \Fx(\stg) \conj
  I(\stg)}{c}{\lambda(\stl,\stg).\stl = s \cup \Fx(\stg) \cup F(\stg)}
$$
Let 
$$
P'(\stl, \stg) \Leftrightarrow P(\stl, \stg) \wedge y=\{r \in \stg | [r/x]e\}
$$

Given $\stl, \stg$, $P'$ just binds $y$ to a set of records which depend on $\stg$. We now have the following from the inductive hypothesis:
$$
\R \vdash \hoare {P'} {c} {\lambda(\stl, \stg). \stl =s \cup  \Fx(\stg) \cup [F'(\stg)/y]F(\stg)}
$$

The reason is that $y$ occurs free in $c$ and by the inductive hypothesis, any binding of $y$ can be used. Note that if $P'(\stl, \stg)$, then $y=(\stg \bind G)$. Suppose $\stable(\R,\Fx[ \lambda \stg.(\stg \bind G)])$. Then given $P'(\stl, \stg_{1})$, we have $y = F'(\stg_1)$. By $\stable(\R,\Fx[ \lambda \stg.(\stg \bind G)])$, we have
$$
\lambda(\stl, \stg).\stl = s \cup \Fx(\stg) \cup [F'(\stg_1)/y]F(\stg) = \lambda(\stl, \stg). \stl = s \cup \Fx(\stg) \cup [F'(\stg)/y]F(s, \stg_1)
$$ 
If $\neg \stable(\R, \Fx[ \lambda \stg.(\stg \bind G)])$, then 
$$\stabilize{\Fx[ \lambda \stg.(\stg \bind G)]} = \lambda \stg .\existsl(\stg', I, \stg' \bind G))
$$ 
Then, given $P'(\stl, \stg_1)$, $I(\stg_1)$ and hence $\stg_1$ gives the existential $\stg'$. 

\textbf{Case : IF-THEN-ELSE}. Given $\R, I, s, \Fx$ such that $\stable(\R, I)$, $\stable(\R, \Fx)$, $\Fx \vdash c_1 \elabsto \F_1$, $\Fx \vdash c_2 \elabsto \F_2$ , we have to show that  
 $$
\R \vdash \hoare{\lambda(\stl,\stg).~\stl=s \cup \Fx(\stg) \conj
  I(\stg)}{\ite{e}{c_1}{c_2}}{\lambda(\stl,\stg).\stl = s \cup \Fx(\stg) \cup (\itel{e}{\F_1(\stg)}{\F_2(\stg)})}
$$

We will prove all the premises of \rulelabel{RG-If}. Here, $P \Leftrightarrow \lambda(\stl,\stg).~\stl=s \cup \Fx(\stg) \conj I(\stg)$, while $Q \Leftrightarrow \lambda(\stl,\stg).\stl = s \cup \Fx(\stg) \cup (\itel{e}{\F_1(\stg)}{\F_2(\stg)})$. By the inductive hypothesis and $\Fx \vdash c_1 \elabsto \F_1$, we know that
$$
\R \vdash \hoare{\lambda(\stl, \stg).~\stl=s \cup \Fx(\stg) \conj I(\stg)}{c_1}{\lambda(\stl, \stg). \stl =s \cup  \Fx(\stg) \cup \F_1(\stg)}
$$

The post-condition in the above statement can also be written as $Q \wedge e$. Since $e$ does not access the global or local database, the above statement can be written as $\R \vdash \hoare{P \wedge e}{c_1}{Q \wedge e}$. Similarly, $\R \vdash \hoare{P \wedge \neg e} {c_2} {Q \wedge \neg e}$. By $\stable(\R, I)$ and $\stable(\R, \Fx)$, we have $\stable(\R, P)$. Thus, all the premises of \rulelabel{RG-If} are satisfied.

\textbf{Case : SEQ}. Given $\R, I, \Fx$ such that $\stable(\R, I)$, $\stable(\R, \Fx)$, $\Fx \vdash c_1 \elabsto F_1$, $\Fx \cup F_1 \vdash c_2 \elabsto F_2$ , we have to show that  
 $$
\R \vdash \hoare{\lambda(\stl,\stg).~\stl=s \cup \Fx(\stg) \conj
  I(\stg)}{c_1;c_2}{\lambda(\stl,\stg).\stl =s \cup  \Fx(\stg) \cup F_1(\stg) \cup F_2(\stg)}
$$

We will prove all the premises of \rulelabel{RG-Seq}. Here, $P \Leftrightarrow \lambda(\stl,\stg).~\stl=s \cup \Fx(\stg) \conj I(\stg)$, while $Q \Leftrightarrow \lambda(\stl,\stg).\stl = s \cup \Fx(\stg) \cup F_1(\stg) \cup F_2(\stg)$. Let $Q' \Leftrightarrow \lambda(\stl, \stg). \stl =s \cup  \Fx(\stg) \cup F_1(\stg)$. Then, by the inductive hypothesis and $\Fx \vdash c_1 \elabsto \F_1$, we have $\R \vdash \hoare{P}{c1}{Q'}$. Further, at this point, since the transaction has not committed, all the changes in the global database state must be due to $\R$. Since $\stable(\R, I)$, we have $\R \vdash \hoare{P}{c1}{\lambda(\stl, \stg).Q'(\stl,\stg) \conj I(\stg)}. $Further, by the inductive hypothesis and $\Fx \cup F_1 \vdash c_2 \elabsto \F_2$, we have $\R \vdash \hoare{\lambda(\stl, \stg).Q'(\stl,\stg) \conj I(\stg)}{c_2}{Q}$. Finally, since the stabilization operator ($\stabilize{}$) is always applied on $F_1$ and $\stable(\R, \Fx)$, we have $\stable(\R, \lambda(\stl, \stg).Q'(\stl,\stg) \conj I(\stg))$. Thus, all premises of \rulelabel{RG-Seq} are satisfied.

\textbf{Case : FOREACH}. Given $\R, I, \Fx,s$ such that $\stable(\R, I)$, $\stable(\R, \Fx)$, $\Fx \vdash c \elabsto F$, we have to show that
$$
\R \vdash \hoare{\lambda(\stl,\stg).~\stl=s \cup \Fx(\stg) \conj
  I(\stg)}{\foreache{x}{\lambda y.\lambda z.~c}}{\lambda(\stl,\stg).\stl = s \cup \Fx(\stg) \cup x\bind(\lambda z.~F(\stg)}
$$

We will prove all the premises of \rulelabel{RG-ForEach} using the loop invariant $\psi(\stl, \stg) \Leftrightarrow \stl = s \cup \Fx(\stg) \cup y \bind (\lambda z.~\F(\stg))$. Here $P \Leftrightarrow \lambda(\stl,\stg).~\stl=s \cup \Fx(\stg) \conj I(\stg)$, while $Q \Leftrightarrow \lambda(\stl,\stg).\stl =s \cup  \Fx(\stg) \cup x\bind(\lambda z.~\F(\stg)$. Since $[\phi/y]\psi(\stl, \stg) \Leftrightarrow \stl = s \cup \Fx(\stg)$, $P \rightarrow [\phi/y]\psi$. By the inductive hypothesis and $\Fx \vdash c \elabsto \F$, we have
$$
\R \vdash \hoare{\lambda(\stl, \stg).~ \stl =s \cup y \bind (\lambda z.~\F(\stg)) \cup \Fx(\stg) \conj I(\stg)} {c} {\lambda(\stl, \stg). \stl = s  \cup y \bind (\lambda z.~\F(\stg)) \cup F(\stg)}
$$

Binding $z$ (which is free in $c$) to a record in $x$ (i.e. $z \in x$) in the pre-condition, the post condition in the above statement implies $\stl = s \cup (y \cup \{z\}) \bind (\lambda z.~\F(\stg))$, which is nothing but $[y \cup \{z\}/y]\psi(\stl, \stg)$. Hence, $\psi$ is a loop invariant. Finally, $[x/y]\psi \rightarrow Q$. 

From $\stable(\R, I)$ and $\stable(\R, \Fx)$, we have $\stable(\R, P)$. Since $F$ has been stabilized using the $\stabilize{}$ function, and $\psi$ is an assertion on the union of multiple applications of $F$, it follows that $\stable(\R, \psi)$. Using the same reasoning, $\stable(\R, Q)$. Thus, all the premises of \rulelabel{RG-Foreach} are satisfied.


\end{proof}

\begin{theorem}
\label{thm:inference-sound}
 For all $i$,$R$,$I$,$c$,$\Fx$, $\F$, if $\stable(\R,I)$, $\stable(\R, \Fx)$ and $\Fx
  \vdash c \elabsto \F$, then:\\\vspace*{-0.2cm}
  \begin{smathpar}
  \begin{array}{c}
    \R \vdash \hoare{\lambda(\stl,\stg).~\stl= \Fx(\stg) \conj
    I(\stg)}{c}{\lambda(\stl,\stg).~\stl = \Fx(\stg) \cup \F(\stg)}
  \end{array}
  \end{smathpar}
\end{theorem}
\begin{proof}
  Follows from the stronger version of this theorem
  (Theorem~\ref{thm:inference-sound-strong}) by substituting
  $\emptyset$ for $s$.
\end{proof}


%

\end{document}